\documentclass[12pt]{article}
\usepackage{fullpage,amsmath,amssymb,mathtools,natbib,hyperref,colortbl}
\usepackage{titling,enumitem}
\usepackage[skip=0pt]{caption} 

\usepackage[font=small]{floatrow}
\newfloatcommand{capbtabbox}{table}[][\FBwidth]
\usepackage{datetime,bigints}
\usdate
\usepackage{multirow}
\usepackage{float}
\usepackage{algorithm,bbm}
\usepackage[noend]{algpseudocode}

\usepackage[section]{placeins}
\usepackage[titletoc,toc,title]{appendix}
\usepackage{tikz}

\usepackage{xr,caption}
\newcommand{\mc}[2]{\multicolumn{#1}{c}{#2}}
\definecolor{Gray}{gray}{0.85}
\definecolor{LightCyan}{rgb}{0.88,1,1}
\newcolumntype{a}{>{\columncolor{Gray}}c}
\newcolumntype{b}{>{\columncolor{white}}c}

\def\mymathcal{\mathbbm}

\newcommand{\param}{\beta}
\newcommand{\tail}{\xi}
\newcommand{\paramspace}{\mathcal{B}}

\usepackage{amsthm}
{
      \theoremstyle{plain}
      \newtheorem{definition}{Definition}
      \newtheorem{theorem}{Theorem}
      \newtheorem{example}{Example}
      \newtheorem{proposition}{Proposition}
      \newtheorem{lemma}{Lemma}
      \newtheorem{corollary}{Corollary}
      \newtheorem{assumption}{Assumption}
      \newtheorem{remark}{Remark}
  }

  \newenvironment{assumptionbis}[1]
    {\renewcommand{\theassumption}{\ref{#1}$'$}%
     \addtocounter{assumption}{-1}%
     \begin{assumption}}
    {\end{assumption}}

\renewcommand{\arraystretch}{1.5}

\def\lQ{\scalebox{-1}[1]{''}}

\makeatletter
\renewenvironment{abstract}{%
    \if@twocolumn
      \section*{\abstractname}%
    \else 
      \begin{center}%
        {\bfseries \normalsize\abstractname\vspace{\z@}}
      \end{center} \vspace{-0.5cm}%
      \quotation
    \fi}
    {\if@twocolumn\else\endquotation\fi}
\makeatother

\begin{document}

  \title{A Sieve-SMM Estimator for Dynamic Models}
  \author{\Large Jean-Jacques Forneron\thanks{Department of Economics, Boston University, 270 Bay State Road, Boston, MA 02215.\newline Email: \href{mailto:jjmf@bu.edu}{jjmf@bu.edu}.   This paper is based on the third chapter of my doctoral dissertation at Columbia University. 
   I am indebted to my advisor Serena Ng for her continuous guidance and support. I would like to thank the co-editor and three anonymous referees for insightful and helpful comments. I also greatly benefited from comments and discussions with Jushan Bai, Tim Christensen, Benjamin Connault, Gregory Cox, Iv\'an Fern\'andez-Val, Ron Gallant, Eric Gautier, Hiro Kaido, Dennis Kristensen, Sokbae Lee, Kim Long, Nour Meddahi, Jos\'e Luis Montiel Olea, Zhongjun Qu, Christoph Rothe, Bernard Salani\'e and the participants of the Columbia Econometrics Colloquium as well as the participants of the econometrics seminar at Boston University, Chicago Booth, UC Berkeley, Bocconi, Georgetown, UPenn and participants at several conferences. All errors are my own.}
  } \date{\today}
  \maketitle

\begin{abstract}  
  \vspace{-0.2cm}
  This paper proposes a Sieve Simulated Method of Moments (Sieve-SMM) estimator for the parameters and the distribution of the shocks in nonlinear dynamic models where the likelihood and the moments are not tractable. An important concern with SMM, which matches sample with simulated moments, is that a parametric distribution is required. However, economic quantities that depend on this distribution, such as welfare and asset-prices, can be sensitive to misspecification. The Sieve-SMM estimator addresses this issue by flexibly approximating the distribution of the shocks with a Gaussian and tails mixture sieve. The asymptotic framework provides consistency, rate of convergence and asymptotic normality results, extending existing results to a new framework with more general dynamics and latent variables. 
  An application to asset pricing in a production economy shows a large decline in the estimates of relative risk-aversion, highlighting the empirical relevance of misspecification bias.
\end{abstract}

\bigskip
\noindent JEL Classification: C14, C15, C32, C33.\newline
\noindent Keywords:  Simulated Method of Moments, Mixture Sieve, Asset Pricing.

\bibliographystyle{ecta}
\baselineskip=18.0pt
\thispagestyle{empty}
\setcounter{page}{0}

\section{Introduction}
Complex nonlinear dynamic models with an intractable likelihood or moments are increasingly common in economics. A popular approach to estimating these models is to match informative sample moments with simulated moments from a fully parameterized model using SMM. However, economic models are rarely fully parametric since theory usually provides little guidance on the distribution of the shocks. The Gaussian distribution is often used in applications but in practice, different choices of distribution may have different economic implications; this is illustrated below. Yet to address this issue, results on semiparametric simulation-based estimation are few. 

This paper proposes a Sieve Simulated Method of Moments (Sieve-SMM) estimator for both the structural parameters and the distribution of the shocks and explains how to implement it. The dynamic models considered in this paper have the form:
\begin{align}
&y_t = g_{obs}(y_{t-1},x_t,\theta,f,u_t) \label{eq:observed}\\
&u_t = g_{latent}(u_{t-1},\theta,f,e_t), \quad e_t \sim f. \label{eq:unobserved}
\end{align}
The observed outcome variable is $y_t$, $x_t$ are exogenous regressors and $u_t$ is a vector of unobserved latent variables. The unknown parameters include $\theta$, a finite dimensional vector, and the distribution $f$ of the shocks $e_t$. The functions $g_{obs},g_{latent}$ are known, or can be computed numerically, up to $\theta$ and $f$. The Sieve-SMM estimator extends the existing Sieve-GMM literature to more general dynamics with latent variables and the literature on sieve simulation-based estimation of some static models. 

The estimator in this paper has two main building blocks: the first one is a sample moment function, such as the empirical characteristic function (CF) or the empirical cumulative distribution function (CDF); infinite dimensional moments are needed to identify the infinite dimensional parameters. As in the finite dimensional case, the estimator simply matches the sample moment function with the simulated moment function. To handle this continuum of moment conditions, this paper adopts the objective function of \citet{Carrasco2000,Carrasco2007} in a semi-nonparametric setting.

The second building block is to nonparametrically approximate the distribution of the shocks using the method of sieves, as numerical optimization over an infinite dimension space is generally not feasible. Typical sieve bases include polynomials and splines which approximate smooth regression functions. Mixtures are particularly attractive to approximate densities for three reasons: they are computationally cheap to simulate from, they are known to have good approximation properties for smooth densities, and draws from the mixture sieve are shown in this paper to satisfy the $L^2$-smoothness regularity conditions required for the asymptotic results. Restrictions on the number of mixture components, the tails and the smoothness of the true density ensure that the bias is small relative to the variance so that valid inferences can be made in large samples. To handle potentially fat tails, this paper also introduces a Gaussian and tails mixture.
The tail densities in the mixture are constructed to be easy to simulate from and also satisfy $L^2$-smoothness properties. The algorithm gives an overview of steps required to compute the estimator, more details are given in Section \ref{sec:the_estimator}.

\begin{algorithm}[h]
  \caption*{\quad \textsc{\textbf{Algorithm:}} Computing the Sieve-SMM Estimator}\label{algo:Sieve-SMM}
  \begin{algorithmic}

  \State Set a sieve dimension $k(n) \geq 1$ and a number of lags $L \geq 1$.
  \State Compute $\hat \psi_n$, the Characteristic Function (CF) of $(y_{t},\dots,y_{t-L},x_t,\dots,x_{t-L})$.
  \For {$s = 1,\dots,S$}
  \State Simulate the shocks $e_t^s$ from a mixture distribution $f$,
  \State Simulate artificial samples $(y_t^s)_{t=1,\dots,n}$ at $(\theta,f)$ using the innovations $(e_t^s)_{t=1,\dots,n}$,
  \State Compute $\hat \psi^s_n(\theta,f)$, the CF of the simulated data $(y_t^s,\dots,y_{t-L}^s,x_t,\dots,x_{t-L})$,
  \EndFor
  \State Compute the average simulated CF $\hat \psi_n^S (\theta,f)= \frac{1}{S}\sum_{s=1}^S \hat \psi_n^s (\theta,f)$.
  \State Compute the objective function $\hat Q_n^S(\theta,f)$: a distance between $\hat\psi_n$ and $\hat\psi_n^S$,
  \State Find the parameters $(\hat \theta_n,\hat f_n)$ that minimize the distance $\hat Q_n^S$.
  \end{algorithmic}
\end{algorithm}

To give intuition why the mixture sieve can be useful in economic analyses, consider an endowment economy where consumption growth $\Delta c_t=\log(C_t/C_{t-1})$ follows a simple AR(1) process  with mean zero innovations $e_t \sim f$. The parameters consist of the AR(1) coefficient and $f$. This simple specification provides a lot of flexibility to model asset prices; for a CRRA utility with risk aversion $\gamma$, the risk-free rate is:
\[ r_t = \text{const} - \log \left(  \mathbb{E}_t[\exp(-\gamma e_t)] \right).\]
The first $\textit{const}$ term involves the AR(1) parameters, $\Delta c_{t-1}$, time-preference $\delta$, and risk-aversion $\gamma$. The last term is the log of the moment generating function (MGF) for $f$ evalutated at $\gamma$. Heavy tails imply an infinite MGF so the paper will focus on Gaussian mixtures for which the MGF is finite and asset prices are well defined. Different distributions $f$ are associated with different MGF which, for a given $\gamma$, leads to different values of the risk-free rate above. Here, flexibly approximating $f$ would allow to better match features of $\Delta c_t$ but also of the risk-free rate. In comparison, with Gaussian shocks, the risk-free rate becomes $r_t =\text{const}-\gamma^2\sigma_e^2/2$, which requires a relatively large value $\gamma$ to match the data \citep{weil1989}. The empirical applications look at the asset pricing implications of using a flexible mixture specification for $f$; first in a very simple endowment model and then in the production economy of \citet{VanBinsbergen2012}. Both applications illustrate the discussion above. For the second application in particular, using quarterly US data between 1961 and 2019, baseline Gaussian estimates confirm their conclusion that macro and financial data are difficult to match. Here, with mixture instead of Gaussian shocks, estimates of risk aversion decline from $35$ to $10$, with a 95\% confidence interval of $[5,16]$.

As usual in the sieve literature, this paper provides a consistency result and derives the rate of convergence of the structural and infinite dimensional parameters, as well as asymptotic normality results for finite dimensional functionals of these parameters. While the results apply to both static and dynamic models alike, two important differences arise in dynamic models compared to the existing literature on sieve estimation: proving uniform convergence of the objective function and  controlling the dynamic accumulation of the nonparametric approximation bias.

The first challenge is to establish the rate of convergence of the objective function for dynamic models. To allow for the general dynamics (\ref{eq:observed})-(\ref{eq:unobserved}) with latent variables, this paper adapts results from \citet{Andrews1994} and \citet{BenHariz2005} to construct an inequality for uniformly bounded empirical processes which may be of independent interest. 
It holds under the geometric ergodicity condition found in \citet{Duffie1993}. 
The boundedness condition is satisfied by the CF and the CDF for instance. Also, the inequality implies a larger variance than typically found in the literature for iid or strictly stationary data with limited dependence induced by the moments.

The second challenge is that in the model (\ref{eq:observed})-(\ref{eq:unobserved}) the nonparametric bias accumulates dynamically. At each time period the bias appears because draws are taken from a mixture approximation instead of the true $f_0$, this bias is also transmitted from one period to the next since $(y_t^s,u_t^s)$ depends on $(y_{t-1}^s,u_{t-1}^s)$. To ensure that this bias does not accumulate too much, a decay condition is imposed on the data generating process (DGP).  For an AR(1) process with coefficient $\rho$, this condition holds if $|\rho|<1$. The resulting bias is generally larger than in static models and usual sieve estimation problems. Together, the increased variance and bias imply a slower rate of convergence for the Sieve-SMM estimates. Hence, in order to achieve the rate of convergence required for asymptotic normality, the Sieve-SMM requires additional smoothness of the true density $f_0$. Bias accumulation seems to be generic to sieve estimation of dynamic models: if the computation of the moments or likelihood involves a filtering step then the bias accumulates inside the prediction error of the filtered values. 
Monte-Carlo simulations illustrate the finite sample properties of the estimator and the effect of dynamics on the bias and the variance properties of the estimator.

While the paper proposes to relax certain parametric assumptions in the estimation, the model can still be misspecified along other dimensions and remain unable to match certain features of the data. This constrasts with more parsimonious choices of moments in indirect inference which could be more robust to misspecification in certain dimensions. This is illustrated in the second application, Table \ref{tab:moms}, using a sieve improves the fit in some dimensions but not all. To interpret the parameters being estimated here under misspecification, note that the objective function $Q = \lim_{n \to \infty} \hat{Q}_n^S$ can be re-written using Fubini's Theorem and direct calculations -- under Gaussian weighting with mean $0$ and variance $\Sigma$ -- as $\int \exp(-\frac{1}{2}[\textbf{y}_t-\tilde{\textbf{y}_t}]^\prime \Sigma [\textbf{y}_t-\tilde{\textbf{y}_t}])[f(\textbf{y}_t)-f_0(\textbf{y}_t)][f(\tilde{\textbf{y}_t})-f_0(\tilde{\textbf{y}_t})]d\textbf{y}_td\tilde{\textbf{y}_t}$ where $f$ and $f_0$ are the distribution of the simulated data $\textbf{y}_t = (y_t,\dots,y_{t-L},x_t,\dots,x_{t-L})$ and the data, respectively. This implies that the parameters minimize this distance between the joint distributions $f$ and $f_0$ weighted by a Gaussian kernel with variance $\Sigma^{-1}$. Notice that it is well defined even if the two distributions have different supports. 
In light of this, extending the theory to cover estimation and inference under misspecification, as in \citet{AI20075}, and understanding the impact of using infinite rather than finite dimensional moments and objective function on the pseudo-true parameters are of interest for future research. 


\subsection*{Related Literature}

The Sieve-SMM estimator presented in this paper combines two literatures: sieve estimation and the Simulated Method of Moments (SMM). This section provides a non-exhaustive review of the existing methods and results.


A key aspect to simulation-based estimation is the choice of moments $\hat \psi_n$. The Simulated Method of Moments (SMM) estimator of \citet{McFadden1989} relies on unconditional moments, the Indirect Inference (IND) estimator of \citet{Gourieroux1993} uses auxliary parameters from a simpler, tractable model and the Efficient Method of Moments (EMM) of \citet{Gallant1996} uses the score of the auxiliary model.
To achieve parametric efficiency, a number of papers consider using nonparametric moments but assume the distribution $f$ is known.\footnote{See e.g. \citet{Gallant1996,Fermanian2004}.} To avoid dealing with nonparametric moments, \citet{Carrasco2007} use the ECF. This paper uses a similar approach in a semi-nonparametric setting. 

General asymptotic results are given by \citet{Pakes1989} for SMM with iid data and \citet{Duffie1993} for time-series. The models considered in this paper are generative, they are fully specified so that one can generate a full dataset by simulation. A related but different class of problems relies on simulations to compute moment conditions for non-linear IV estimation, but these models are not fully parametrically specified. They cannot be used to simulate artificial datasets without additional modelling assumptions. 

Few papers are concerned with sieve simulation-based estimation; \citet{Bierens2012} and \citet{Newey2001} consider specific static models. \citet{Blasques2011} considers generic semi-nonparametric indirect inference. The results rely on generic uniform convergence results which do not apply in the present setting because of the non-standard dependence. His assumptions imply $\sqrt{n}$-convergence of $(\hat\theta_n,\hat f_n)$ which is restrictive. 
\citet{Dridi2000} propose a partial encompassing principle where parameters of interest are consistently estimable even if nuisance parameters are inconsistent because of misspecification.

An alternative to using sieves is to model several moments of the distribution with a parametric distribution. \citet{Ruge-Murcia2016} uses the skew Normal and Generalized Extreme Value distributions to model skewness in an asset pricing model. \citet{Gospodinov2013} use the Generalized Lambda family to estimate a non-invertible Moving Average (MA) model. In applications where the full unknown distribution matters for outcomes, estimates may be sensitive to the choice of distribution. As discussed in the introduction, asset prices depend on the full distribution via the MGF. 

Another related literature is the sieve estimation of models defined by moment conditions. These models can be estimated using either Sieve-GMM, Sieve Empirical Likelihood or Sieve Minimum Distance \citep[see][for a review]{Chen2007}. Applications include nonparametric estimation of IV and quantile IV regressions, and the semi-nonparametric estimation of asset pricing models,\footnote{See e.g. \citet{Hansen1987,Chen2009,Chen2013,Christensen2017}. } for instance. Existing results cover the consistency and the rate of convergence of the estimator as well as asymptotic normality of functional of the parameters for both iid and dependent data. See e.g. \citet{Chen2012,Chen2015a} and \citet{Chen2015b} for recent results with iid data and dependent data.

In the empirical Sieve-GMM literature, an application closely related to the dynamics encountered in this paper appears in \citet{Chen2013}. They estimate an Euler equation with recursive preferences where the value function is approximated using sieves. 
\citet{norets2014a} consider semiparametric Gaussian mixture estimation of dynamic discrete choice models. More generally, there is a large literature on Bayesian nonparametric estimation using mixtures. For non-linear state-space models where the likelihood is often intractable, simulations are also used to compute the objective function. This is implemented with the particle filter. Bayesian inference starts with a prior on both the finite dimensional and the non-parametric components, a common choice of prior for mixtures is used in Section \ref{sec:MonteCarlo}. Monte-Carlo Markov-Chain methods are then used to sample from the posterior.\\
To summarize, this paper extends existing results on Sieve and SMM estimation to a framework with non-linear dynamics, latent variables and flexible semi-nonparametric estimation.

\subsection*{Notation}
The following notation and assumptions will be used throughout the paper: the parameter of interest is $\param=(\theta,f) \in \Theta \times \mathcal{F} = \paramspace$.  The finite dimensional parameter space $\Theta$ is compact and the infinite dimensional set of densities $\mathcal{F}$ is possibly non-compact. The sets of mixtures satisfy $\paramspace_k \subseteq \paramspace_{k+1} \subseteq \paramspace$, $k$ is dimension of the sieve set $\paramspace_{k}$. The dimension $k$ increases with the sample size: $k(n)\to \infty$ as $n\to \infty$. $\Pi_{k(n)}f$ is the mixture approximation of $f$. The vector of shocks $e \sim f$ has dimension $d_e \geq 1$. The total variation (TV) distance between two densities is $\|f_1-f_2\|_{TV} = 1/2 \int |f_1(e)-f_2(e)|de$ and the supremum (or sup) norm is $\|f_1-f_2\|_{\infty} = \sup_{e \in \mathbb{R}^{d_e}}|f_1(e)-f_2(e)|$. Let $\|\param_1-\param_2\|_{TV}=\|\theta_1-\theta_2\|+\|f_1-f_2\|_{TV}$ and $\|\param_1-\param_2\|_{\infty}=\|\theta_1-\theta_2\|+\|f_1-f_2\|_{\infty}$, where $\|\theta\|$ and $\|e\|$ correspond the Euclidian norm of $\theta$ and $e$ respectively. $\|\param_1\|_m$ is a norm on the mixture components: $\|\param_1\|_m=\|\theta\|+\|(\omega,\mu,\sigma)\|$ where $\|\cdot\|$ is the Euclidian norm and $(\omega,\mu,\sigma)$ are the mixture parameters. For a functional $\phi$, its pathwise, or G\^ateaux, derivative at $\param_1$ in the direction $\param_2$ is $\frac{d \phi(\param_1)}{d \param}[\param_2] = \frac{d \phi\left( \param_1+\varepsilon\param_2 \right)}{d\varepsilon}\Big|_{\varepsilon=0}$. For two sequences $a_n$ and $b_n$, $a_n \asymp b_n$ implies that there exists $0<c_1\leq c_2 < \infty$ such that $c_1 a_n \leq b_n \leq c_2 a_n$ for all $n \geq 1$.

\subsection*{Structure of the Paper}
The paper is organized as follows: Section \ref{sec:the_estimator} provides an overview of the Sieve-SMM estimator and its implementation. Section \ref{sec:AsymMixture} gives the main asymptotic results. Section \ref{sec:MonteCarlo} illustrates the finite sample properties of the estimator using Bayesian nonparametric estimation as a benchmark. Section \ref{sec:empirical} applies the methodolgy to asset pricing in a production economy. Section \ref{sec:conclusion} concludes. Appendices \ref{apx:prelim}, \ref{apx:asymptotic_main} consist of preliminary lemmas and the proofs for the main results. The Supplement provides several additional appendices. Appendices \ref{sec:proof_prelim}, \ref{apx:additional_results} and \ref{apx:proof_additional} consist of the proofs for the preliminary lemmas, intermediate results and their proofs. Appendix \ref{apx:add_appli} provides additional material for the empirical applications and additional results.


\section{A Sieve-SMM Estimator} \label{sec:the_estimator}
This section describes the estimator and its implementation, including practical aspects such as tuning parameters and optimization, using a simple illustrative AR(1) example:
\begin{align} y_t = \rho y_{t-1} + e_t,\quad e_t \overset{iid}{\sim} f, \label{eq:ar1}\end{align}
where $t=1,\dots,n$; $n$ is the sample size. The parameters of interest are $\theta = \rho$ and the distribution  $f$. The latter is approximated by a mixture of $k$ Gaussians densities:
\begin{align*} f_{\omega,\mu,\sigma}(\cdot) = \sum_{j=1}^k \frac{\omega_j}{\sigma_j} \phi\left( \frac{\cdot - \mu_j}{\sigma_j} \right), \end{align*}
where $\phi$ is the normal pdf. The weights $\omega_j$ are positive and sum to one. The location and scale parameters are also restricted as discussed below. The sieve dimension $k$ increases with $n$ to reduce the approximation bias as sampling uncertainty declines. 

Simulation-based estimation requires sampling from $f_{\omega,\mu,\sigma}$ and then generating data from (\ref{eq:ar1}). For a given value of the mixture coefficients $(\omega,\mu,\sigma)$, $S\geq 1$ samples of $n$ Gaussian mixtures are simulated as follows. First, let $\omega_0=0$, compute the cumulative $\overline{\omega}_j = \sum_{\ell = 0}^j \omega_{\ell}$, draw a uniform and a Gaussian random variable $u_t^s \overset{iid}{\sim} \mathcal{U}_{[0,1]}, Z_t^s \overset{iid}{\sim} \mathcal{N}(0,1)$, $t=1,\dots,n$; $s=1,\dots,S$ to generate $e_t^s = \sum_{j=1}^k \mathbbm{1}_{u_t^s \in [\overline{\omega}_{j-1},\overline{\omega}_{j}]}(\mu_j + \sigma_j Z_t^s)$. By construction $e_t^s \overset{iid}{\sim} f_{\omega,\mu,\sigma}$. The pair $(u_t^s,Z_t^s)$ is only drawn once so that the optimization problem is well behaved and stochastic equicontinuity conditions hold. Then, to simulate from (\ref{eq:ar1}), set $y_0^s = y_0$ fixed and compute recursively $y_{t}^s = \rho y_{t-1}^s + e_t^s$ for $t=1,\dots,n$ and $s=1,\dots,S$. In DSGE models, $y_0$ is typically set at the steady-state; another common choice is $y_0 = 0$.

Estimation then requires comparing the sample with the simulated data. In particular, identifying both $\rho$ and $f$ requires information about the persistence of $y_t$ and the marginal distribution of $y_t-\rho y_{t-1}$. In (\ref{eq:ar1}), all of this information is contained in the joint distribution of $\mathbf{y}_t = (y_t,\dots,y_{t-L})$ for any $L \geq 1$. Following \citet{Carrasco2007}, this joint distribution is summarized by the ECF of the sample and simulated data:
\begin{align*} \hat \psi_n(\tau) = \frac{1}{n} \sum_{t=1}^n e^{i\tau^\prime \mathbf{y}_t }, \quad \hat \psi_n^S(\tau,\theta,f) = \frac{1}{nS} \sum_{s=1}^S \sum_{t=1}^n e^{i\tau^\prime \mathbf{y}^s_t(\theta,f) }, \end{align*}
where $\tau \in \mathbb{R}^{\text{dim}(\mathbf{y})}$ and $i$ is the imaginary number such that $i^2=-1$. In the general setting (\ref{eq:observed})-(\ref{eq:unobserved}), the joint ECF of $(\mathbf{y}_t,\mathbf{x}_t)$ and $(\mathbf{y}_t^s,\mathbf{x}_t)$ will be used. Matching the two ECFs over $\tau \in \mathbb{R}^{\text{dim}(\mathbf{y})}$ implies a continuum of moment conditions $\mathbb{E}[\hat \psi_n(\tau)-\hat \psi^S_n(\tau,\theta,f)]=0,\forall \tau$. The objective function is computed as a weighted distance of the moment functions \citep{Carrasco2000,Carrasco2007}:
\begin{align}
  \hat Q_n^S(\theta,f) = \int_{\mathbb{R}^{L+1}} \Big| \hat \psi_n(\tau) - \hat \psi_n^S(\tau,\theta,f) \Big|^2 \pi(\tau)d\tau, \label{eq:objective}
\end{align}
where $\pi$ is a continuous density with full support. In practice, the Gaussian density is used and the integral is computed over a fine grid as discussed below.
The estimated parameter $\hat \param_n = (\hat \rho_n,\hat f_n)$ is an approximate minimizer of this weighted distance:
\begin{align}
 \hat Q_n^S(\hat \param_n) \leq \inf_{\param \in \paramspace_{k(n)}} \hat Q_n^S(\beta) + O_p(\hat\eta_n),  \label{eq:estimator}
\end{align}
where $\hat\eta_n \geq 0$, $\hat\eta_n = O_p(\eta_n)$, $\eta_n=o(1)$ corresponds to numerical optimization and integration errors, assumed negligible. The following provides the detailed steps to implement the estimation and suggestions for the tuning parameters. 

\begin{algorithm}[h]
  \caption*{\quad \textsc{\textbf{Algorithm:}} Computing the Sieve-SMM Estimator}\label{algo:Sieve-SMM2}
  \begin{algorithmic}

  \State Set $k(n) \geq 1$, $L \geq 1$ the sieve dimension and number of lags.
  \State Compute the CF: $\hat \psi_n(\tau) = \frac{1}{n} \sum_{t=1}^n e^{i \tau^\prime \mathbf{y}_t}$,  $\mathbf{y}_t = (y_{t},\dots,y_{t-L},x_t,\dots,x_{t-L})$.
  \For {$s = 1,\dots,S$}
  \State Simulate the $e_t^s \sim f_{\omega,\mu,\sigma}$: a $k(n)$ component mixture distribution indexed by $(\omega,\mu,\sigma)$.
  \State Simulate artificial samples $(y_1^s,\dots,y_n^s)$ at $(\theta,f_{\omega,\mu,\sigma})$ using the innovations $e_t^s$ and $x_t$.
  \State Compute $\hat \psi^s_n(\tau,\theta,f_{\omega,\mu,\sigma})$, the CF of the simulated data $(y_t^s,\dots,y_{t-L}^s,x_t,\dots,x_{t-L})$.
  \EndFor
  \State Compute the average simulated CF $\hat \psi_n^S (\tau,\theta,f_{\omega,\mu,\sigma})= \frac{1}{S}\sum_{s=1}^S \hat \psi_n^s (\tau,\theta,f_{\omega,\mu,\sigma})$.
  \State Compute the objective function $\hat Q_n^S(\theta,f_{\omega,\mu,\sigma}) = \int |\hat\psi_n(\tau) - \hat\psi_n^S(\tau,\theta,f_{\omega,\mu,\sigma})|^2 \pi(\tau)d\tau$.
  \State Find the parameters $(\hat \theta_n,\hat f_{n})$ that minimize the objective $\hat Q_n^S$.
  \end{algorithmic}
\end{algorithm}

\paragraph{Inputs for the mixture:} The sieve dimension $k(n)$ and bounds on location/scale parameters $(\mu_j,\sigma_j)_{j=1,\dots,k}$ should be chosen jointly. Bounds complying with theoretical requirements are $|\mu_j-\mu| \leq \sigma C_\mu \log(k)$ and $\sigma_j \geq \sigma C_\sigma \log(k+1)/(k+1)$ where $\mu = \sum_{j=1}^k \omega_j \mu_j$ and $\sigma^2 = \sum_{j=1}^k \omega_j (\mu_j^2+\sigma_j^2)-\mu^2$ are the expected value and variance of the mixture. Both bounds adapt to the density's mean/variance and are easily handled by the preferred optimizer below. $C_\mu$ should be large enough to fit the tails of $f$, $C_\mu = 7$ performs well in the simulations and the applications. For a given $k$, $C_\sigma$ plays a similar role to a bandwidth in kernel density estimation. In particular, the local measure of ill-posedness increases with the inverse of the lower bound $\underline{\sigma}_{k(n)}$ on $\sigma_j$ which implies a slower rate of convergence for the estimator. The simulations and the application use $C_\sigma = 1.8$ and vary $k=2,\dots,5$.

\paragraph{Inputs for $\hat Q_n^S$:}  Three inputs are required: $L,\pi$ and an integration grid. If $y_t$ is markovian of order $\ell$, the first $\ell$ lags contain all of the information on the dependence; a natural choice is then $L=\ell$ \citep{Carrasco2007}. For non-markovian $y_t$, finding $L$ such that the first $L$ (non)linear autocorrelations capture the dependence is necessary. For instance, $L \geq \ell$ for $\text{MA}(\ell)$ models and $L \geq 1$ for a canonical stochastic volatility model with AR(1) volatility. With respect to $\pi$, using the ECF of $\Sigma_n^{-1/2}(\mathbf{y}_t - \bar{\mathbf{y}}_n)$ and $\Sigma_n^{-1/2}(\mathbf{y}^s_t - \bar{\mathbf{y}}_n)$, where $\bar{\mathbf{y}}_n,\Sigma_n$ are the sample mean and variance of $\mathbf{y}_t$, with Gaussian density weights corresponds (by a change of variable argument) to a choice of $\pi$ which has appealing features. Indeed, expanding the difference in ECF around $\tau=0$:
\[ \hat\psi_n(\tau) - \hat\psi_n^S(\tau) = i \tau^\prime \Sigma_n^{-1/2}(\bar{\mathbf{y}}_n^S - \mathbf{y}_n) + \frac{1}{2 n S} \sum_{s=1}^S \sum_{t=1}^n  (\mathbf{y}_t^s - \mathbf{y}_n)^\prime \Sigma_n^{-1/2} \tau \tau^\prime \Sigma_n^{-1/2} (\mathbf{y}_t^s - \mathbf{y}_n) - \frac{1}{2} \tau \tau^\prime + \dots\]
reveals that a density which puts more weight around $0$ gives more weight to lower-order moments. Akin to a GMM weighting scheme, the researcher can put more (or less) weight on lower-order moments (means, co-variances) vs. higher-order moments (skewness, kurtosis) in the estimation by choosing a smaller (or larger) variance for the Gaussian weights. To compute the integral in (\ref{eq:objective}) a finite grid of scrambled Sobol points is used. These can be more accurate than a Monte-Carlo approximation even for relatively large dimensions, unlike quadrature rules. The grid will be assumed to be large enough so that the integration error is negligible. In the second empirical application the integral is computed for $\text{dim}(\mathbf{y}_t)=28$ based on $7$ variables with $L=3$ lags. 

\paragraph{Choice of optimizer:} Numerical optimization is required to find a $\hat \param_n$ satisfying (\ref{eq:estimator}). Since $\hat Q_n^S$ is typically non-convex, has intractable derivatives and the numbers of coefficients is moderately large (between $11$ and $35$ in the application), a derivative-free global optimizer should be used. The simulations and application rely on particle swarm optimization, a stochastic search algorithm which converges fairly quickly. Matlab's implementation can evaluate $\hat Q_n^S$ in parallel, speeding up estimation significantly in the application where the policy function is very time-consuming to compute. After terminating the search, run several iterations of a local optimizer to check convergence. 

\paragraph{Modelling fat tails:} Gaussian mixtures can only approximate smooth densities $f$ sufficiently fast under a thin tail condition \citep{Kruijer2010}. Similar to \citet{Gallant1987}, adding a parametric tail component to form a Gaussian and tails mixture allows to model asymmetric excess tail behaviour:
\[ f_k(\cdot)  =\sum_{j=1}^k \frac{\omega_j}{\sigma_j} \phi \left(\frac{\cdot-\mu_j}{\sigma_j}\right) + \frac{\omega_{k+1}}{\sigma_{k+1}}\mymathcal{1}_{\cdot \leq \mu_{k+1}}f_L\left( \frac{\cdot-\mu_{k+1}}{\sigma_{k+1}} \right) + \frac{\omega_{k+2}}{\sigma_{k+2}}\mymathcal{1}_{\cdot \geq \mu_{k+2}}f_R\left( \frac{\cdot-\mu_{k+2}}{\sigma_{k+2}} \right), \] 
where $f_L(e,\tail_L) = (2+\tail_L)\frac{|e|^{1+\tail_L}}{[1+|e|^{2+\tail_L}]^2}$  for $e \leq 0$ and $f_R(e,\tail_R) = (2+\tail_R)\frac{e^{1+\tail_R}}{[1+e^{2+\tail_R}]^2}$ for $e \geq 0$ are the left and right tail components. They have finite variance if $\tail_L,\tail_R \geq 1$ which allows to prove $L^2$-smoothness of the tail draws. To sample from $f_L,f_R$, draw $u_L,u_R \sim \mathcal{U}_{[0,1]}$ and compute $Z_L = -(1/u_L-1)^{\frac{1}{2+\xi_L}},Z_R = (1/u_R-1)^{\frac{1}{2+\xi_R}}$.

\section{Asymptotic Properties} \label{sec:AsymMixture}


\subsection{Consistency} \label{sec:ConsistencyMixture}

Let $Q_n(\param) = \int \big| \mathbb{E} \big( \hat \psi_n(\tau) - \hat \psi_n^S(\tau,\param) \big)  \big|^2 \pi(\tau) d\tau$ be the population analog of the sample objective $\hat Q_n^S$. The dependence on $n$ arizes from $(y_t^s,x_t)$ not being covariance stationary since $y_0^s$ is usually not drawn from the stationary distribution. Since the CF is bounded, under geometric ergodicity, the dominated convergence theorem implies that it has a well-defined limit $Q(\param) = \int \big| \lim_{n\to\infty}\mathbb{E} \big( \hat \psi_n(\tau) - \hat \psi_n^S(\tau,\param) \big)  \big|^2 \pi(\tau) d\tau.$
For both $Q_n$ and $Q$, expectations are taken over the data $(\mathbf{y}_t,\mathbf{x}_t)$ and the simulated $(\mathbf{y}_t^s,\mathbf{x}_t)$. 

The space of true densities satisfying the assumptions will be denoted as $\mathcal{F}$ and $\mathcal{F}_{k}$ is the corresponding space of Gaussian and tails mixtures $\Pi_{k}f$. 
\begin{assumption}[Sieve, Identification, Dependence] \label{ass:sid} Suppose the following conditions hold. i) Sieve Space: the true density admits the decomposition $f = f_1 \times \dots \times f_{d_e}$ where for each $j=1,\dots,d_e$ $f_{j} = (1-\omega_{j,1}-\omega_{j,2})f_{j,S} + \omega_{j,1}f_L + \omega_{j,2}f_R$. $f_{j,S}$ is a smooth density with thin tails and the mixture space $\mathcal{F}_{k(n)}$ satisfying the assumptions of Lemma \ref{lem:ApproxSimuGAUT} with $k(n)^4\log[k(n)]^4/n \to 0$ as $k(n)$ and $n \to \infty.$  $\Theta$ is compact and $1 \leq \tail_L,\tail_R \leq \overline{\tail} < \infty$. ii) Identification: $\lim_{n\to \infty} \mathbb{E} \left( \hat \psi_n(\tau)-\hat \psi_n^s(\tau,\param)\right)=0, \pi  \text{ a.s. } \Leftrightarrow \|\param-\param_0\|_\paramspace = 0$. $\sup_\tau \|\tau\|_\infty \pi(\tau)^{1/4}$ is bounded and $\sqrt{\pi}$ is integrable. For any $n,k \geq 1$ and for all $\varepsilon>0$, $\inf_{ \param \in \paramspace_{k},\, \|\param-\param_0 \|_\paramspace \geq \varepsilon } Q_n(\param)$ is strictly positive and weakly decreasing in both $n$ and $k$. iii) Dependence: $(y_t,x_t)$ is strictly stationary and $\beta$-mixing with exponential decay, the simulated $(y_t^s(\param),x_t)$ are uniformly geometrically ergodic in $\param \in \paramspace$.
\end{assumption}

Condition \textit{i.} allows to use the approximation results in \citet{Kruijer2010}. Here the shocks are independent from one-another, this is a common assumption for structural shocks but could be restrictive in some settings.\footnote{The independence condition can be relaxed by using the results in \citet{DeJonge2010}.} The requirement on $k(n)$ is stronger than usual. First, the $\log[k(n)]$ term is due to simulating from the mixture. Second, the fourth-power is due to the non-standard dependence. The dependence properties of $y_t^s$ vary with $\param$ so that, even though it is strongly mixing, results from \citet{DoukhanP.MassartP.1995,Chen1998a} do not apply. Lemma \ref{lemma:maxineq_dep} provides a more conservative bound for the supremum of the empirical process, of order $\sqrt{k(n)^4\log[k(n)]^4/n}$ compared to $\sqrt{k(n)\log[k(n)]/n}$ with iid or strictly stationary data with fixed dependence.

Condition \textit{ii.} requires $L$ large enough and $g_{\text{obs}},g_{\text{latent}}$ to uniquely identify $\param = (\theta,f)$ as discussed for the AR(1) earlier. Condition \textit{iii.} is common in SMM \citep{Duffie1993}. It implies $(y_t^s,x_t)$ is strongly-mixing \citep{Liebscher2005} and the initial condition bias is negligible, i.e. $Q_n(\param_0)= O(1/n^2)$ as shown in Lemma \ref{lem:geom_ergo}.

Further restrictions on the data generating process are required for establishing uniform convergence of the simulated empirical process. Lemma \ref{lem:mixtures} below shows that mixture draws satisfy an $L^2$-smoothness property. Then, using restrictions on the DGP, Lemma \ref{lem:L2smooth} extends this property to the moments. Combined with Assumption \ref{ass:sid}, these allow to derive consistency and the rate of convergence of $\hat\param_n$.

\begin{lemma}[$L^2$-Smoothness of the Mixture Draws] \label{lem:mixtures}
  Let $e_t^s = \sum_{j=1}^{k(n)} \mymathcal{1}_{\nu^s_t \in [\sum_{l=0}^{j-1} \omega_l, \sum_{l=0}^{j} \omega_l]} ( \mu_j + \sigma_j Z_{t,j}^s )$ and $\tilde e_t^s = \sum_{j=1}^{k(n)} \mymathcal{1}_{\nu^s_t \in [\sum_{l=0}^{j-1}  \tilde \omega_l, \sum_{l=0}^{j} \tilde \omega_l]} ( \tilde \mu_j + \tilde \sigma_j Z_{t,j}^s )$ 
   with bounds $|\mu_j|,|\tilde \mu_j| \leq \bar{\mu}_{k(n)}$ and $|\sigma_j|,|\tilde \sigma_j| \leq \bar{\sigma}$ as in Lemma \ref{lem:ApproxSimuGAUT}. If $\mathbb{E}(|Z_{t,j}^s|^2) \leq C_Z^2 < \infty$ then there exists a finite constant $C$ which only depends on $C_Z$ such that:
   \[ \left[ \mathbb{E} \left( \sup_{\| f_{\omega,\mu,\sigma}-f_{\tilde\omega,\tilde\mu,\tilde\sigma} \|_m \leq \delta}  \Big| e_t^s-\tilde e_t^s  \Big|^2 \right) \right]^{1/2} \leq C \left(1+\bar \mu_{k(n)} + \bar \sigma + k(n)\right) \delta^{1/2}, \]
   where $\| f_{\omega,\mu,\sigma}-f_{\tilde\omega,\tilde\mu,\tilde\sigma} \|_m = \|(\omega,\mu,\sigma) - (\tilde\omega,\tilde\mu,\tilde\sigma)\|_1$.
\end{lemma}

The $L^2$-smoothness constant depends on the upper bound $\overline{\mu}_{k(n)} = O(\log[k(n)])$ and the sieve dimension $k(n)$ in the pseudo-norm $\|\cdot\|_m$. As shown in \citet{Kruijer2010}, $\|\cdot\|_{\text{TV}}$ and $\|\cdot\|_{\infty}$ are bounded above by $\|\cdot\|_m$ up to a multiplicative factor which depends on the scales' lower bound $\underline{\sigma}_{k(n)}$. This implies $L^2$-smoothness also holds in these norms.

\begin{assumption}[Data Generating Process] \label{ass:DGPMixt}
 $y_t^s$ is simulated according to (\ref{eq:observed})-(\ref{eq:unobserved}) where $g_{obs}$ and $g_{latent}$ satisfy the following H\"older conditions for some $\gamma \in (0,1]$:
\begin{enumerate}[nosep]
 \item[y(i).] $\|g_{obs}(y_1,x,\param_1,u)-g_{obs}(y_2,x,\param_1,u)\| \leq C_1(x,u) \|y_1-y_2\|$; $\mathbb{E}\left( C_1(x_t,u_t^s)^2|y_{t-1}^s \right) \leq \bar{C}_1 < 1$
 \item[y(ii).] $\|g_{obs}(y,x,\param_1,u)-g_{obs}(y,x,\param_2,u)\| \leq C_2(y,x,u) \|\param_1-\param_2\|_\paramspace^\gamma$; $\mathbb{E}\left( C(y_t^s,x_t,u_t^s)^2 \right) \leq \bar{C}_2 < \infty$
 \item[y(iii).] $\|g_{obs}(y,x,\param_1,u_1)-g_{obs}(y,x,\param_1,u_2)\| \leq C_3(y,x) \|u_1-u_2\|^\gamma$; $\mathbb{E}\left( C_3(y_t^s,x_t)^2|u_t^s \right) \leq \bar{C}_3 < \infty$
 \item[u(i).] $\|g_{latent}(u_1,\param_1,e)-g_{latent}(u_2,\param_1,e)\| \leq C_4(e) \|u_1-u_2\|$
; $\mathbb{E}\left( C_4(e_t^s)^2 \right) \leq \bar{C}_4 < 1$
 \item[u(ii).] $\|g_{latent}(u,\param_1,e)-g_{latent}(u,\param_2,e)\| \leq C_5(u,e) \|\param_1-\param_2\|_\paramspace^\gamma$
; $\mathbb{E}\left( C_5(u_{t-1}^s,e_t^s)^2 \right) \leq \bar{C}_5 < \infty$
 \item[u(iii).] $\|g_{latent}(u,\param_1,e_1)-g_{latent}(u,\param_1,e_2)\| \leq C_6(u) \|e_1-e_2\|$
; $\mathbb{E}\left( C_6(u_{t-1}^s)^2 \right) \leq \bar{C}_6 < \infty$
\end{enumerate}
for any $(\param_1,\param_2) \in \paramspace$, $(y_1,y_2) \in \mathbb{R}^{\text{dim}(y)}$, $(u_1,u_2) \in \mathbb{R}^{\text{dim}(u)}$ and $(e_1,e_2) \in \mathbb{R}^{\text{dim}(e)}$. $\|\cdot\|_\paramspace$ is either the TV or supremum norm.
\end{assumption} 

Assumption \ref{ass:DGPMixt} requires a contraction property \textit{y(i),u(i)} comparable to the $L^2$ unit circle condition in \citet{Duffie1993}. For an AR(1) model (\ref{eq:ar1}) this implies $|\rho| \leq \bar{C}_1<1$. While restrictive, it is enforced in SMM and particle-filter likelihood estimation of DSGE non-stationary models via de-trending and pruning of the state-space model. De-trending transforms deterministic or stochastic trends into stationary variables, pruning further guarantees stability by essentially enforcing \textit{y(i),u(i)}. For non-smooth models, Assumption \ref{ass:DGPmixtbis} in the Supplement substitutes H\"older with $L^2$-smoothness conditions. These assumptions allow to explicitly derive the effect of the approximation bias on $Q_n$, as shown in Lemma \ref{lem:ObjApproxRate}:
\[ Q_n(\Pi_{k(n)}\param_0) \asymp \max\left(\|\param_0-\Pi_{k(n)} \param_0\|_\paramspace^2 \log \left( \|\param_0-\Pi_{k(n)} \param_0\|_\paramspace \right)^2,\|\param_0-\Pi_{k(n)} \param_0\|_\paramspace^{2\gamma^2}, 1/n^2 \right), \]
where $\|\param_0-\Pi_{k(n)} \param_0\|_\paramspace^2 \log \left( \|\param_0-\Pi_{k(n)} \param_0\|_\paramspace \right)$ is due to approximation bias and its propagation via \textit{y(i),u(i)}, $1/n^2$ is due to nonstationarity, $\|\param_0-\Pi_{k(n)} \param_0\|_\paramspace^{2\gamma^2}$ comes from the H\"older conditions. The main difference with the literature is the propogation and accumulation of the approximation bias due to the dynamics of the model. 

\begin{lemma}[Assumption \ref{ass:DGPMixt}/\ref{ass:DGPmixtbis} implies $L^2$-Smoothness of the Moments] \label{lem:L2smooth}
  Suppose Assumption \ref{ass:DGPMixt} or \ref{ass:DGPmixtbis} and the conditions of Lemma \ref{lem:mixtures} hold. If $\|\tau\|_\infty \pi(\tau)^{1/4}$ is bounded, then there exists $\overline{C}>0$ such that for all $\delta >0$, uniformly in $t \geq 1$, $(\param_1,\param_2) \in \paramspace_{k(n)}$ and $\tau \in \mathbb{R}^{d_\tau}$:
  \begin{align*}
    &\mathbb{E} \left[ \sup_{\|\param_1-\param_2\|_m \leq \delta} \left| e^{i\tau^\prime (\mathbf{y}^s_t(\param_1),\mathbf{x}_t)}-e^{i\tau^\prime (\mathbf{y}_t^s(\param_2),\mathbf{x}_t)} \right|^2 \sqrt{\pi(\tau)} \right]
    \leq \overline{C} \max \left( \frac{\delta^{\gamma^2}}{ \underline{\sigma}_{k(n)}^{2\gamma^2} }, [k(n)+\overline{\mu}_{k(n)}+\overline{\sigma}]^\gamma \delta^{\gamma^2/2}\right)
  \end{align*}
  where $\|\param\|_m = \|\theta\|+\|(\omega,\mu,\sigma)\|_1$.
\end{lemma}

The key to establishing $L^2$-smoothness of the continuum of moments with unbounded support is to involve $\pi$. By Lispschitz-continuity of the sine and cosine functions: $|e^{i\tau^\prime (\mathbf{y}^s_t(\param_1),\mathbf{x}_t)}-e^{i\tau^\prime (\mathbf{y}_t^s(\param_2),\mathbf{x}_t)} | \pi(\tau) \leq 2 \|\mathbf{y}^s_t(\param_1)-\mathbf{y}^s_t(\param_2)\| \times \|\tau\|_\infty \pi(\tau)$. The simulated data are shown to be $L^2$-smooth under Lemma \ref{lem:mixtures} and Assumption \ref{ass:DGPMixt} or \ref{ass:DGPmixtbis}. With $\|\tau\|_\infty \pi(\tau)^{1/4}$ bounded this property holds for the ECF, uniformly in $\tau$. Further, $\sqrt{\pi}$ integrable implies $L^2$-smoothness of the weighted ECF distance used in $\hat Q_n^S$. These two conditions are also needed to handle the empirical process over the growing sieve space $\paramspace_{k(n)}$ and the unbounded index $\tau$. Density weights $\pi$ with fat tails, such as the Cauchy density, do not satisfy the last condition.

\begin{theorem}[Consistency] \label{th:consistencyMixtures}
  Suppose Assumptions \ref{ass:sid} and \ref{ass:DGPMixt} (or \ref{ass:DGPmixtbis}) hold, $Q_n(\cdot)$ is continuous on $(\paramspace_{k(n)},\|\cdot\|_{\paramspace})$ and the numerical optimization and integration errors are negligible, i.e. $\eta_n = o(1/n)$. If for all $\varepsilon>0$,
  \begin{align} \label{eq:separation}
    \max\left(\frac{log[k(n)]^{4r/b+2}}{k(n)^{2\gamma^2 r}},\frac{k(n)^4\log[k(n)]^4}{n}, \frac{1}{n^2} \right) = o\left( \inf_{\param \in \paramspace_{k(n)}, \| \param-\param_0 \|_\mathcal{B} \geq \varepsilon } Q_n(\param) \right),
  \end{align}
   where $r$ is the smoothness of the thin-tail component $f_S$ and $b$ its exponential tail index, then: 
   \[\|\hat \param_n-\param_0\|_\paramspace =o_p(1).\]
\end{theorem}
Theorem \ref{th:consistencyMixtures} is a consequence of the high-level consistency Lemma in \citet{Chen2012}. It is not a direct implication of their theorems due to non-standard dependence, continuum of moments, simulation process and bias propagation.

\subsection{Rate of Convergence} \label{sec:conv_rate_mixt}

Following \citet{Ai2003}, the rate of convergence is derived in a weak-norm given below.

\begin{assumption}[Weak Norm, Local Properties] \label{ass:weaknorm}
  Let $\paramspace_{osn} = \paramspace_{k(n)} \cap \{ \|\param-\param_0\|_\paramspace \leq \varepsilon \}$ be a neighborhood of $\param_0$ with $\varepsilon >0$ small. For any $(\param_1,\param_2) \in \paramspace_{osn}$: \[\|\param_1-\param_2\|_{weak}^2 = \int \Big| \frac{d \mathbb{E} \left( \hat \psi^S_n(\tau,\param_0) \right)}{d \param}[\param_1-\param_2] \Big|^2 \pi(\tau)d\tau \]
  is the weak norm of $\param_1-\param_2$. The derivative $(\param_1,\param_2) \to \frac{d \mathbb{E} \left( \hat \psi^S_n(\tau,\param_1) \right)}{d \param}[\param_2]$ is continuous in $\param_1$, linear in $\param_2$. Suppose there exists $\underline{C}_w >0$ such that for all $\param \in \paramspace_{osn}$: $\underline{C}_w \|\param-\param_0\|^2_{weak} \leq  \int \big|  \mathbb{E} \left( \hat \psi^S_n(\tau,\param_0)-\hat\psi^S_n(\tau,\param) \right)  \big|^2 \pi(\tau)d\tau.$
\end{assumption}

Assumption \ref{ass:weaknorm} with the rate $Q_n(\Pi_{k(n)}\param_0)$ discussed above allows to bound the approximation error $\|\Pi_{k(n)}\param_0-\param_0\|_{weak}$ in the weak norm. Then, standard arguments combined with the results on the empirical process derived for consistency imply the result below.

\begin{theorem}[Rate of Convergence] \label{th:convrateMixture}
  Suppose that the assumptions for Theorem \ref{th:consistencyMixtures} hold and Assumption \ref{ass:weaknorm} also holds.The convergence rate in weak norm is:
  \begin{align} \label{eq:sv_rate_weaknorm1}
    &\|\hat \param_n - \param_0\|_{weak} = O_p\left(\max\left( \frac{\log[k(n)]^{r/b+1}}{k(n)^{\gamma^2 r}},\frac{k(n)^2\log[k(n)]^2}{\sqrt{n}} \right)\right).
  \end{align}
   The convergence rate in either the total variation or supremum norm $\|\cdot\|_\paramspace$ is:
  \begin{align*}
    &\|\hat \param_n - \param_0\|_{\paramspace} = O_p\left(\frac{\log[k(n)]^{r/b}}{k(n)^r}+\tau_{\paramspace,n}\max\left( \frac{\log[k(n)]^{r/b+1}}{k(n)^{\gamma^2 r}},\frac{k(n)^2\log[k(n)]^2}{\sqrt{n}} \right)\right)
  \end{align*}
  where $\tau_{\paramspace,n}$ is the local measure of ill-posedness:
  $\tau_{\paramspace,n} = \sup_{\param\in\paramspace_{osn}, \,\|\param-\Pi_{k(n)}\param_0\|_{weak}\neq 0} \frac{\|\param-\Pi_{k(n)}\param_0\|_\paramspace \,\,}{ \quad\|\param-\Pi_{k(n)}\param_0\|_{weak}}$.
\end{theorem}

As usual, the rate of convergence involves a bias/variance trade-off. Here, the bias is inflated because of the dynamics. The variance is larger than usual because of the conservative empirical process bound. This implies slower convergence compared to the iid case. 

There are two sources of ill-posedness in this setting. First, the distance between two CFs is weaker than the TV or supremum distance: the CF characterizes convergence in distribution while the other two do not. Second, the problem may be fundamentally ill-posed in which case convergence is necessarily slower in the strong than in the weak norm. Lemma \ref{lem:cv_rate_mixture_norm} relates the convergence in $\|\cdot\|_{weak}$ to rate in $\|\cdot\|_m$ which is useful for proving asymptotic normality. A by-product of this Lemma is a simple bound on $\tau_{n,TV/\infty}$. From \citet{Kruijer2010}, $\|\param-\Pi_{k(n)}\param_0\|_{TV} \leq \underline{\sigma}_{k(n)}^{-1}\|\param-\Pi_{k(n)}\param_0\|_m$ and $\|\param-\Pi_{k(n)}\param_0\|_{\infty} \leq \underline{\sigma}_{k(n)}^{-2}\|\param-\Pi_{k(n)}\param_0\|_m$ on $\paramspace_{k(n)}$. Combined with the Lemma, $\tau_{TV,n} \leq \underline{\lambda}_n^{-1/2}\underline{\sigma}_{k(n)}^{-1}$ and $\tau_{\infty,n} \leq \underline{\lambda}_n^{-1/2}\underline{\sigma}_{k(n)}^{-2}$ where $\lambda_n$ measures local curvature and can be approximated numerically. Note that decreasing $\underline{\sigma}_{k(n)}$ too fast as $k(n) \to \infty$ deteriorates the rate of convergence. 
For SMM, a larger $S$ implies a smaller asymptotic variance for the estimates. Here, a refinement of the theorem shows that using $S \to \infty$ as $n \to \infty$ can additionally result in faster convergence.

\begin{corollary}[Number of Simulated Samples $S$ and Rate of Convergence] \label{rmk:full_simu}
  Suppose a long sample $(y_{1}^s,\dots,y_{n S}^s)$ can be simulated. Then given $k(n)$, (\ref{eq:sv_rate_weaknorm1}) becomes: \[\|\hat \param_n - \param_0\|_{weak} = O_p\left(\max\left( \frac{\log[k(n)]^{r/b+1}}{k(n)^{\gamma^2 r}},\max\left(\frac{k(n)^2\log[n]^2}{\sqrt{n\times S}},\frac{1}{\sqrt{n}}\right) \right)\right).\]
The fastest possible rate in weak norm is then
$\|\hat \param_n - \param_0\|_{weak} = O_p\left(\max\left( \frac{\log[k(n)]^{r/b+1}}{k(n)^{\gamma^2 r}},\frac{1}{\sqrt{n}} \right)\right)$ which is attained with $S(n) \asymp k(n)^4\log[k(n)]^4$. The fastest rate in TV or supremum norm is then: $\|\hat \param_n - \param_0\|_{\paramspace} = O_p\left(\frac{\log[k(n)]^{r/b}}{k(n)^r}+\tau_{\paramspace,n}\max\left( \frac{\log[k(n)]^{r/b+1}}{k(n)^{\gamma^2 r}},\frac{1}{\sqrt{n}} \right)\right)$.
\end{corollary}
Asymptotic normality requires sufficiently fast convergence which usually implies stronger smoothness restrictions on the unknown $f$. Corollary \ref{rmk:full_simu} implies that smoothness requirements can be replaced with the computation requirement of making $S$ large, allowing $k(n)$ to grow more rapidly. Variance reduction techniques are often used in empirical work to reduce simulation noise without taking $S$ large. Whether they could also enhance convergence rates here could be an interesting avenue for research.

\subsection{Asymptotic Normality} \label{sec:asymnorm_mixt}

The following provides asymptotic normality results for plug-in estimates $\phi(\hat \param_n)$ where $\phi$ are smooth functionals of the parameters. The main steps to derive the results are fairly standard. As in the finite-dimensional case, stochastic equicontinuity results are needed to derive these results which are derived under  $\|\cdot\|_m$ in Lemmas  \ref{lem:cv_rate_mixture_norm}, \ref{lem:stoch_eq_mixture_short}; the natural norm for handling simulation draws. Let $M_n = \log\log(n+1)$, $\delta_n$ is the rate of convergence in weak norm above and $\underline{\lambda}_n = \lambda_{\min}( \int \frac{d\mathbb{E}(\hat \psi_n^S(\tau,\Pi_{k(n)}\param_0))}{d(\theta,\omega,\mu,\sigma)}^\prime \overline{\frac{d\mathbb{E}(\hat \psi_n^S(\tau,\Pi_{k(n)}\param_0))}{d(\theta,\omega,\mu,\sigma)}} \pi(\tau)d\tau )$ assumed strictly positive.
\begin{definition}[Sieve Representer, Score and Variance] \label{def:Sieve_Variance}
  Let $\param_{0,n}$ be such that $\|\param_{0,n}-\param_0\|_{weak} = \inf_{\param \in \paramspace_{osn}} \| \param-\param_0\|_{weak}$, let $\overline{V}_{k(n)}$ be the closed span of $\paramspace_{osn} - \{\param_{0,n}\}$. The inner product $\langle\cdot,\cdot\rangle$ of $(v_1,v_2) \in \overline{V}_{k(n)}$ is defined as:
  $\langle v_1,v_2 \rangle = \frac{1}{2} \int [ \psi_\param(\tau,v_1) \overline{\psi_\param(\tau,v_2)} +  \overline{\psi_\param(\tau,v_1)} \psi_\param(\tau,v_2)]\pi(\tau)d\tau.$ The sieve representer is the unique $v_n^* \in \overline{V}_{k(n)}$ such that $\langle v_n^*,v  \rangle = \frac{d \phi(\param_0)}{d\param}[v], \forall v \in \overline{V}_{k(n)}$. The sieve score $S_n^*$ is:
    $S_n^* = \int Real ( \psi_\param(\tau,v_n^*) \overline{[\hat \psi_n^S(\tau,\param_0)-\hat \psi_n(\tau)]} )\pi(\tau)d\tau$ and 
    the sieve long-run variance $\sigma_n^{*2} = n\mathbb{E}(S_n^{*2})
                    = n\mathbb{E}( [ \int Real ( \psi_\param(\tau,v_n^*) \overline{[\hat \psi_n^S(\tau,\param_0)-\hat \psi_n(\tau)]} )\pi(\tau)d\tau ]^2 ).$
    The scaled sieve representer $u_n^*$ is: $u_n^* = v_n^*/\sigma_n^*.$
\end{definition}

\begin{assumption}[Equivalence Condition] \label{ass:equiv_cond_mixt}
There exists $\underline{a}>0$ such that for all $n \geq 1$:
$\underline{a}\|v_n^*\|_{weak} \leq \sigma_n^*.$ Furthermore, suppose that $\sigma_n^*$ does not increase too fast: $\sigma_n^* =o(\sqrt{n}).$
\end{assumption}


\begin{assumption}[Convergence Rate, Smoothness, Bias] \label{ass:sufficient_cv_rate_mixture} Suppose that the set $\paramspace_{osn}$ is a convex neighborhood of $\param_0$ and: i) Rate of convergence: $M_n\delta_n = o(n^{-1/4})$ and $M_n\delta_n = o( \sqrt{\underline{\lambda}_n}/\left(k(n)\log(n)\right)^{4/\gamma^2} )$. ii) Smoothness: a linear expansion of $\phi$ is locally uniformly valid $\sup_{\|\param-\param_0\| \leq M_n \delta_n } \frac{\sqrt{n}}{\sigma_n^*} \Big| \phi(\param) - \phi(\param_0) - \frac{d\phi(\param_0)}{d\param}[\param-\param_0]\Big| = o(1)$ and $(\param_1,\param_2) \to \frac{d\phi(\param_1)}{d\param}[\param_2]$ is continuous in $\param_1$, linear in $\param_2$,  as well as for the moments
    $\sup_{\|\param-\param_0\|_{weak}\leq M_n\delta_n} ( \int \big| \mathbb{E}(\hat \psi_n^S(\tau,\param))-\mathbb{E}(\hat \psi_n^S(\tau,\param_0))- \frac{d\mathbb{E}(\hat \psi_n^S(\tau,\param_0))}{d\param}[\param-\param_0]\big|^2 \pi(\tau)d\tau )^{1/2} = O([M_n\delta_n ]^2 ).$
    Bounded second derivative:
    $\sup_{\|\param-\param_0\|_{weak}\leq M_n\delta_n}  \int \Big| \frac{d^2\mathbb{E}(\hat \psi_n^S(\tau,\param_0))}{d\param d\param}[u_n^*,u_n^*]\Big|^2 \pi(\tau)d\tau  = O(1).$
    iii) Bias: negligible approximation bias: $\frac{\sqrt{n}}{\sigma_n^*} \frac{d\phi(\param_0)}{d\param}[\param_{0,n}-\param_0] = o(1).$
\end{assumption}

Definition \ref{def:Sieve_Variance} adapts standard quantities to a continuum of complex-valued moments, the sieve variance corresponds to the square of the standard errors used for inference. A simple plug-in estimator is described below. Similarly, Assumptions \ref{ass:equiv_cond_mixt}, \ref{ass:sufficient_cv_rate_mixture} are commonly used to derive asymptotic linear expansions to then apply a Central Limit Theorem to the leading term. By Corollary \ref{rmk:full_simu} above, allowing $S\to\infty$ makes the rate assumptions \ref{ass:sufficient_cv_rate_mixture}\textit{i)} easier to verify. Condition \textit{ii)} automatically holds for linear functionals, like reporting the finite dimensional $\theta$ or pointwise evaluation of the density $f$. 

\begin{theorem}[Asymptotic Normality] \label{th:asymnormal_mixture}
  Suppose the assumptions of Theorems \ref{th:consistencyMixtures}, \ref{th:convrateMixture} and Lemmas  \ref{lem:cv_rate_mixture_norm}, \ref{lem:stoch_eq_mixture_short} hold as well as Assumptions \ref{ass:equiv_cond_mixt} and \ref{ass:sufficient_cv_rate_mixture}, then as $n$ goes to infinity:
  \[ r_n \times \left( \phi(\hat \param_n)-\phi(\param_0) \right) \overset{d}{\to} \mathcal{N}\left(0,1\right), \text{ where } r_n=\sqrt{n}/\sigma^*_n \to \infty.\]
\end{theorem}

Theorem \ref{th:asymnormal_mixture} shows that, under the above assumptions, inference on $\phi(\param_0)$ can be conducted using the confidence interval $[\phi(\hat \param_n) \pm 1.96 \times \sigma_n^*/\sqrt{n}]$. The standard errors $\sigma_n^* > 0$ adjust automatically so that $r_n=\sqrt{n}/\sigma_n^*$ gives the correct rate of convergence. If $\lim_{n\to \infty} \sigma_n^* < \infty$, then $\phi(\hat \param_n)$ is $\sqrt{n}-$convergent. The dependence conditions are sufficient to apply the Central Limit Theorem of \citet{Wooldridge1988} which yields the result.

To compute standard errors in practice, two matrices are computed and multiplied in a sandwich form.\footnote{Alternatively, one could build confidence sets by inverting a test based on an Integrated Conditional Moment (ICM) statistic, see \citet{Santos2012} for an application to NPIV.} First, the bread is computed using:
\[ D_n = \text{real}\left[\int \partial^\prime_{\theta,\omega,\mu,\sigma} \hat\psi_n^S(\tau) \overline{\partial_{\theta,\omega,\mu,\sigma} \hat\psi_n^S(\tau)} \pi(\tau)d\tau \right], \]
where $\overline{\partial_{\theta,\omega,\mu,\sigma} \hat\psi_n^S(\tau)}$ is the complex conjugate of $\partial_{\theta,\omega,\mu,\sigma} \hat\psi_n^S(\tau)$ evaluated at the estimates $(\hat \theta_n,\hat \omega_n,\hat \mu_n,\hat \sigma_n)$. Derivatives are computed by finite differences. With a finite integration grid, this is simply a matrix product. Second, the meat is the sum of HAC variance estimates $V_1,V_{1,S}$ for the following two vector-valued series:
\begin{align*}
  Z_t &= \int \left(\text{real}[\partial^\prime_{\theta,\omega,\mu,\sigma} \hat\psi_n^S(\tau)]\text{real}[\hat\psi_t(\tau)] + \text{im}[\partial^\prime_{\theta,\omega,\mu,\sigma} \hat\psi_n^S(\tau)]\text{im}[\hat\psi_t(\tau)]\right)\pi(\tau)d\tau,\\
  Z_{t,S} &= \int \left(\text{real}[\partial^\prime_{\theta,\omega,\mu,\sigma} \hat\psi_t^S(\tau)]\text{real}[\hat\psi_t(\tau)] + \text{im}[\partial^\prime_{\theta,\omega,\mu,\sigma} \hat\psi_n^S(\tau)]\text{im}[\hat\psi_t^S(\tau)]\right)\pi(\tau)d\tau,
\end{align*}
where real and im take the real and imaginary part.
As in the finite-dimensional case $V_{1,S} = V_{1,1}/S$: a larger $S$ implies more precise estimates. The plug-in estimate of $\sigma_n^{*2}$ is:
\[  \hat \sigma_n^{*2} = \partial_{\theta,\omega,\mu\sigma} \phi(\hat \param_n) D_n^{-1}(V_1+V_{1,S}) D_n^{-1} \partial^\prime_{\theta,\omega,\mu\sigma} \phi(\hat \param_n). \]
Supplemental Appendix \ref{sec:sieveLRR} derives the formula and provides the assumptions required for consistency of the standard errors.
On efficiency, it can be showed that estimating $f$ typically affects the estimates for $\theta$. Further, \citet{yu2004} shows for the location-scale Gaussian model that minimizing the ECF distance is not efficient because the objective puts weight on all moments. To achieve efficiency, \citet{Carrasco2000} apply a regularized inverse of the covariance operator to the moment function which extends optimal weighting in GMM. The proofs for the results above allow to apply a bounded operator $B$ which implies a non-vanishing regularization in their setting. While this could improve the properties of the estimates in theory, simulations (not reported here) suggest that the estimates are sensitive to the choice of regularization parameter. 

\section{Monte-Carlo Illustrations} \label{sec:MonteCarlo}
Three simple examples illustrate the properties of the estimator and compare it with a Bayesian nonparametric estimator based on Gaussian mixtures. All examples are conducted in R using the \textit{PSO} package. Given sample and simulated data, the ECF and ECF distance are computed with \textit{RcppArmadillo} which is more efficient for standard matrix operations than baseline R. The unknown $f$ is the skewed-logistic distribution. To illustrate Theorem \ref{th:asymnormal_mixture}, rejection rates of confidence intervals for $\theta$ are reported. The DGPs considered are:
\begin{align}
  y_t &= e_t, \label{eq:iid}\\
  y_t &= \rho_y y_{t-1} + e_t, \label{eq:ar_1}\\
  y_t &= \mu_y + \rho_y (y_{t-1} - \mu_y) + \sigma_t (e_{1,t} + \vartheta_y e_{1,t-1}), \, \sigma_t^2 = \mu_\sigma + \rho_\sigma \sigma_{t-1}^2 + \kappa_\sigma e_{2,t}. \label{eq:SV2}
\end{align}
where $e_t \overset{iid}{\sim} f$ in (\ref{eq:iid})-(\ref{eq:ar_1}) and $e_{1,t} \overset{iid}{\sim} f$, $e_{1,t} \overset{iid}{\sim} \chi^2_1$ in (\ref{eq:SV2}). (\ref{eq:iid}) and (\ref{eq:ar_1}) provide a benchmark to illustrate the effect of the dependence on the estimated $\hat f_n$ and the effect of $S$ in Corollary \ref{rmk:full_simu}. In (\ref{eq:iid}), $\beta = (f)$, and for (\ref{eq:ar_1}) $\beta = (\rho_y,f)$. The stochastic volatility (SV) model (\ref{eq:SV2}) illustrates the first empirical example with a DGP similar to those used in estimations of Long-Run Risks (LRR) models. In (\ref{eq:SV2}), $f$ is restricted to have mean zero and unit variance. For DGPs (\ref{eq:iid})-(\ref{eq:ar_1}), the sample size is $n=200$; for (\ref{eq:SV2}), $n=750$ similar to the application. The inputs are chosen as described in Section \ref{sec:the_estimator}, with $L=0,1$ for (\ref{eq:iid}), (\ref{eq:ar_1}) respectively. In (\ref{eq:SV2}), $L=5$ is used -- it is sufficiently large to identify $(\theta,f)$,  where $\theta = (\mu_y,\rho_y,\vartheta_y,\mu_\sigma,\rho_\sigma,\kappa_\sigma)$. The first two examples  use $200$ and the third $500$ integration points. $1000$ Monte-Carlo replications are used in (\ref{eq:iid})-(\ref{eq:ar_1}), $200$ in (\ref{eq:SV2}). 
\begin{table}[h] \setlength\tabcolsep{4.5pt} \renewcommand{\arraystretch}{0.935}
  \begin{center} \caption{Models (\ref{eq:ar_1})-(\ref{eq:SV2}) -- Bias, Standard Deviation and Size for $\rho_y$} \label{tab:resMC} {
    \small
    \begin{tabular}{cl|aaaa|bb|b} \hline \hline
      & & \multicolumn{4}{c}{Sieve-SMM} & \multicolumn{2}{c}{Bayesian} & GMM\\ \hline
      & $k$ & \multicolumn{2}{c}{2} & \multicolumn{2}{c|}{3} & 2 & 3 & -\\
      & $S$ & \mc{1}{1} & \mc{1}{5} & \mc{1}{1} & \multicolumn{1}{c|}{5} & - & - & -\\ \hline
      \multirow{3}{*}{AR(1)}       & bias  & -0.014 & -0.010 & -0.017 & -0.016 & -0.011 & -0.010 & -0.015 \\ 
      & std & \,\,0.082 & \,\,0.064 & \,\,0.077 & \,\,0.062 & \,\,0.048 & \,\,0.049 & \,\,0.056 \\
      & size  & \,\,0.044 & \,\,0.033 & \,\,0.026 & \,\,0.018 & \,\,0.051 & \,\,0.054 & \,\,0.019 \\ \hline
      \multirow{3}{*}{SV} & bias  & -0.003 & -0.006 & -0.002 & -0.006 & - & - & -0.006 \\ 
       & std & \,\,0.014 & \,\,0.012 & \,\,0.015 & \,\,0.012 & - & - & \,\,0.027 \\
      & size  & \,\,0.200 & \,\,0.170 & \,\,0.190 & \,\,0.140 & - & - & \,\,0.060 \\ \hline  \hline
    \end{tabular} }\\
  {\footnotesize \textit{Note: size reported for 95\% confidence intervals. Model (\ref{eq:ar_1}): $\rho_y = 0.6$, $n=200$.\\ Model (\ref{eq:SV2}): $(\mu_y,\rho_y,\vartheta_y,\rho_\sigma,\kappa_\sigma) = (0.025,0.98,-0.73,0.7,0.6)$, $n=750$}}
  \end{center}
\end{table}
Bayesian estimation is conducted using a Metropolis-Hastings algorithm, the proposal is tuned to target an acceptance rate between $20$ and $40\%$ accross simulations. For (\ref{eq:iid}), (\ref{eq:ar_1}) the likelihood is computed analytically. Bayesian estimates are not reported for (\ref{eq:SV2}), due to the computational burden of performing many Monte Carlo replications. 
The prior is uniform for $\rho$, Dirichlet(1/2) for $\omega$, $\mathcal{N}(0,10)$ for $\mu$ and inverse-Gamma($2.1,1.1$) for $\sigma$. For reference, a semiparametric GMM estimator is also reported for (\ref{eq:ar_1}), (\ref{eq:SV2}); the usual OLS estimator for (\ref{eq:ar_1}) and moments conditions that identify $(\mu_y,\rho_y)$ separately from other parameters in $(\theta,f)$. Figure \ref{fig:MCsimu} illustrates estimates of $f$ in (\ref{eq:iid})-(\ref{eq:SV2}) for $k=3$, $S=1,5$ and compares with Bayesian estimates in (\ref{eq:iid})-(\ref{eq:ar_1}) and an infeasible kernel estimator that directly observed $e_{1,t}$ in (\ref{eq:SV2}).
\begin{figure}[h]
  \begin{center} \caption{Models (\ref{eq:iid})-(\ref{eq:SV2}) -- Density Estimates: Sieve-SMM and Bayesian}\label{fig:MCsimu}
  \includegraphics[scale=0.48]{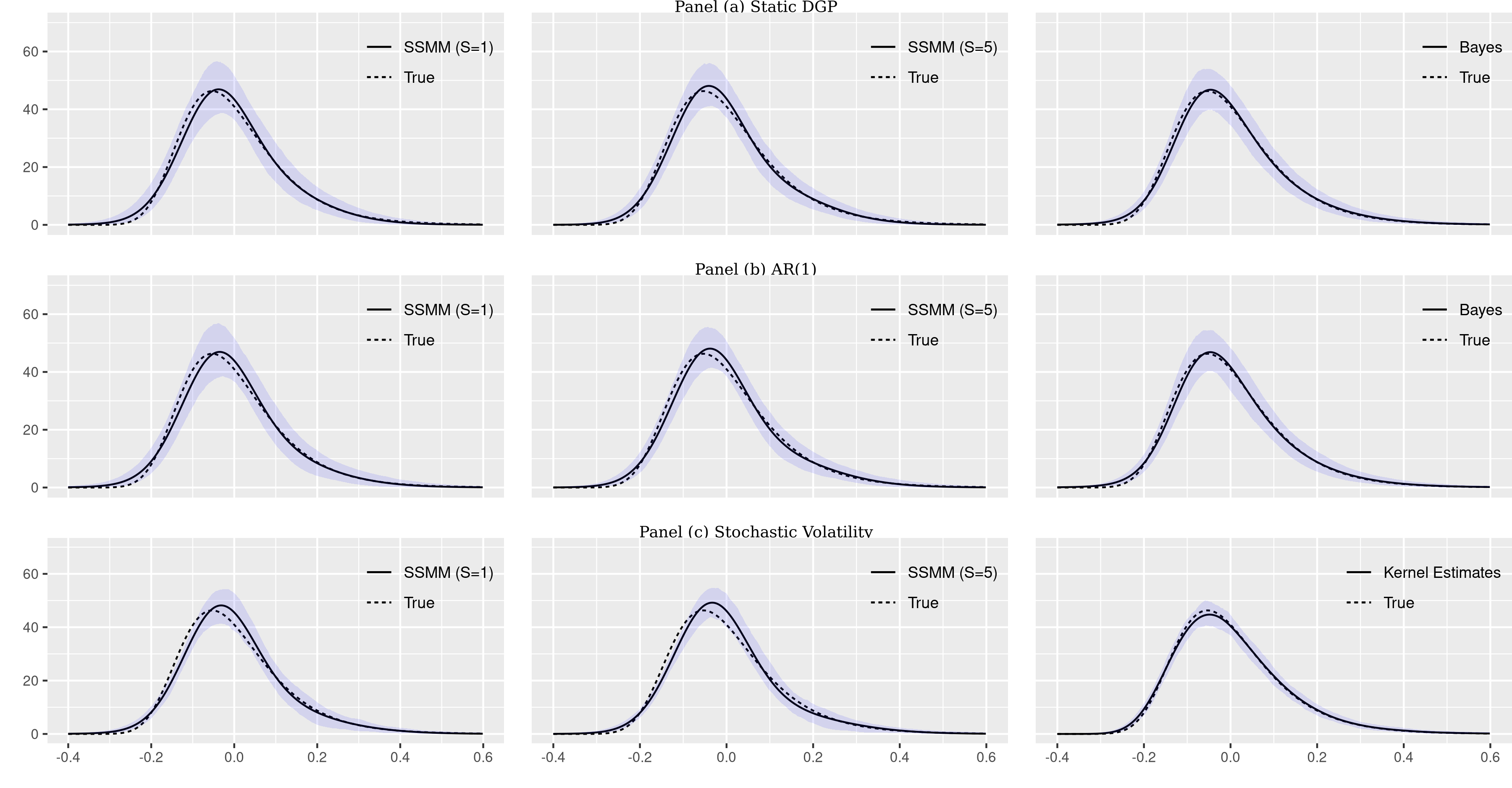}\\
  {\footnotesize \textit{Note: bands = 95\% pointwise interquantile range, $n=200$ for (a)-(b), $n=750$ for (c), $k=3$.}} \end{center} 
\end{figure} 
Table \ref{tab:resMC} summarizes the simulation results for $\rho_y=0.6$ in (\ref{eq:ar_1}) and $\rho_y = 0.98$ in (\ref{eq:SV2}). For the AR(1), biases, standard deviations and sizes are similar across methods. There is size distortion due to small sample bias. There is more size distortion for the SV model (\ref{eq:SV2}), but distortion declines with $k$ and $S$. Here, using a larger $k=4$ leads to rejection rates of $0.11$ and $0.10$ for $S=5,20$ respectively which is closer to the nominal $5\%$ level. Figure \ref{fig:MCsimu} show the properties of the estimated distribution $f$. Both bias and variance are slightly larger with dynamics (panel b) compared to the static case (panel a). Bias is larger for the SV model (panel c) where the higher-order moments of $y$ identify both $(\mu_\sigma,\rho_\sigma,\kappa_\sigma)$ and $f$. 
\section{Applications to Asset Pricing}\label{sec:empirical}
\subsection{Non-Gaussian Shocks and Long-Run Uncertainty}
The first application considers a reduced-form specification of the consumption process in \citet{Yaron2016}. They model consumption growth as a persistent AR(1) process plus white noise with a shared stochastic volatility component. The reduced form used here is the same ARMA(1,1) with time-varying volatility process used in the simulations above:
\begin{align}
  y_t &= \mu_y + \rho_y (y_{t-1} - \mu_y) + \sigma_t (e_{1,t} + \vartheta_y e_{1,t-1}), \, \sigma_t^2 = \mu_\sigma + \rho_\sigma \sigma_{t-1}^2 + \kappa_\sigma e_{2,t}, \tag{\ref{eq:SV2}}
\end{align}
The data consists of real monthly consumption growth, excluding food and energy, from Feb 1959 to Dec 2019.
Parameter estimates and standard errors for $S=20$ are reported in Table \ref{tab:SVest}. Gaussian ARMA QMLE estimates (without SV) are reported for reference. The estimate of $\hat\rho_{y,n}$ is large, in line with calibrations and estimates in the LRR literature. The large negative $\hat \vartheta_{y,n}$ further confirms that the persistent long-run risk component is small. $\hat\kappa_{\sigma,n}$ is multiplied by $10^4$ in the table for readability. Volatility is less persistent than calibrated in \citet{Yaron2016}, but its magnitude is comparable. Table \ref{tab:riskfree} shows the effect of uncertainty on the risk-free rate by evaluating the last term in $r_t = \text{const} - \log[\mathbb{E}_t(\exp(-\gamma \sigma_{t+1} e_{1,t+1}))]$ conditional on $\sigma_t = \overline{\sigma}$, $e_{1,t}=0$; i.e. both are set equal to their long-run average. For the Gaussian ARMA model, the effect is small and negative. The Gaussian SV model finds a positive effect. The symmetry in the distribution implies equal probability for large positive and negative outcomes. This can lead to surprising results such as a higher welfare with time-varying uncertainty than in a deterministic economy \citep{Cho2015}. Using a recursive utility with preference for early resolutions of uncertainty can negate this positive income effect with a intertemporal substitution effect. Without changing the utility function, mixture estimates with $k=2,3,4$ find a larger, negative term compared to the other two baseline predictions. This simple exercise suggests $f_e$ can have interesting asset pricing implications which are further explored in the second application.
\begin{figure}[ht]
  \begin{floatrow} 
    \capbtabbox{%
    \caption{Model (\ref{eq:SV2}) -- Estimates, Standard Errors} \label{tab:SVest}
    }{  \setlength\tabcolsep{4.5pt} \renewcommand{\arraystretch}{0.935}
    \small
  \begin{tabular}{l|ccccc}
    \hline \hline
     $k$ & $\mu_y$ & $\rho_y$ & $\vartheta_y$ & $\rho_\sigma$ & $\kappa_\sigma$ \\ 
    \hline
    \multirow{2}{*}{\textsc{arma}} & 0.030 & 0.988 & -0.759 & - & - \\ 
     & (0.008) & (0.006) & (0.026) & - & - \\ \hline
    \multirow{2}{*}{$1$} & 0.028 & 0.986 & -0.803 & 0.656 & 0.458 \\ 
     & (0.002) & (0.005) & (0.103) & (0.316) & (0.184) \\ \hline
     \multirow{2}{*}{$2$} & 0.029 & 0.987 & -0.784 & 0.736 & 0.370 \\ 
     & (0.002) & (0.005) & (0.086) & (0.330) & (0.253) \\ \hline
     \multirow{2}{*}{$3$} & 0.029 & 0.988 & -0.806 & 0.718 & 0.397 \\ 
     & (0.003) & (0.004) & (0.120) & (0.387) & (0.327) \\ \hline
     \multirow{2}{*}{$4$} & 0.029 & 0.988 & -0.829 & 0.680 & 0.430 \\ 
     & (0.003) & (0.005) & (0.123) & (0.374) & (0.251) \\ 
     \hline \hline
  \end{tabular} 
  }
  \capbtabbox{  \setlength\tabcolsep{4.5pt} \renewcommand{\arraystretch}{0.935}
    \small \hspace{-0.8cm}
    \begin{tabular}{l|ccccc}
      \hline \hline
     $k/\gamma$ & 1 & 2 & 4 & 6 & 10 \\ 
      \hline
      \textsc{arma} & -0.00 & -0.00 & -0.00 & -0.00 & -0.01 \\ 
      $1$ & 0.03 & 0.05 & 0.10 & 0.14 & 0.21 \\ 
      $2$ & -0.50 & -0.99 & -1.99 & -2.98 & -4.97 \\ 
      $3$ & -0.49 & -0.99 & -1.97 & -2.96 & -4.93 \\ 
      $4$ & -0.65 & -1.29 & -2.58 & -3.88 & -6.46 \\ 
       \hline \hline
    \end{tabular}
  }{%
  \caption{Effect of Uncertainty on $r_t$} \label{tab:riskfree}
  }
  \end{floatrow}
  {\footnotesize \textit{Note: Monthly consumption growth data: 02/1959-12/2019, $S=20$.}}
  \end{figure}
  
\subsection{Bond Pricing in a Production Economy}

This empirical application illustrates the empirical relevance of non-Gaussian shocks for estimates of relative risk-aversion using the model of \citet{VanBinsbergen2012}. They estimate a bond pricing model in a production economy with inflation and recursive utility by maximum likelihood using the particle-filter and report large estimates of risk-aversion. 
\paragraph{Model:} A representative agent maximizes intertemporal utility over consumption $C_t$:
\begin{align} U_t = \left[C_t^{1-\psi} + \beta \left[\mathbb{E}_t\left(U_{t+1}^{1-\gamma}\right)\right]^{\frac{1-\psi}{1-\gamma}}\right]^{\frac{1}{1-\psi}}, \end{align} 
where $\beta \in (0,1)$ is the discount factor, $\gamma$ measures relative risk aversion and $1/\psi$ the intertemporal elasticity of substitution (IES). If $\gamma=\psi$, the utility becomes CRRA. Leisure is omitted here because the calibrated specification in \citet{VanBinsbergen2012} fits the data very poorly.\footnote{See \citet{Gourio2012} Section IV.A and \citet{rudebusch2012} p110 for more detailed discussions.} Technology evolves in logs according to a random-walk with drift:
\begin{align} \log Z_{t+1} = \lambda + \log Z_{t} + e_{1,t+1}, \end{align} 
where $e_{1,t+1} \overset{iid}{\sim} f_1$ has mean zero. The budget and ressource constraints are: 
\begin{align*}
  C_t + I_t + \frac{B_{t+1}}{P_t R_t} = r_t K_t + w_t l_t + \frac{B_t}{P_t}, \quad Y_t = C_t + I_t.
\end{align*} 
$I_t$ is investement, $K_t$ capital, $l_t=1$ hours worked, $B_t$ number of contingent bonds with price $1/R_t$, $P_{t}$ the price of goods and $Y_t= Z_t^{1-\alpha}K_t^\alpha$ is aggregate output. Capital accumulation evolves according to $K_{t+1} = (1-\delta)K_t + G(I_t/K_t)K_t$.  $\delta\in(0,1)$ is the depreciation rate, $G(I_t/K_t) = a_1 + \frac{a_2}{1-1/\tau}(I_t/K_t)^{1-1/\tau}$ with $\tau > 0$ measures adjustment costs as in \citet{jermann1998}. $a_1,a_2$ are set to have no adjustement costs in the steady-state.  Inflation $\pi_{t+1} = P_{t+1}/P_t$ follows ARMA(1,1) dynamics in logs, where the MA(1) component is the sum of two independent MA(1) processes:
\begin{align} \log\pi_{t+1} = \bar\pi + \rho (\log\pi_t-\bar\pi) + e_{2,t+1} + \nu_{\pi} e_{2,t} + \nu_{z,1}e_{1,t+1} + \nu_{z,2} e_{1,t}, \end{align}
where $e_{2,t+1} \overset{iid}{\sim} f_2$ has mean zero. The stochastic discount factor (SDF) is:
\begin{align} M_{t+1} = \beta \mathbb{E}\left[ \left( \frac{V_{t+1}}{W_t} \right)^{\psi-\gamma} \left( \frac{C_{t+1}}{C_t} \right)^{-\psi}\right], \end{align}
where $V_t = \max_{C_t,I_t} U_t$ is the value function and $W_t = \mathbb{E}_t[V_{t+1}^{1-\gamma}]^{\frac{1}{1-\gamma}}$ is certainty-equivalent future utility. With the SDF, the price $Q_{t,\ell}$ of a $\ell \geq 1$ period bond is computed recursively:
\begin{align}
  Q_{t,\ell} = \mathbb{E}_t\left( M_{t+1} \frac{Q_{t+1}^{\ell-1}}{\pi_{t+1}} \right), \label{eq:bond}
\end{align}
where $Q^0_{t+1} = 1$. The $\ell$-period yield is $i_{t,\ell} = -\log(Q_{t,\ell})/\ell$. \citet{Ruge-Murcia2016} estimates a similar model with skewness but in a stationary economy.
\paragraph{Solution method:} Accurate approximations require using as much information from $f_1,f_2$ as possible. Perturbation methods of order $\ell$ only use the first $\ell$ moments of $f_1,f_2$. Value function iteration is too computationally costly. Projection methods are not sufficiently stable for estimation. Taylor projection \citep{levintal2018} provides a good compromise between perturbation and projection as shown in \citet{fernandez2018}. This appears to be the first application of Taylor projection for estimation. Besides solving in logs rather than levels, the equations above should be normalized to ensure the solution is stable and accurate, e.g. $1=\mathbb{E}_t\left( M_{t+1} Q_{t+1}^{\ell-1}/[Q_{t}^{\ell}\pi_{t+1}] \right)$ for (\ref{eq:bond}).

The model is non-stationary since all variables, except inflation and yields, are driven by a unit-root. It is solved in terms of de-trended variables $\tilde C_t = C_t/Z_{t-1},\tilde K_t = K_t/Z_{t-1},\tilde I_t = I_t/Z_{t-1}, \tilde Z_t = \exp(\lambda + e_{1,t})$. Growth rates are computed as $\Delta \log(C_{t+1}) = \Delta \log(\tilde C_{t+1}) + \Delta \log(Z_t)$. Assumption \ref{ass:DGPMixt} \textit{y(i),u(i)} holds for $\tilde Z_t$ and $\log\pi_t$ if $|\rho| \leq \bar{\rho} <1$. Pruning is then used to stabilize the remaining variables. 

\paragraph{Data:} The data consist of $n=235$ observations for quarterly growth rate of consumption (personal expenditure in services plus durables),   investment growth (private non-residential fixed), quarterly inflation (growth rate of GDP deflator) and three/six-month Treasury yields between 1961Q2 and 2019Q4, all taken from the FRED database. One and two-year yields are constructed from the Federal Reserve's daily nominal yield curve database. \citet{VanBinsbergen2012} also use yields at longer horizons.

\paragraph{Estimation results:} Several parameters are calibrated $\lambda = 0.0045$, $\delta = 0.0294$, $\alpha = 0.3$. \citet{VanBinsbergen2012} also calibrate $\rho = 0.955$ which is estimated here. Sample and simulated consumption and investment growth are de-meaned to remove the effect of the calibration on the levels. The model is estimated with $S=10$ for $k= 1$ (Gaussian), 2, 3, 4 and 5. Preliminary estimates are computed by first-order projection, which are then added to the initial swarm matrix to compute the final estimates with second-order projection. 

Table \ref{tab:estimates} below reports the estimates for $\theta$. The main pattern of interest is the decline of $\hat\gamma_n$ with $k$. This is reminiscent of the long-run risks and rare disasters literature which essentially find that a better representation of risk allows to match the asset prices with lower levels of risk aversion. Empirically however, \citet{backus2011} find that rare disasters are not large and frequent enough in the data to solve the equity premium puzzle. Here, the focus is on business cycle frequency risks, with a sample that excludes world wars and the great depression but still includes several recessions and inflationary events. The interesting finding is that these risks accomodate much lower levels of relative risk aversion \textit{in an estimation setting}.  Standard errors also decrease with $k$  since the objective has more curvature for smaller $\gamma$. As in \citet{VanBinsbergen2012}, the model is very hard to estimate with Gaussian shocks. Here, several coefficients are close to the optimization bounds (lb, ub). The choice of $k=4$ seems to best balance bias and variance: estimates are similar with $k=5$ but standard errors are greater. In a 12 core cluster environment, estimation takes 14h15m, 8h34m, 7h24m, 7h19m and 6h2m for $k=1,\dots,5$ respectively. The main bottleneck is in solving the model. Taylor projection is initialized with a third-order perturbation, which is more accurate for smaller $\gamma$. For $\gamma \geq 40$ and some corner values, the default solver may fail to converge, using exceptions to switch for a slower more robust solver after a failed convergence works but makes estimation very time-consuming. This mostly affects $k=1$ for which $\hat\theta_n$ is closer to the bounds.

The estimated $\widehat{\text{IES}} = 1/\hat\psi_n$ is greater than $5$ in all specifications. The null hypothesis of a CRRA utility is rejected, $H_0: \gamma = \psi$, with t-statistics of 2.2, 4.1, 3.8, 3.8 and 3.7 for $k=1,\dots,5$ respectively. For reference, in their calibration \citet{Yaron2016} favour $\text{IES}=1.5$. Using aggregate consumption, \citet{Chen2013} report a confidence interval ranging from $2$ to $5$. \citet{VanBinsbergen2012} estimate $\text{IES}=1.7$ with a very large  $\hat\gamma_n = 66$ and a small $\hat\tau_n = 0.1$. On the latter, they exclude investment from the estimation and report a poor fit in that dimension. Here, $k=1$ estimates a small $1/\hat\tau_n=1e-3$, $\text{se}(1/\hat\tau_n) = 0.035$ not significantly different from zero.\footnote{The model is solved in terms of $1/\tau$ making these quantities readily available.} The delta-method used to produce Table \ref{tab:estimates} is invalid at $1/\tau=0$, using the continuous mapping theorem to a CI for $1/\tau$ yields a more robust CI for $\tau$ itself: $[128,+\infty)$. For $k=2$, $1/\hat\tau_n=0.0198$, $\text{se}(1/\hat\tau_n) = 0.0076$ is statistically different from zero at the 1\% significance level and less problematic. Table \ref{tab:CMT} in the Supplement provides additional results for $1/\hat\tau_n$ as well as $1/\hat\gamma_n$.
\begin{table}[h] \setlength\tabcolsep{4.5pt} \renewcommand{\arraystretch}{0.935}
  \begin{center} \caption{Production Economy: Parameter Estimates} \label{tab:estimates} {
    \small
    \begin{tabular}{l|ccccccccc|c} \hline \hline
      $k$ & $\beta$ & $\gamma$ & $\psi$ & $\tau$ & $\rho$ & $\nu_\pi$  & $\nu_{z,1}$ & $\nu_{z,2}$ & $\log(\bar\pi)$ & $\hat Q_n^S(\hat\param_n)$\\ \hline

      \multirow{2}{*}{$1$} & 0.994 & 34.629 & 0.145 & 1000.000 & 0.985 & -0.990 & -0.100 & 0.001 & 0.005 & \multirow{2}{*}{3.58} \\ & (0.001) & (15.373) & (0.007) & (3465.036) & (0.004) & (0.051) & (0.012) & (0.011) & (0.003) & \\ \hline

      \multirow{2}{*}{$2$} & 0.972 & 20.153 & 0.192 & 50.526 & 0.766 & -0.152 & -0.127 & -0.015 & 0.007 & \multirow{2}{*}{2.73} \\ & (0.014) & (4.939) & (0.010) & (19.425) & (0.036) & (0.049) & (0.014) & (0.014) & (0.002)  & \\ \hline

      \multirow{2}{*}{$3$} & 0.988 & 12.754 & 0.189 & 54.559 & 0.774 & -0.059 & -0.099 & -0.010 & 0.007 & \multirow{2}{*}{2.70} \\ & (0.007) & (3.339) & (0.011) & (18.263) & (0.029) & (0.046) & (0.015) & (0.014) & (0.002) \\ \hline

      \multirow{2}{*}{$4$} & 0.992 & 10.476 & 0.165 & 51.906 & 0.711 & 0.005 & -0.112 & -0.004 & 0.007 & \multirow{2}{*}{2.58} \\ & (0.006) & (2.692) & (0.012) & (13.748) & (0.039) & (0.048) & (0.014) & (0.014) & (0.002)&\\ \hline

      \multirow{2}{*}{$5$} & 0.991 & 11.970 & 0.170 & 66.961 & 0.694 & 0.139 & -0.085 & -0.016 & 0.007& \multirow{2}{*}{2.46}\\ & (0.007) & (3.147) & (0.011) & (24.847) & (0.034) & (0.045) & (0.012) & (0.012) & (0.002)&\\ \hline\hline
      lb & 0.965 & 0.5 & 0.05 & 0.01 & 0.2 & -0.99 &-0.25 &-0.25 & 0.005 & - \\
      ub & 0.999 & 70  & 110 & 1000  & 0.995  & 0.2   & 0.25 & 0.25 & 0.0085 & -\\ \hline\hline
    \end{tabular} }
  \end{center}
\end{table}
Table \ref{tab:moms} below compares selected sample with simulated moments. The fit is generally better with larger $k$. To better understand the estimated IES, the last row changes the IES to $1.5$, keeping the other coefficients at the $k=4$ estimates. The smaller IES increases average yields but reduces the slope of the yield curve and the variance of consumption. The correlation between consumption growth and yields is slightly positive in the data but very negative for $k=1$ and $\text{IES}=1.5$. For larger $k$, these correlations are closer to the sample.
\begin{table}[h] \setlength\tabcolsep{4.5pt} \renewcommand{\arraystretch}{0.935}
  \begin{center} \caption{Production Economy: Sample and Simulated Moments} \label{tab:moms} {
    \small
    \begin{tabular}{l|cccc|cccccc|cccc} \hline \hline
      & \multicolumn{4}{c|}{Average Yield} & \multicolumn{6}{c|}{Standard Deviations} & \multicolumn{4}{c}{Corr($\Delta c_t$,Yield)} \\
       & $3m$ & $6m$ & $1y$ & $2y$ & $3m$ & $6m$ & $1y$ & $2y$ & $\Delta c_t$ & $\Delta i_t$  & $3m$ & $6m$ & $1y$ & $2y$ \\ \hline
      Sample & 4.57 & 4.71 & 5.03 & 5.24 & 3.19 & 3.18 & 3.31 & 3.25 & 0.57 & 1.96 & \,\,0.07 & \,\,0.07 & \,\,0.08 & \,\,0.09 \\ \hline 
      Gaussian & 2.80 & 2.87 & 2.99 & 3.05 & 2.44 & 2.35 & 2.31 & 2.26 & 0.53 & 1.66 & -0.25 & -0.25 & -0.25 & -0.25 \\ 
      $k=2$ & 3.59 & 3.64 & 3.73 & 3.85 & 1.66 & 1.43 & 1.11 & 0.73 & 0.83 & 1.23 & -0.11 & -0.07 & -0.04 & -0.02 \\ 
      $k=3$ & 3.74 & 3.76 & 3.80 & 3.85 & 1.96 & 1.72 & 1.37 & 0.92 & 0.68 & 1.48 & -0.06 & -0.03 & -0.01 & \,\,0.01 \\ 
      $k=4$ & 3.88 & 3.88 & 3.89 & 3.90 & 1.90 & 1.63 & 1.22 & 0.76 & 0.69 & 1.42 & \,\,0.01 & \,\,0.03 & \,\,0.05 & \,\,0.07 \\ 
      $k=5$ & 3.81 & 3.83 & 3.86 & 3.89 & 1.94 & 1.70 & 1.28 & 0.79 & 0.66 & 1.41 & -0.02 & \,\,0.02 & \,\,0.04 & \,\,0.05 \\ \hline 
      $\text{IES}=1.5$ & 5.08 & 5.08 & 5.08 & 5.08 & 1.90 & 1.63 & 1.23 & 0.77 & 0.34 & 2.04 & -0.16 & -0.17 & -0.15 & -0.08 \\ \hline\hline
    \end{tabular}\\ }
    {\footnotesize \textit{Note: yields (3m, 6m, 1y, 2y) are annualized, growth rates $\Delta c_t,\Delta i_t$ are \%QoQ.}}
  \end{center}
\end{table}
Figure \ref{fig:est} compares sample with simulated distributions and shows the estimated densities $f_1,f_2$. The latter two are normalized to have variance equal to $1$. For $k=4$, estimates of $f_1,f_2$ have skewnesses of $-1.45,2.18$ and kurtoses of $6.57,11.48$ which indicate excess downside risks for technology shocks and upwards risks for inflation. While mixtures improve the fit for consumption and inflation, investment and 3m yields are more challenging to match. In the sample, investment has smaller kurtosis than consumption, $5$ and $9$ respectively, while in simulations, the converse is true: investment has larger kurtosis than consumption, $5$ and $4$ with $k=4$. In the model, investment is the only source of endogenous variation for output, adding labor would provide another. Also, varying capital utilization could provide more realistic fluctuations in output and investment \citep{king1999}. 3m yields were above 10\% only between 1979Q3 and 1984Q3 and below 0.5\% only between 2008Q4 and 2016Q3, i.e. both tails are associated with specific monetary policy regimes. This suggests that modelling monetary policy regimes is needed to improve the fit of yields in the tails.
\begin{figure}[h]
  \begin{center} \caption{Production Economy: Sample and Simulated Distributions, Density Estimates}\label{fig:est}
  \includegraphics[scale=0.52]{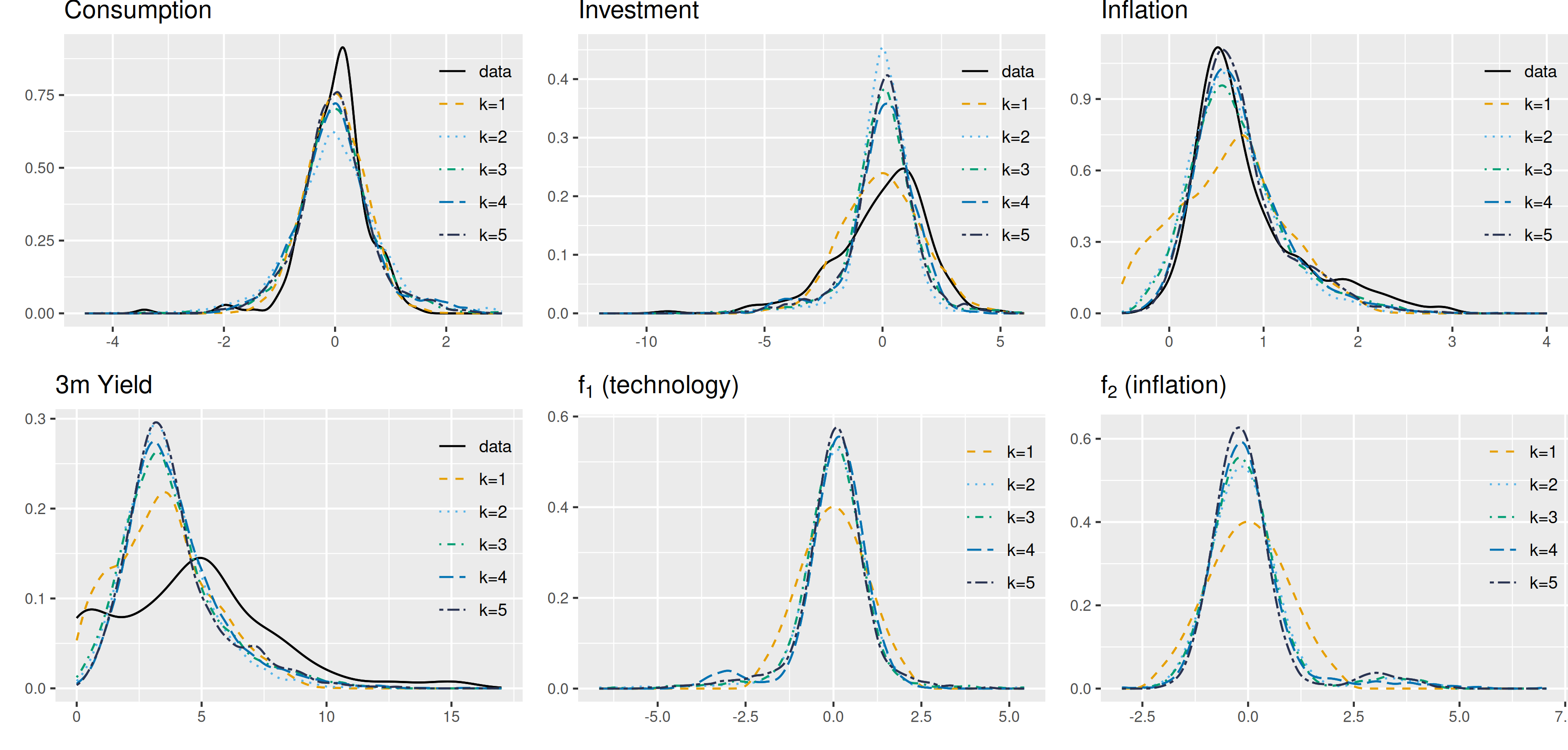}\\
   {\footnotesize \textit{Note: 3m yields are annualized, growth rates $g_c,g_i$ are \%QoQ and de-meaned, inflation is \%QoQ.}} \end{center}
\end{figure}
There are two main takeways from this application. First, allowing for a flexible distribution in the shocks $(e_{1,t},e_{2,t})$ allows to better capture risks and leads to much smaller estimates of relative risk aversion. This highlights the empirical relevance of using a semi-nonparametric approach in this setting.  Second, the model is very simple and has limitations that show in the results. It cannot match the variance of consumption without a large IES, as shown in Table \ref{tab:moms}. Overall, the flexible estimation fits the data better in some dimensions using more reasonable parameters values that also seem to be more accurately estimated. However, the flexible distribution does not improve the fit in all dimensions and issues remain such as the zero lower bound on interest rates, the joint dynamics of consumption and investment, among others. Going forward, estimating the distribution of the shocks in a more realistic model that can capture these feature would be of interest.


\section{Conclusion} \label{sec:conclusion}
Simulation-based estimation is a powerful approach to estimate intractable models. Using a mixture sieve with the empirical characteristic function, this paper provides an approach to estimate semi-nonparametric models by simulation. Estimation using the ECF can be unstable depending on the choice of the weight function $\pi$, see e.g. \citet{chen2019} section 2.1.1 for a discussion. The approach suggested in Section \ref{sec:the_estimator} provides a simple way to give more or less weight to lower-order moments and then to check the fit for selected moments as in Table \ref{tab:moms}. Alternatively, the conditional cdf or pdf can be used as moments. 
Approximation results in \citet{DeJonge2010}, \citet{Norets2010} can be used to consider joint or conditional densities. Estimating other objects nonparametrically, such as a utility or production function, can also be of interest. Another direction of research would be to develop general theory for sieve indirect inference estimation.

\newpage

\baselineskip=12.0pt
\bibliography{refs.bib,refsABC.bib}

\begin{thebibliography}{57}
\newcommand{\enquote}[1]{``#1''}
\expandafter\ifx\csname natexlab\endcsname\relax\def\natexlab#1{#1}\fi

\bibitem[\protect\citeauthoryear{Ai and Chen}{Ai and Chen}{2003}]{Ai2003}
\textsc{Ai, C. and X.~Chen} (2003): \enquote{{Efficient Estimation of Models
  with Conditional Moment Restrictions Containing Unknown Functions},}
  \emph{Econometrica}, 71, 1795--1843.

\bibitem[\protect\citeauthoryear{Ai and Chen}{Ai and Chen}{2007}]{AI20075}
---\hspace{-.1pt}---\hspace{-.1pt}--- (2007): \enquote{Estimation of possibly
  misspecified semiparametric conditional moment restriction models with
  different conditioning variables,} \emph{Journal of Econometrics}, 141, 5 --
  43, semiparametric methods in econometrics.

\bibitem[\protect\citeauthoryear{Andrews and Pollard}{Andrews and
  Pollard}{1994}]{Andrews1994}
\textsc{Andrews, D. W.~K. and D.~Pollard} (1994): \enquote{{An Introduction to
  Functional Central Limit Theorems for Dependent Stochastic Processes},}
  \emph{International Statistical Review / Revue Internationale de
  Statistique}, 62, 119.

\bibitem[\protect\citeauthoryear{Backus, Chernov, and Martin}{Backus
  et~al.}{2011}]{backus2011}
\textsc{Backus, D., M.~Chernov, and I.~Martin} (2011): \enquote{Disasters
  implied by equity index options,} \emph{The journal of finance}, 66,
  1969--2012.

\bibitem[\protect\citeauthoryear{Bansal and Yaron}{Bansal and
  Yaron}{2004}]{Yaron2016}
\textsc{Bansal, R. and A.~Yaron} (2004): \enquote{{Risks for the Long Run: A
  Potential Resolution of Asset Pricing Puzzles},} \emph{The Journal of
  Finance}, 59, 1481--1509.

\bibitem[\protect\citeauthoryear{{Ben Hariz}}{{Ben Hariz}}{2005}]{BenHariz2005}
\textsc{{Ben Hariz}, S.} (2005): \enquote{{Uniform CLT for empirical process},}
  \emph{Stochastic Processes and their Applications}, 115, 339--358.

\bibitem[\protect\citeauthoryear{Bierens and Song}{Bierens and
  Song}{2012}]{Bierens2012}
\textsc{Bierens, H.~J. and H.~Song} (2012): \enquote{{Semi-nonparametric
  estimation of independently and identically repeated first-price auctions via
  an integrated simulated moments method},} \emph{Journal of Econometrics},
  168, 108--119.

\bibitem[\protect\citeauthoryear{Blasques}{Blasques}{2011}]{Blasques2011}
\textsc{Blasques, F.} (2011): \enquote{{Semi-Nonparametric Indirect
  Inference},} \emph{PhD Thesis, Maastricht University}, 1--221.

\bibitem[\protect\citeauthoryear{Carrasco, Chernov, Florens, and
  Ghysels}{Carrasco et~al.}{2007}]{Carrasco2007}
\textsc{Carrasco, M., M.~Chernov, J.-P. Florens, and E.~Ghysels} (2007):
  \enquote{{Efficient estimation of general dynamic models with a continuum of
  moment conditions},} \emph{Journal of Econometrics}, 140, 529--573.

\bibitem[\protect\citeauthoryear{Carrasco and Florens}{Carrasco and
  Florens}{2000}]{Carrasco2000}
\textsc{Carrasco, M. and J.-P. Florens} (2000): \enquote{{Generalization of GMM
  to a Continuum of Moment Conditions},} \emph{Econometric Theory}, 16,
  797--834.

\bibitem[\protect\citeauthoryear{Chen}{Chen}{2007}]{Chen2007}
\textsc{Chen, X.} (2007): \enquote{{Chapter 76 Large Sample Sieve Estimation of
  Semi-Nonparametric Models},} in \emph{Handbook of Econometrics}, vol.~6,
  5549--5632.

\bibitem[\protect\citeauthoryear{Chen, Favilukis, and Ludvigson}{Chen
  et~al.}{2013}]{Chen2013}
\textsc{Chen, X., J.~Favilukis, and S.~C. Ludvigson} (2013): \enquote{{An
  estimation of economic models with recursive preferences},}
  \emph{Quantitative Economics}, 4, 39--83.

\bibitem[\protect\citeauthoryear{Chen and Liao}{Chen and
  Liao}{2015}]{Chen2015b}
\textsc{Chen, X. and Z.~Liao} (2015): \enquote{{Sieve semiparametric two-step
  GMM under weak dependence},} \emph{Journal of Econometrics}, 189, 163--186.

\bibitem[\protect\citeauthoryear{Chen, Linton, Schneeberger, and Yi}{Chen
  et~al.}{2019}]{chen2019}
\textsc{Chen, X., O.~Linton, S.~Schneeberger, and Y.~Yi} (2019):
  \enquote{Semiparametric estimation of the bid--ask spread in extended roll
  models,} \emph{Journal of econometrics}, 208, 160--178.

\bibitem[\protect\citeauthoryear{Chen, Linton, and {Van Keilegom}}{Chen
  et~al.}{2003}]{Chen2003}
\textsc{Chen, X., O.~Linton, and I.~{Van Keilegom}} (2003):
  \enquote{{Estimation of Semiparametric Models when the Criterion Function Is
  Not Smooth},} \emph{Econometrica}, 71, 1591--1608.

\bibitem[\protect\citeauthoryear{Chen and Ludvigson}{Chen and
  Ludvigson}{2009}]{Chen2009}
\textsc{Chen, X. and S.~C. Ludvigson} (2009): \enquote{{Land of addicts? An
  empirical investigation of habit-based asset pricing models},} \emph{Journal
  of Applied Econometrics}, 24, 1057--1093.

\bibitem[\protect\citeauthoryear{Chen and Pouzo}{Chen and
  Pouzo}{2012}]{Chen2012}
\textsc{Chen, X. and D.~Pouzo} (2012): \enquote{{Estimation of Nonparametric
  Conditional Moment Models With Possibly Nonsmooth Generalized Residuals},}
  \emph{Econometrica}, 80, 277--321.

\bibitem[\protect\citeauthoryear{Chen and Pouzo}{Chen and
  Pouzo}{2015{\natexlab{a}}}]{Chen2015a}
---\hspace{-.1pt}---\hspace{-.1pt}--- (2015{\natexlab{a}}): \enquote{{Sieve
  Wald and QLR Inferences on Semi/Nonparametric Conditional Moment Models},}
  \emph{Econometrica}, 83, 1013--1079.

\bibitem[\protect\citeauthoryear{Chen and Pouzo}{Chen and
  Pouzo}{2015{\natexlab{b}}}]{Chen2015aS}
---\hspace{-.1pt}---\hspace{-.1pt}--- (2015{\natexlab{b}}):
  \enquote{{Supplement to "Sieve Wald and QLR Inferences on Semi/Nonparametric
  Conditional Moment Models"},} \emph{Econometrica}, 83, 1013--1079.

\bibitem[\protect\citeauthoryear{Chen and Shen}{Chen and
  Shen}{1998}]{Chen1998a}
\textsc{Chen, X. and X.~Shen} (1998): \enquote{{Sieve Extremum Estimates for
  Weakly Dependent Data},} \emph{Econometrica}, 66, 289.

\bibitem[\protect\citeauthoryear{Cho, Cooley, and Kim}{Cho
  et~al.}{2015}]{Cho2015}
\textsc{Cho, J.-O., T.~F. Cooley, and H.~S.~E. Kim} (2015): \enquote{{Business
  cycle uncertainty and economic welfare},} \emph{Review of Economic Dynamics},
  18, 185--200.

\bibitem[\protect\citeauthoryear{Christensen}{Christensen}{2017}]{Christensen2017}
\textsc{Christensen, T.~M.} (2017): \enquote{Nonparametric stochastic discount
  factor decomposition,} \emph{Econometrica}, 85, 1501--1536.

\bibitem[\protect\citeauthoryear{Cline and Pu}{Cline and Pu}{1999}]{cline1999}
\textsc{Cline, D.~B. and H.-m.~H. Pu} (1999): \enquote{Geometric ergodicity of
  nonlinear time series,} \emph{Statistica Sinica}, 1103--1118.

\bibitem[\protect\citeauthoryear{Davydov}{Davydov}{1968}]{Davydov1968}
\textsc{Davydov, Y.~A.} (1968): \enquote{{Convergence of Distributions
  Generated by Stationary Stochastic Processes},} \emph{Theory of Probability
  {\&} Its Applications}, 13, 691--696.

\bibitem[\protect\citeauthoryear{{De Jonge} and {Van Zanten}}{{De Jonge} and
  {Van Zanten}}{2010}]{DeJonge2010}
\textsc{{De Jonge}, R. and J.~H. {Van Zanten}} (2010): \enquote{{Adaptive
  nonparametric Bayesian inference using location-scale mixture priors},}
  \emph{Annals of Statistics}, 38, 3300--3320.

\bibitem[\protect\citeauthoryear{Doukhan, Massart, and Rio}{Doukhan
  et~al.}{1995}]{DoukhanP.MassartP.1995}
\textsc{Doukhan, P., P.~Massart, and E.~Rio} (1995): \enquote{{Invariance
  principles for absolutely regular empirical processes},} \emph{Annales de
  l'Institut Henri Poincar{\'{e}}, section B}, tome 31, 393--427.

\bibitem[\protect\citeauthoryear{Dridi and Renault}{Dridi and
  Renault}{2000}]{Dridi2000}
\textsc{Dridi, R. and E.~Renault} (2000): \enquote{{Semi-parametric indirect
  inference},} Tech. rep., Suntory and Toyota International Centres for
  Economics and Related Disciplines, LSE.

\bibitem[\protect\citeauthoryear{Duffie and Singleton}{Duffie and
  Singleton}{1993}]{Duffie1993}
\textsc{Duffie, D. and K.~J. Singleton} (1993): \enquote{{Simulated Moments
  Estimation of Markov Models of Asset Prices},} \emph{Econometrica}, 61, 929.

\bibitem[\protect\citeauthoryear{Fermanian and Salani{\'{e}}}{Fermanian and
  Salani{\'{e}}}{2004}]{Fermanian2004}
\textsc{Fermanian, J.-D. and B.~Salani{\'{e}}} (2004): \enquote{{A
  Nonparametric Simulated Maximum Likelihood Estimation Method},}
  \emph{Econometric Theory}, 20, 701--734.

\bibitem[\protect\citeauthoryear{Fern{\'a}ndez-Villaverde and
  Levintal}{Fern{\'a}ndez-Villaverde and Levintal}{2018}]{fernandez2018}
\textsc{Fern{\'a}ndez-Villaverde, J. and O.~Levintal} (2018): \enquote{Solution
  methods for models with rare disasters,} \emph{Quantitative Economics}, 9,
  903--944.

\bibitem[\protect\citeauthoryear{Gallant and Nychka}{Gallant and
  Nychka}{1987}]{Gallant1987}
\textsc{Gallant, a.~R. and D.~W. Nychka} (1987): \enquote{{Semi-Nonparametric
  Maximum Likelihood Estimation},} \emph{Econometrica}, 55, 363--390.

\bibitem[\protect\citeauthoryear{Gallant and Tauchen}{Gallant and
  Tauchen}{1996}]{Gallant1996}
\textsc{Gallant, a.~R. and G.~Tauchen} (1996): \enquote{{Which Moments to
  Match?}} \emph{Econometric Theory}, 12, 657.

\bibitem[\protect\citeauthoryear{Gospodinov and Ng}{Gospodinov and
  Ng}{2015}]{Gospodinov2013}
\textsc{Gospodinov, N. and S.~Ng} (2015): \enquote{{Minimum Distance Estimation
  of Possibly Noninvertible Moving Average Models},} \emph{Journal of Business
  {\&} Economic Statistics}, 33, 403--417.

\bibitem[\protect\citeauthoryear{Gouri{\'{e}}roux, Monfort, and
  Renault}{Gouri{\'{e}}roux et~al.}{1993}]{Gourieroux1993}
\textsc{Gouri{\'{e}}roux, C., A.~Monfort, and E.~Renault} (1993):
  \enquote{{Indirect inference},} \emph{Journal of Applied Econometrics}, 8,
  S85----S118.

\bibitem[\protect\citeauthoryear{Gourio}{Gourio}{2012}]{Gourio2012}
\textsc{Gourio, F.} (2012): \enquote{Disaster risk and business cycles,}
  \emph{American Economic Review}, 102, 2734--66.

\bibitem[\protect\citeauthoryear{Hansen and Richard}{Hansen and
  Richard}{1987}]{Hansen1987}
\textsc{Hansen, L.~P. and S.~F. Richard} (1987): \enquote{{The Role of
  Conditioning Information in Deducing Testable Restrictions Implied by Dynamic
  Asset Pricing Models},} \emph{Econometrica}, 55, 587--613.

\bibitem[\protect\citeauthoryear{Jermann}{Jermann}{1998}]{jermann1998}
\textsc{Jermann, U.~J.} (1998): \enquote{Asset pricing in production
  economies,} \emph{Journal of monetary Economics}, 41, 257--275.

\bibitem[\protect\citeauthoryear{King and Rebelo}{King and
  Rebelo}{1999}]{king1999}
\textsc{King, R.~G. and S.~T. Rebelo} (1999): \enquote{Resuscitating real
  business cycles,} \emph{Handbook of macroeconomics}, 1, 927--1007.

\bibitem[\protect\citeauthoryear{Kolmogorov and Tikhomirov}{Kolmogorov and
  Tikhomirov}{1959}]{kolmogorov1959varepsilon}
\textsc{Kolmogorov, A.~N. and V.~M. Tikhomirov} (1959):
  \enquote{$\varepsilon$-entropy and $\varepsilon$-capacity of sets in function
  spaces,} \emph{Uspekhi Matematicheskikh Nauk}, 14, 3--86.

\bibitem[\protect\citeauthoryear{Kruijer, Rousseau, and van~der Vaart}{Kruijer
  et~al.}{2010}]{Kruijer2010}
\textsc{Kruijer, W., J.~Rousseau, and A.~van~der Vaart} (2010):
  \enquote{{Adaptive Bayesian density estimation with location-scale
  mixtures},} \emph{Electronic Journal of Statistics}, 4, 1225--1257.

\bibitem[\protect\citeauthoryear{Levintal}{Levintal}{2018}]{levintal2018}
\textsc{Levintal, O.} (2018): \enquote{Taylor projection: A new solution method
  for dynamic general equilibrium models,} \emph{International Economic
  Review}, 59, 1345--1373.

\bibitem[\protect\citeauthoryear{Liebscher}{Liebscher}{2005}]{Liebscher2005}
\textsc{Liebscher, E.} (2005): \enquote{{Towards a Unified Approach for Proving
  Geometric Ergodicity and Mixing Properties of Nonlinear Autoregressive
  Processes},} \emph{Journal of Time Series Analysis}, 26, 669--689.

\bibitem[\protect\citeauthoryear{McFadden}{McFadden}{1989}]{McFadden1989}
\textsc{McFadden, D.} (1989): \enquote{{A Method of Simulated Moments for
  Estimation of Discrete Response Models Without Numerical Integration},}
  \emph{Econometrica}, 57, 995.

\bibitem[\protect\citeauthoryear{Newey}{Newey}{2001}]{Newey2001}
\textsc{Newey, W.~K.} (2001): \enquote{{Flexible Simulated Moment Estimation of
  Nonlinear Errors-in-Variables Models},} \emph{Review of Economics and
  Statistics}, 83, 616--627.

\bibitem[\protect\citeauthoryear{Norets}{Norets}{2010}]{Norets2010}
\textsc{Norets, A.} (2010): \enquote{{Approximation of conditional densities by
  smooth mixtures of regressions},} \emph{Annals of Statistics}, 38,
  1733--1766.

\bibitem[\protect\citeauthoryear{Norets and Tang}{Norets and
  Tang}{2014}]{norets2014a}
\textsc{Norets, A. and X.~Tang} (2014): \enquote{Semiparametric inference in
  dynamic binary choice models,} \emph{Review of Economic Studies}, 81,
  1229--1262.

\bibitem[\protect\citeauthoryear{Pakes and Pollard}{Pakes and
  Pollard}{1989}]{Pakes1989}
\textsc{Pakes, A. and D.~Pollard} (1989): \enquote{{Simulation and the
  Asymptotics of Optimization Estimators},} \emph{Econometrica}, 57, 1027.

\bibitem[\protect\citeauthoryear{Pisier}{Pisier}{1983}]{Pisier1983}
\textsc{Pisier, G.} (1983): \enquote{{Some applications of the metric entropy
  condition to harmonic analysis},} in \emph{Banach spaces, Harmonic analysis
  and Probability, Univ. of Connecticut 1980-81. Lecture Notes in Mathematics,
  995}, ed. by R.~C. Blei and S.~J. Sidney, Berlin, Heidelberg: Springer Berlin
  Heidelberg, 123--154.

\bibitem[\protect\citeauthoryear{Rio}{Rio}{2000}]{rio2000}
\textsc{Rio, E.} (2000): \emph{{Th{\'{e}}orie Asymptotique des Processus
  Al{\'{e}}atoires Faiblement D{\'{e}}pendants}}, vol.~31 of
  \emph{Math{\'{e}}matiques et Applications}, Springer Berlin Heidelberg.

\bibitem[\protect\citeauthoryear{Rudebusch and Swanson}{Rudebusch and
  Swanson}{2012}]{rudebusch2012}
\textsc{Rudebusch, G.~D. and E.~T. Swanson} (2012): \enquote{The bond premium
  in a DSGE model with long-run real and nominal risks,} \emph{American
  Economic Journal: Macroeconomics}, 4, 105--43.

\bibitem[\protect\citeauthoryear{Ruge-Murcia}{Ruge-Murcia}{2017}]{Ruge-Murcia2016}
\textsc{Ruge-Murcia, F.} (2017): \enquote{{Skewness Risk and Bond Prices},}
  \emph{Journal of Applied Econometrics}, 32, 379--400.

\bibitem[\protect\citeauthoryear{Santos}{Santos}{2012}]{Santos2012}
\textsc{Santos, A.} (2012): \enquote{{Inference in Nonparametric Instrumental
  Variables With Partial Identification},} \emph{Econometrica}, 80, 213--275.

\bibitem[\protect\citeauthoryear{{Van Binsbergen}, Fern\'andez-Villaverde,
  Koijen, and Rubio-Ram\'irez}{{Van Binsbergen}
  et~al.}{2012}]{VanBinsbergen2012}
\textsc{{Van Binsbergen}, J.~H., J.~Fern\'andez-Villaverde, R.~S.~J. Koijen,
  and J.~Rubio-Ram\'irez} (2012): \enquote{{The term structure of interest
  rates in a DSGE model with recursive preferences},} \emph{Journal of Monetary
  Economics}, 59, 634--648.

\bibitem[\protect\citeauthoryear{van~der Vaart and Wellner}{van~der Vaart and
  Wellner}{1996}]{VanderVaart1996}
\textsc{van~der Vaart, A.~W. and J.~A. Wellner} (1996): \emph{{Weak Convergence
  and Empirical Processes}}, Springer Series in Statistics, New York, NY:
  Springer New York.

\bibitem[\protect\citeauthoryear{Weil}{Weil}{1989}]{weil1989}
\textsc{Weil, P.} (1989): \enquote{The equity premium puzzle and the risk-free
  rate puzzle,} \emph{Journal of Monetary Economics}, 24, 401 -- 421.

\bibitem[\protect\citeauthoryear{Wooldridge and White}{Wooldridge and
  White}{1988}]{Wooldridge1988}
\textsc{Wooldridge, J.~M. and H.~White} (1988): \enquote{{Some Invariance
  Principles and Central Limit Theorems for Dependent Heterogeneous
  Processes},} \emph{Econometric Theory}, 4, 210--230.

\bibitem[\protect\citeauthoryear{Yu}{Yu}{2004}]{yu2004}
\textsc{Yu, J.} (2004): \enquote{Empirical characteristic function estimation
  and its applications,} \emph{Econometric reviews}, 23, 93--123.

\end{thebibliography}
\baselineskip=18.0pt
\newpage
\begin{appendices}
  \renewcommand\thetable{\thesection\arabic{table}}
  \renewcommand\thefigure{\thesection\arabic{figure}}
  \renewcommand{\theequation}{\thesection.\arabic{equation}}
  \renewcommand\thelemma{\thesection\arabic{lemma}}
  \renewcommand\thetheorem{\thesection\arabic{theorem}}
  \renewcommand\thedefinition{\thesection\arabic{definition}}
    \renewcommand\theassumption{\thesection\arabic{assumption}}
  \renewcommand\theproposition{\thesection\arabic{proposition}}
    \renewcommand\theremark{\thesection\arabic{remark}}
    \renewcommand\thecorollary{\thesection\arabic{corollary}}

\setcounter{lemma}{0} 
\section{Preliminary Results} \label{apx:prelim}
\begin{lemma}[Approximation Properties of the Gaussian and Tails Mixture] \label{lem:ApproxSimuGAUT}
      Suppose that the shocks $e=(e_{t,1},\dots,e_{t,d_e})$ are independent with density $f = f_1 \times \dots \times f_{d_e}$. Suppose that each marginal $f_j$ can be decomposed into a smooth density $f_{j,S}$ and the two tails density $f_L,f_R$:
      \[ f_j = (1-\omega_{j,1}-\omega_{j,2})f_{j,S} + \omega_{j,1}f_L + \omega_{j,2}f_R. \]
      Let each $f_{j,S}$ satisfy the assumptions of \citet{Kruijer2010}: i) Smoothness: $f_{j,S}$ is $r$-times continuously differentiable with bounded $r$-th derivative. ii) Tails: $f_{j,S}$ has exponential tails, i.e. there exists $\bar{e}, M_{f}, a,b>0$ such that $f_{j,S}(e) \leq M_{f}e^{-a|e|^b},\, \forall |e| \geq \bar{e}.$ iii) Monotonicity in the Tails: $f_{j,S}$ is strictly positive and there exists $\underline{e}<\overline{e}$ such that $f_{j,S}$ is weakly decreasing on $(-\infty,\underline{e}]$ and weakly increasing on $[\overline{e},\infty)$ and $\|f_j\|_\infty \leq \overline{f}$ for all $j$. Then there exists a Gaussian and tails mixture $\Pi_{k}f = \Pi_{k}f_1 \times \dots \times \Pi_{k}f_{d_e}$ satisfying the restrictions of \citet{Kruijer2010}: iv) Bandwidth: $\sigma_j \geq \underline{\sigma}_{k} = O(\frac{\log[k]^{2/b}}{k})$. v) Location Parameter Bounds: $\mu_j \in [-\bar{\mu}_{k},\bar{\mu}_{k}]$ with $\bar{\mu}_{k} = O\left(\log[k]^{1/b} \right)$ such that as $k\to\infty$:
      \begin{align*}
        &\|f-\Pi_{k} f\|_{\mathcal{F}} = O \left(  \frac{\log[k]^{2r/b}}{k^r} \right)
      \end{align*}
      where $\|\cdot\|_\mathcal{F} = \|\cdot\|_{TV}$ or $\|\cdot\|_{\infty}$.  
\end{lemma}

The following Lemma is needed to verify the $L_2$-smoothness condition when using the Gaussian and tails mixture.
\begin{lemma}[Properties of the Tails Distributions] \label{lem:tailsdist}
  Let $\bar{\tail} \geq \tail_1,\tail_2 \geq \underline{\tail} > 0$. Let $\nu_{t,1}^s$ and $\nu_{t,2}^s$ be uniform $\mathcal{U}_{[0,1]}$ draws and:
  \begin{align*}
  &e_{t,1}^s = - \left(\frac{1}{\nu_{t,1}^s}-1\right)^{\frac{1}{2+\tail_1}}, \quad
  e_{t,2}^s =  \left(\frac{1}{1-\nu_{t,2}^s}-1\right)^{\frac{1}{2+\tail_2}}.
  \end{align*}
  The densities of $e_{t,1}^s,e_{t,2}^s$ satisfy $f_{e_{t,1}^s}(e) \sim e^{-3-\tail_1}$ as $e\to -\infty$, $f_{e_{t,2}^s}(e) \sim e^{-3-\tail_2}$ as $e\to +\infty$. There exists a finite $C$ bounding the second moments $\mathbb{E}\left( |e_{t,1}^s|^2 \right) \leq C < \infty$ and $\mathbb{E}\left( |e_{t,2}^s|^2 \right) \leq C < \infty$. Furthermore, the draws $y_{t,1}^s$ and $y_{t,2}^s$ are $L^2$-smooth in $\tail_1$ and $\tail_2$ respectively:
  \begin{align*}
    &\left[ \mathbb{E}\left( \sup_{|\tail_1-\tilde \tail_1| \leq \delta} |e_{t,1}^s(\tail_1)-e_{t,1}^s(\tilde \tail_1)|^2 \right) \right]^{1/2} \leq C \delta, \quad
    \left[ \mathbb{E}\left( \sup_{|\tail_2-\tilde \tail_2| \leq \delta} |e_{t,2}^s(\tail_2)-e_{t,2}^s(\tilde \tail_2)|^2 \right) \right]^{1/2} \leq C \delta
  \end{align*}
  Where the constant $C$ only depends on $\underline{\tail}$ and $\bar{\tail}$.
\end{lemma}

\begin{lemma}[Covering Numbers] \label{lem:CoveringNumbers}
  Under the $L^2$-smoothness of the DGP (as in Lemma \ref{lem:L2smooth}), the bracketing number satisfies for $x\in (0,1)$ and some $\overline{C}$:
  \begin{align*}
  &N_{[\,]}(x,\Psi_{k(n)}(\tau),\|\cdot\|_{L^2}) \\&\leq  \left( 3[k(n)+2]+d_\theta \right)  \left( 2\max(\bar\mu_{k(n)},\underline{\sigma})\overline{C}^{2/\gamma^2}\frac{\left( k(n)+\overline{\mu}_{k(n)}+\overline{\sigma}\right)^{2/\gamma}+\underline{\sigma}_{k(n)}^4 }{x^{2/\gamma^2}}+1 \right)^{3[k(n)+2]+d_\theta}.
  \end{align*}
  For $\tau \in \mathbb{R}^{d_\tau}$, let $\Psi_{k(n)}(\tau)$ be the set of functions $\Psi_{k(n)}(\tau) = \left\{ \param \rightarrow e^{i \tau^\prime (\mathbf{y}_t(\param),\mathbf{x}_t)} \pi(\tau)^{1/2},\, \param \in \paramspace_{k(n)}  \right\}$.
  The bracketing entropy of each set $\Psi_{k(n)}(\tau)$ satisfies for some $\tilde{C}$:
  \begin{align*}
  &\log \left(N_{[\,]}(x,\Psi_{k(n)}(\tau),\|\cdot\|_{L^2} \right) \leq \tilde{C} k(n)\log[k(n)])|\log \delta|.
  \end{align*}
  Using the above, for some $\tilde{C}_2 < \infty$:
  \begin{align*}
  \int_0^1 \log^2 \left(N_{[\,]}(x,\Psi_k(n),\|\cdot\|_{L^2} \right) dx \leq \tilde{C}_2  k(n)^2\log[k(n)]^2.
  \end{align*}
  \end{lemma}

\begin{lemma}[Nonparametric Approximation Bias] \label{lem:ObjApproxRate} Suppose Assumptions \ref{ass:sid} and \ref{ass:DGPMixt} (or \ref{ass:DGPmixtbis}) hold. Furthermore suppose that $\mathbb{E}\left( \|y_t^s\|^2 \right)$ and $\mathbb{E}\left(\|u_t^s\|^2\right)$ are bounded for $\param=\param_0$ and $\param=\Pi_{k(n)}\param_0$ for all $k(n)\geq 1$, $t \geq 1$ then:
  \begin{align*}
    Q_n(\Pi_{k(n)}\param_0) &= O\left( \max \left[  \frac{\log[k(n)]^{4r/b+2}}{k(n)^{2r}},\frac{\log[k(n)]^{4 \gamma^2 r/b}}{k(n)^{2\gamma^2 r}}, \frac{1}{n^2} \right]\right) = O\left(\frac{\log[k(n)]^{4r/b+2}}{k(n)^{2\gamma^2 r}}\right)
  \end{align*}
where $\Pi_{k(n)}\param_0$ is the mixture approximation of $\param_0$, $\gamma$ the H\"older coefficient in Assumption \ref{ass:DGPMixt}, $b$ and $r$ are the exponential tail index and the smoothness of the density $f_S$ in Lemma \ref{lem:ApproxSimuGAUT}.
\end{lemma}

\begin{lemma}[Convergence Rate in $\|\cdot\|_m$] \label{lem:cv_rate_mixture_norm}
  Let $\delta_n = \sqrt{(k(n)\log[k(n)])^4/n}$ and $M_n = \log\log(n+1)$. Suppose the following undersmoothing assumptions hold: i) Rate of Convergence: $\| \hat \param_n-\param_0 \|_{weak}=O_p(\delta_n)$. ii) Negligible Bias: $\| \Pi_{k(n)} \param_0-\param_0 \|_{weak}=o(\delta_n)$.
  Furthermore, suppose that the population CF is smooth in $\param$ and satisfies:
    iii) Rate 1: uniformly over $\param \in \{ \param \in \paramspace_{osn}, \|\param-\param_0\|_{weak} \leq M_n\delta_n \}$:
    $\int \big| \frac{d\mathbb{E}(\hat \psi_n^S(\tau,\param_0))}{d\param}[\param-\param_0]-\frac{d\mathbb{E}(\hat \psi_n^S(\tau,\Pi_{k(n)}\param_0))}{d\param}[\param-\param_0] \big|^2 \pi(\tau)d\tau = O(\delta_n^2).$
    iv) Rate 2: $\Pi_{k(n)}\param_0$ satisfies 
    $\int \Big| \frac{d\mathbb{E}(\hat \psi_n^S(\tau,\Pi_{k(n)}\param_0))}{d\param}[\Pi_{k(n)} \param_0-\param_0] \Big|^2 \pi(\tau)d\tau = O(\delta_n^2).$
  Suppose $\underline{\lambda}_n = \lambda_{\min}( \int \frac{d\mathbb{E}(\hat \psi_n^S(\tau,\Pi_{k(n)}\param_0))}{d(\theta,\omega,\mu,\sigma)}^\prime \overline{\frac{d\mathbb{E}(\hat \psi_n^S(\tau,\Pi_{k(n)}\param_0))}{d(\theta,\omega,\mu,\sigma)}} \pi(\tau)d\tau )$ is strictly positive and $\delta_n\underline{\lambda}_n^{-1/2}=o(1)$ then: \[\|\hat \param_n-\Pi_{k(n)}\param_0\|_{m} =O_p(\delta_n \underline{\lambda}_n^{-1/2}).\]
\end{lemma}
The following stochastic equicontinuity result, together with a longer version presented in Lemma \ref{lem:stoch_eq_mixture}, is needed to prove asymption normality (Theorem \ref{th:asymnormal_mixture}).
\begin{lemma}[Stochastic Equicontinuity] \label{lem:stoch_eq_mixture_short}
  Let $\delta_{mn} = \delta_n \underline{\lambda}_n^{-1/2}$, $M_n = \log\log(n)$. If the assumptions in Lemma \ref{lem:cv_rate_mixture_norm} hold and $(M_n\delta_{mn})^{\frac{\gamma^2}{2}}\max(\log[k(n)]^2,|\log[M_n\delta_{mn}]|^2)k(n)^2=o(1)$, then:
  \begin{align*}
    &\sup_{\|\param-\Pi_{k(n)}\param_0\|_m \leq M_n \delta_{mn}} \int   \Big| [\hat \psi_n^S(\tau,\param)-\hat \psi_n^S(\tau,\Pi_{k(n)}\param_0)]-\mathbb{E}[\hat \psi_n^S(\tau,\param)-\hat \psi_n^S(\tau,\Pi_{k(n)}\param_0)]\Big|^2  \pi(\tau)d\tau  \\ &= o_p(1/n).
  \end{align*}
  Also, suppose that $\param \rightarrow \int \mathbb{E} \Big| \hat \psi_t^s(\tau,\param_0)-\hat \psi_t^s(\tau,\param) \Big|^2\pi(\tau)d\tau$ is continuous with respect to $\|\cdot\|_\paramspace$ at $\param=\param_0$, uniformly in $t \geq 1$, then a second stochastic equicontinuity result holds:
  \begin{align*}
    \sup_{\|\param-\Pi_{k(n)}\param_0\|_m \leq M_n \delta_{mn}} \int   \Big| [\hat \psi_n^S(\tau,\param)-\hat \psi_n^S(\tau,\param_0)]-\mathbb{E}[\hat \psi_n^S(\tau,\param)-\hat \psi_n^S(\tau,\param_0)]\Big|^2  \pi(\tau)d\tau  = o_p(1/n).
  \end{align*}
\end{lemma}

\newpage
\section{Proofs for the Main Results} \label{apx:asymptotic_main}

The proofs for the main results allow for a bounded linear operator $B$, as in \citet{Carrasco2000}, to weight the moments. The operator is assumed to be fixed:
\[ \hat Q_n^S(\param) = \int \Big| B\hat \psi_n(\tau)-B\hat \psi_n^S(\tau,\param) \Big|^2 \pi(\tau)d\tau. \]
Since $B$ is bounded linear there exists a $M_B>0$ such that for any two CFs:
\[ \int \Big| B\hat \psi_n(\tau)-B\hat \psi_n^S(\tau,\param) \Big|^2 \pi(\tau)d\tau \leq M_B^2 \int \Big| \hat \psi_n(\tau)-\hat \psi_n^S(\tau,\param) \Big|^2 \pi(\tau)d\tau. \]
As a result, the rate of convergence for the objective function with the weighting $B$ is the same as the rate of convergence without the operator $B$.\footnote{For results on estimating the optimal $B$ see \citet{Carrasco2000,Carrasco2007}. Using their method would lead to $M_{\hat B}\to \infty$ as $n\to \infty$ resulting in a slower rate of convergence for $\hat \param_n$. Having $M_{\hat B}\to \infty$ sufficiently slow would not alter the main results besides having a different, possibly more efficient, asympotic variance.}

\subsection{Consistency} \label{apx:proofsconsistency}

\begin{proof}[Proof of Lemma \ref{lem:mixtures}]:\\
The difference between $e_t^s$ and $\tilde e_t^s$ can be split into two terms:
\begin{align}
& \sum_{j=1}^{k(n)} \left( \mymathcal{1}_{\nu^s_t \in [\sum_{l=0}^{j-1} \omega_l, \sum_{l=0}^{j} \omega_l]} - \mymathcal{1}_{\nu^s_t \in [\sum_{l=0}^{j-1} \tilde \omega_l, \sum_{l=0}^{j} \tilde \omega_l]} \right) \left( \mu_j + \sigma_j Z_{t,j}^s \right) \label{eq:mixt1}\\
& \sum_{j=1}^{k(n)} \mymathcal{1}_{\nu^s_t \in [\sum_{l=0}^{j-1} \tilde \omega_l, \sum_{l=0}^{j} \tilde \omega_l]} \left( \mu_j-\tilde \mu_j + [\sigma_j-\tilde \sigma_j] Z_{t,j}^s \right). \label{eq:mixt2}
\end{align}
To bound the term (\ref{eq:mixt1}) in expectation, combine the fact that $|\mu_j| \leq \bar{\mu}_{k(n)}, |\sigma_j| \leq \bar{\sigma}$ and $\nu_t^s$ and $Z_{t,j}^s$ are independent so that:
\begin{align*}
&\left[\mathbb{E} \left( \sup_{\| (\omega,\mu,\sigma)-(\tilde\omega,\tilde\mu,\tilde\sigma) \|_2 \leq \delta} \Big| \sum_{j=1}^{k(n)} \left( \mymathcal{1}_{\nu^s_t \in [\sum_{l=0}^{j-1} \omega_l, \sum_{l=0}^{j} \omega_l]} - \mymathcal{1}_{\nu^s_t \in [\sum_{l=0}^{j-1} \tilde \omega_l, \sum_{l=0}^{j} \tilde \omega_l]} \right) \left( \mu_j + \sigma_j Z_{t,j}^s \right) \Big|^2 \right) \right]^{1/2}\\
& \leq  \sum_{j=1}^{k(n)} \left[ \mathbb{E} \left( \sup_{\| (\omega,\mu,\sigma)-(\tilde\omega,\tilde\mu,\tilde\sigma) \|_2 \leq \delta} \Big|  \mymathcal{1}_{\nu^s_t \in [\sum_{l=0}^{j-1} \omega_l, \sum_{l=0}^{j} \omega_l]} - \mymathcal{1}_{\nu^s_t \in [\sum_{l=0}^{j-1} \tilde \omega_l, \sum_{l=0}^{j} \tilde \omega_l]} \Big|^2 \right) \right]^{1/2} \\ &\times \left( \bar{\mu}_{k(n)} + \bar\sigma \mathbb{E} \left( |Z_{t,j}^s|^2 \right)^{1/2}  \right).
\end{align*}
The last term is bounded above by $\bar{\mu} + \bar\sigma C_Z.$ Next, note that\\ $ \mymathcal{1}_{\nu^s_t \in [\sum_{l=0}^{j-1} \omega_l, \sum_{l=0}^{j} \omega_l]} - \mymathcal{1}_{\nu^s_t \in [\sum_{l=0}^{j-1} \tilde \omega_l, \sum_{l=0}^{j} \tilde \omega_l]} \in \{0,1\}$ so that:
\begin{align*}
&\mathbb{E} \left( \sup_{\| (\omega,\mu,\sigma)-(\tilde\omega,\tilde\mu,\tilde\sigma) \|_2 \leq \delta} \Big|  \mymathcal{1}_{\nu^s_t \in [\sum_{l=0}^{j-1} \omega_l, \sum_{l=0}^{j} \omega_l]} - \mymathcal{1}_{\nu^s_t \in [\sum_{l=0}^{j-1} \tilde \omega_l, \sum_{l=0}^{j} \tilde \omega_l]} \Big|^2 \right)\\
&= \mathbb{E} \left( \sup_{\| (\omega,\mu,\sigma)-(\tilde\omega,\tilde\mu,\tilde\sigma) \|_2 \leq \delta} \Big|  \mymathcal{1}_{\nu^s_t \in [\sum_{l=0}^{j-1} \omega_l, \sum_{l=0}^{j} \omega_l]} - \mymathcal{1}_{\nu^s_t \in [\sum_{l=0}^{j-1} \tilde \omega_l, \sum_{l=0}^{j} \tilde \omega_l]} \Big| \right).
\end{align*}
Also, for any $j$: $|\sum_{l=0}^j \tilde \omega_l - \sum_{l=0}^j  \omega_l| \leq \sum_{l=0}^j |\tilde \omega_l - \sum_{l=0}^j  \omega_l| \leq \left( \sum_{l=0}^j |\tilde \omega_l - \omega_l|^2 \right)^{1/2} \leq \|\tilde \omega - \omega\|_2 \leq \delta.$ Following a similar approach to \citet{Chen2003}:
\begin{align*}
&\left[ \mathbb{E} \left( \sup_{\| (\omega,\mu,\sigma)-(\tilde\omega,\tilde\mu,\tilde\sigma) \|_2 \leq \delta} \Big|  \mymathcal{1}_{\nu^s_t \in [\sum_{l=0}^{j-1} \omega_l, \sum_{l=0}^{j} \omega_l]} - \mymathcal{1}_{\nu^s_t \in [\sum_{l=0}^{j-1} \tilde \omega_l, \sum_{l=0}^{j} \tilde \omega_l]} \Big| \right) \right]^{1/2} \\
&\leq \left[ \mathbb{E} \left( \sup_{\| (\omega,\mu,\sigma)-(\tilde\omega,\tilde\mu,\tilde\sigma) \|_2 \leq \delta} \Big|  \mymathcal{1}_{\nu^s_t \in [(\sum_{l=0}^{j-1} \tilde \omega_l)-\delta, (\sum_{l=0}^{j} \tilde \omega_l)+\delta]} - \mymathcal{1}_{\nu^s_t \in [\sum_{l=0}^{j-1} \tilde \omega_l, \sum_{l=0}^{j} \tilde \omega_l]} \Big| \right) \right]^{1/2}\\
&= \left[ \left( [(\sum_{l=0}^{j} \tilde \omega_l)+\delta] - [(\sum_{l=0}^{j-1} \tilde \omega_l)-\delta] - [(\sum_{l=0}^{j} \tilde \omega_l) - (\sum_{l=0}^{j-1} \tilde \omega_l)] \right) \right]^{1/2} =\sqrt{2\delta}.
\end{align*}
Overall the term (\ref{eq:mixt1}) is bounded above by $\sqrt{2}(1+C_Z)\left(\bar \mu_{k(n)} + \bar \sigma + k(n) \right) \sqrt{\delta}$.
The term (\ref{eq:mixt2}) can be bounded above by using $0 \leq \mymathcal{1}_{\nu^s_t \in [\sum_{l=0}^{j-1} \tilde \omega_l, \sum_{l=0}^{j} \tilde \omega_l]} \leq 1$ and:
\begin{align*}
&\left[ \mathbb{E} \left( \sup_{\| (\omega,\mu,\sigma)-(\tilde\omega,\tilde\mu,\tilde\sigma) \|_2 \leq \delta} \Big| \sum_{j=1}^{k(n)} \mymathcal{1}_{\nu^s_t \in [\sum_{l=0}^{j-1} \tilde \omega_l, \sum_{l=0}^{j} \tilde \omega_l]} \left( \mu_j-\tilde \mu_j + [\sigma_j-\tilde \sigma_j] Z_{t,j}^s \right) \Big|^2 \right) \right]^{1/2}\\
&\leq \sum_{j=1}^{k(n)} \left[ \mathbb{E} \left( \sup_{\| (\omega,\mu,\sigma)-(\tilde\omega,\tilde\mu,\tilde\sigma) \|_2 \leq \delta} \Big|  (\mu_j-\tilde \mu_j) + [\sigma_j-\tilde \sigma_j] Z_{t,j}^s\Big|^2 \right)\right] ^{1/2} \\
&\leq \sum_{j=1}^{k(n)} \sup_{\| (\omega,\mu,\sigma)-(\tilde\omega,\tilde\mu,\tilde\sigma) \|_2 \leq \delta} \left( |\mu_j-\tilde \mu_j| + |\sigma_j-\tilde \sigma_j|C_Z\right)\\
&\leq (1+C_Z) \sup_{\| (\omega,\mu,\sigma)-(\tilde\omega,\tilde\mu,\tilde\sigma) \|_2 \leq \delta}    \left(\sum_{j=1}^{k(n)} |\mu_j-\tilde \mu_j|^2 + |\sigma_j-\tilde \sigma_j|^2\right)^{1/2} \leq (1+C_Z)\delta.
\end{align*}
Without loss of generality assume that $\delta \leq 1$ so that:
\begin{align*}
  \left[ \mathbb{E} \left( \sup_{\| (\omega,\mu,\sigma)-(\tilde\omega,\tilde\mu,\tilde\sigma) \|_2 \leq \delta} \Big| e_t^s -\tilde e_t^s \Big|^2 \right) \right]^{1/2} \leq 2\sqrt{2}(1+C_Z) \left( 1 +\bar{\mu}_{k(n)}+\bar{\sigma}+k(n) \right)\delta^{1/2}.
\end{align*}
which concludes the proof.
\end{proof}

\begin{proof}[Proof of Lemma \ref{lem:L2smooth}:]
     First note that the cosine and sine functions are uniformly Lispchitz on the real line with Lipschitz coefficient $1$. This implies for any two $(\mathbf{y}_1,\mathbf{y}_2,\mathbf{x})$ and any $\tau \in \mathbb{R}^{d_\tau}$:
  \begin{align*}
    & |\cos(\tau^\prime(\mathbf{y}_1,\mathbf{x}) )-\cos(\tau^\prime(\mathbf{y}_2,\mathbf{x}) )| \leq |\tau^\prime(\mathbf{y}_1-\mathbf{y}_2,0)| \leq \|\tau\|_\infty \|\mathbf{y}_1-\mathbf{y}_2\|,\\
    & |\sin(\tau^\prime(\mathbf{y}_1,\mathbf{x}) )-\sin(\tau^\prime(\mathbf{y}_2,\mathbf{x}) )| \leq |\tau^\prime(\mathbf{y}_1-\mathbf{y}_2,0)| \leq \|\tau\|_\infty \|\mathbf{y}_1-\mathbf{y}_2\|.
  \end{align*}
  As a result, the moment function is also Lipschitz in $\mathbf{y},\mathbf{x}$:
  \begin{align*}
    & |e^{i\tau^\prime(\mathbf{y}_1,\mathbf{x})}-e^{i\tau^\prime(\mathbf{y}_2,\mathbf{x})}|\pi(\tau)^{\frac{1}{4}} \leq 2 \|\tau\|_\infty \pi(\tau)^{\frac{1}{4}}\|\mathbf{y}_1-\mathbf{y}_2\|.
  \end{align*}
  Since $\pi$ is chosen to be the Gaussian density, it satisfies $\sup_\tau \|\tau\|_\infty \pi(\tau)^{\frac{1}{4}} \leq C_{\pi} < \infty$ and $\pi(\tau)^{\frac{1}{2}} \propto \pi(\tau/\sqrt{2})$ which has finite integral.
  The Lispschitz properties of the moments combined with the conditions properties of $\pi$ imply that the $L^2$-smoothness of the moments is implied by the $L^2$-smoothness of the simulated data itself. As a result, the remainder of the proof focuses on the $L^2$-smoothness of $\mathbf{y}_t^s$.
  First note that since $\mathbf{y}_t=(y_t,\dots,y_{t-L})$:
  \begin{align*}
    &\|\mathbf{y}_t(\beta_1)-\mathbf{y}_t(\beta_2)\| \leq \sum_{j=1}^L \|y_{t-j}(\beta_1)-y_{t-j}(\beta_2)\|.
  \end{align*}
  To bound the term in $\mathbf{y}$ above, it suffices to bound the expression for each term $y_t$ with arbitrary $t \geq 1$.
  Assumptions \ref{ass:DGPMixt}, \ref{ass:DGPmixtbis} imply that, for some $\gamma \in (0,1]$:
  \begin{align*}
    &\left[ \mathbb{E} \left( \sup_{\|\beta_1-\beta_2\|_m}\|y_t(\beta_1)-y_t(\beta_2)\|^2 \right)\right]^{1/2} \leq \overline{C}_1 \left[\mathbb{E} \left( \sup_{\|\beta_1-\beta_2\|_m}\|y_{t-1}(\beta_1)-y_{t-1}(\beta_2)\|^2  \right) \right]^{1/2} \\&+ \overline{C}_2 \frac{\delta^\gamma}{\underline{\sigma}_{k(n)}^{2\gamma}} + \overline{C}_3 \left[ \mathbb{E} \left( \sup_{\|\beta_1-\beta_2\|_m}\|u_{t}(\beta_1)-u_{t}(\beta_2)\|^2  \right) \right]^{\gamma/2}.
  \end{align*}
  The term $ \frac{\delta^\gamma}{\underline{\sigma}_{k(n)}^{2\gamma}} $ comes from the fact that $\|\beta_1-\beta_2\|_\infty \leq \frac{\|\beta_1-\beta_2\|_m}{\underline{\sigma}^2_{k(n)}}$ and $\|\beta_1-\beta_2\|_{TV} \leq \frac{\|\beta_1-\beta_2\|_m}{\underline{\sigma}_{k(n)}}$ on $\paramspace_{k(n)}$. Without loss of generality, suppose that $\underline{\sigma}_{k(n)} \leq 1$.\footnote{Recall that by assumption  $\underline{\sigma}_{k(n)} = O(\frac{log[k(n)]^{2/b}}{k(n)})$ goes to zero.}
  Applying this inequality recursively, and using the fact that $y_0^s,u_0^s$ are the same regardless of $\beta$, yields:
  \begin{align*}
    &\left[ \mathbb{E} \left( \sup_{\|\beta_1-\beta_2\|_m}\|y_t(\beta_1)-y_t(\beta_2)\|^2 \right) \right]^{1/2} \\&\leq  \frac{\overline{C}_2}{1-\overline{C}_1} \frac{\delta^\gamma}{\underline{\sigma}_{k(n)}^{2\gamma}} + \overline{C}_3 \sum_{l=0}^{t-1} \overline{C}_1^l \left[ \mathbb{E} \left( \sup_{\|\beta_1-\beta_2\|_m}\|u_{t-l}(\beta_1)-u_{t-l}(\beta_2)\|^2  \right) \right]^{\gamma/2}.
  \end{align*}
  Using Lemmas \ref{lem:mixtures} and \ref{lem:tailsdist} and the same approach as above:
  \begin{align*}
    &\left[ \mathbb{E} \left( \sup_{\|\beta_1-\beta_2\|_m}\|u_t(\beta_1)-u_t(\beta_2)\|^2 \right) \right]^{1/2} \leq  \overline{C}_4 \left[ \mathbb{E} \left( \sup_{\|\beta_1-\beta_2\|_m}\|u_{t-1}(\beta_1)-u_{t-1}(\beta_2)\|^2 \right) \right]^{1/2} \\&+ \overline{C}_5 \frac{\delta^\gamma}{\underline{\sigma}_{k(n)}^{2\gamma}} + \overline{C}_6 C \left(k(n)+ \bar \mu_{k(n)} + \bar \sigma \right) \delta^{\gamma/2}.
  \end{align*}
  Again, applying this inequality recursively yields:
  \begin{align*}
    &\left[ \mathbb{E} \left( \sup_{\|\beta_1-\beta_2\|_m}\|u_t(\beta_1)-u_t(\beta_2)\|^2 \right) \right]^{1/2} \leq  \frac{\overline{C}_5}{1-\overline{C}_4} \frac{\delta^\gamma}{\underline{\sigma}_{k(n)}^{2\gamma}} + \frac{\overline{C}_6}{1-\overline{C}_4} C \left(k(n)+ \bar\mu_{k(n)} + \bar\sigma \right) \delta^{\gamma/2}.
  \end{align*}
  Putting everything together:
  \begin{align*}
    &\left[ \mathbb{E} \left( \sup_{\|\beta_1-\beta_2\|_m}\|y_t(\beta_1)-y_t(\beta_2)\|^2 \right) \right]^{1/2} \\&\leq  \frac{\overline{C}_2}{1-\overline{C}_1} \frac{\delta^\gamma}{\underline{\sigma}_{k(n)}^{2\gamma}} + \frac{\overline{C}_3}{1-\overline{C}_1} \left(  \frac{\overline{C}_5}{1-\overline{C}_4} \frac{\delta^\gamma}{\underline{\sigma}_{k(n)}^{2\gamma}} + \frac{\overline{C}_6}{1-\overline{C}_4} C \left(k(n)+ \bar\mu_{k(n)} + \bar\sigma \right) \delta^{\gamma/2} \right)^\gamma.
  \end{align*}
  Without loss of generality, suppose that $\delta \leq 1$. Then, for some positive constant $\overline{C}$:
  \begin{align*}
    &\left[ \mathbb{E} \left( \sup_{\|\beta_1-\beta_2\|_m}\|y_t(\beta_1)-y_t(\beta_2)\|^2 \right) \right]^{1/2} \leq  \overline{C} \max \left( \frac{\delta^{\gamma^2}}{ \underline{\sigma}_{k(n)}^{2\gamma^2} }, [k(n)+\overline{\mu}_{k(n)}+\overline{\sigma}]^\gamma \delta^{\gamma^2/2} \right).
  \end{align*}
\end{proof}

\begin{proof}[Proof of Theorem \ref{th:consistencyMixtures}:]

The main idea is to show that the Assumptions for Lemma \ref{lem:consistency} hold.
The proof proceeds in in four steps:
\begin{enumerate}
\item First, geometric ergodicity and uniform boundedness of $\hat \psi_n$ implies:
\begin{align*}
  \int |\hat \psi_n(\tau)-\mathbb{E}(\hat \psi_n(\tau))|^2 \pi(\tau)d\tau = O_p(1/n)
\end{align*}
\item Then Lemma \ref{lem:L2smooth} combined with Lemmas \ref{lem:CoveringNumbers}, \ref{lemma:maxineq_dep} imply that uniformly over $\param \in \paramspace_{k(n)}$:
\begin{align*}
  \int |\hat \psi^S_n(\tau,\param)-\mathbb{E}(\hat \psi^S_n(\tau,\param))|^2 \pi(\tau)d\tau = O_p(C_n/n),
\end{align*}
where $C_n$ is given below.
\item The triangle inequality and the previous steps imply that, uniformly over $\param \in \paramspace_{k(n)}$:
\begin{align*}
  \int \Big|[\hat \psi^S_n(\tau,\param)-\hat \psi_n(\tau)]-\mathbb{E}[\hat \psi^S_n(\tau,\param)-\hat \psi_n(\tau)]\Big|^2\pi(\tau)d\tau = O_p(\max(1,C_n)/n).
\end{align*}
And, because $B$ is a bounded linear operator:
\begin{align*}
  &\int \Big|[B\hat \psi^S_n(\tau,\param)-B\hat \psi_n(\tau)]-\mathbb{E}[B\hat \psi^S_n(\tau,\param)-B\hat \psi_n(\tau)]\Big|^2\pi(\tau)d\tau \\ &\leq
  M_B^2 \int \Big|[\hat \psi^S_n(\tau,\param)-\hat \psi_n(\tau)]-\mathbb{E}[\hat \psi^S_n(\tau,\param)-\hat \psi_n(\tau)]\Big|^2\pi(\tau)d\tau = O_p(\max(1,C_n)/n).
\end{align*}
\item By the inequality $|a-b|^2 \geq 1/2|a|^2+|b|^2$ and the previous step, uniformly over $\param\in\paramspace_{k(n)}$:
\begin{align*}
  &1/2\int |B\hat \psi^S_n(\tau,\param)-B\hat \psi_n(\tau)|^2 \pi(\tau)d\tau \\ &\leq \int |\mathbb{E}(B\hat \psi^S_n(\tau,\param)-B\hat \psi_n(\tau))|^2 \pi(\tau)d\tau + O_p(\max(1,C_n)/n)
\end{align*}
and $1/2\int |\mathbb{E}(B\hat \psi^S_n(\tau,\param)-B\hat \psi_n(\tau))|^2 \pi(\tau)d\tau \leq \int |B\hat \psi^S_n(\tau,\param)-B\hat \psi_n(\tau)|^2 \pi(\tau)d\tau + O_p(\max(1,C_n)/n).$
  \end{enumerate}
  This will help show that condition d) in Lemma \ref{lem:consistency} holds.

  First, consider steps \textit{1.} and \textit{2}:\\
\textit{Step 1.:}
For $M>0$, a convergence rate $r_n$ and Markov's inequality:
\begin{align*} \mathbb{P} \left(   \int |\hat \psi_n(\tau)-\mathbb{E}(\hat \psi_n(\tau))|^2 \pi(\tau)d\tau \geq M r_n \right) &\leq \frac{1}{Mr_n}  \mathbb{E} \left(   \int |\hat \psi_n(\tau)-\mathbb{E}(\hat \psi_n(\tau))|^2 \pi(\tau)d\tau  \right)
\\ &= \frac{1}{Mr_n}\int \mathbb{E} \left( |\hat \psi_n(\tau)-\mathbb{E}(\hat \psi_n(\tau))|^2 \right)\pi(\tau)d\tau
\\ & \leq \frac{2}{Mr_n}\frac{1+24 \sum_{m \geq 0} \alpha(m)^{1/p}}{n} \int \pi(\tau)d\tau
\\ &\leq \frac{C_{\alpha,p}}{Mr_nn}.
\end{align*}
The last two inequalities come from Lemma \ref{lemma:covariance_ineq}. If the data is iid then the mixing coefficients $\alpha(m)=0$ for all $m \geq 1$. $C_{\alpha,p}$ is a constant that only depends on the mixing rate $\alpha$, $p$ and the bound on $|\hat \psi_t(\tau)-\mathbb{E}(\hat \psi_t(\tau))| \leq 2$. For $r_n=1/n$ and $M \to \infty$ the probability goes to zero. As a result: $\int |\hat \psi_n(\tau)-\mathbb{E}(\hat \psi_n(\tau))|^2 \pi(\tau)d\tau=O_p(1/n)$.\\

\textit{Step 2.:} The proof is similar to the proof of Lemma C.1 in \citet{Chen2012}. It also begins similarly to \textit{Step 1}, for $M>0$, a convergence rate $r_n$; using Markov's inequality:
\begin{align*} &\mathbb{P} \left(  \sup_{\param \in \paramspace_{k(n)}} \int |\hat \psi^S_n(\tau,\param)-\mathbb{E}(\hat \psi^S_n(\tau,\param))|^2 \pi(\tau)d\tau \geq M r_n \right) \\ &\leq \frac{1}{Mr_n}  \mathbb{E} \left(   \sup_{\param \in \paramspace_{k(n)}} \int |\hat \psi^S_n(\tau,\param)-\mathbb{E}(\hat \psi^S_n(\tau,\param))|^2 \pi(\tau)d\tau  \right)
\\ &\leq \frac{1}{Mr_n}\int \mathbb{E} \left(  \sup_{\param \in \paramspace_{k(n)}} |\hat \psi^S_n(\tau,\param)-\mathbb{E}(\hat \psi^S_n(\tau))|^2 \right)\pi(\tau)d\tau
\\ &\leq \frac{1}{Mr_n}\int \mathbb{E} \left(  \sup_{\param \in \paramspace_{k(n)}} |\hat \psi^s_n(\tau,\param)-\mathbb{E}(\hat \psi^s_n(\tau))|^2 \right)\pi(\tau)d\tau
\end{align*}

Suppose that there is an upper bound $C_n$ such that for all $\tau$:
\[ \mathbb{E} \left(  \sup_{\param \in \paramspace_{k(n)}} |[\hat \psi^s_n(\tau,\param)-\mathbb{E}(\hat \psi^s_n(\tau,\param))]\pi(\tau)^{1/(2+\eta)}|^2 \right) \leq C_n/n\]
If the following also holds $\int \pi(\tau)^{1-2/(2+\eta)}d\tau = C_\eta < \infty$ then:
\[ \frac{1}{Mr_n}\int \mathbb{E} \left(  \sup_{h \in \paramspace_{k(n)}} |\hat \psi^s_n(\tau,\param)-\mathbb{E}(\hat \psi^s_n(\tau,\param))|^2 \right)\pi(\tau)d\tau \leq \frac{C_\eta C_n}{Mr_n n}. \]
Take $r_n = C_n/n = o(1)$, then for $M \to \infty$ the probability goes to zero. As a result:
\[\sup_{\param \in \paramspace_{k(n)}} \int |\hat \psi^S_n(\tau,\param)-\mathbb{E}(\hat \psi^S_n(\tau,\param))|^2 \pi(\tau)d\tau = O_p(C_n/n).\]

The bounds $C_n$ are now computed, first in the iid case. By theorem 2.14.5 of \citet{VanderVaart1996}:
\begin{align*}
&\mathbb{E} \left(  \sup_{\param \in \paramspace_{k(n)}} \Big|\sqrt{n}[\hat \psi^s_n(\tau,\param)-\mathbb{E}(\hat \psi^s_n(\tau,\param))]\pi(\tau)^{1/(2+\eta)}\Big|^2 \right) \\& \leq \left( 1+ \mathbb{E} \left(  \sup_{\param \in \paramspace_{k(n)}} \Big|\sqrt{n}[\hat \psi^s_n(\tau,\param)-\mathbb{E}(\hat \psi^s_n(\tau,\param))]\pi(\tau)^{1/(2+\eta)}\Big| \right)\right)^2.
\end{align*}
Also, by theorem 2.14.2 of \citet{VanderVaart1996} there exists a universal constant $K>0$ such that for each $\tau \in \mathbb{R}^{d_\tau}$:
\[ \mathbb{E} \left(  \sup_{\param \in \paramspace_{k(n)}} \Big|\sqrt{n}[\hat \psi^s_n(\tau,\param)-\mathbb{E}(\hat \psi^s_n(\tau,\param))]\pi(\tau)^{1/(2+\eta)}\Big| \right) \leq K \int_0^1 \sqrt{1+\log N_{[\,]}(x,\Psi_{k(n)},\|\cdot\|)}dx \]
with $\Psi_{k(n)} = \big\{\psi: \paramspace_{k(n)}\to \mathbb{C}, \param \to  \psi_t^S(\tau,\param)\pi(\tau)^{1/(2+\eta)}\big\}$, $N_{[\,]}$ is the covering number with bracketing. Because of the $L^p$-smoothness, it is bounded above by:
\[ N_{[\,]}(x,\Psi_{k(n)},\|\cdot\|) \leq N_{[\,]}(\frac{x^{1/\gamma}}{C^{1/\gamma}},\paramspace_{k(n)},\|\cdot\|) \leq C^\prime N_{[\,]}(x^{1/\gamma},\paramspace_{k(n)},\|\cdot\|). \]
Let $\sqrt{C_n}=\sqrt{1+\log N_{[\,]}(x^{1/\gamma},\paramspace_{k(n)},\|\cdot\|)}dx$, together with the previous inequality, it implies:
\[ \mathbb{E} \left(  \sup_{\param \in \paramspace_{k(n)}} \Big|\sqrt{n}[\hat \psi^s_n(\tau,\param)-\mathbb{E}(\hat \psi^s_n(\tau,\param))]\pi(\tau)^{1/(2+\eta)}\Big|^2 \right) \leq \left(1 +K\sqrt{C_n} \right)^2 \leq 4(1+K^2)C_n.\]
To conclude, divide by $n$ on both sides to get the bound:
\[ \mathbb{E} \left(  \sup_{\param \in \paramspace_{k(n)}} \Big|[\hat \psi^s_n(\tau,\param)-\mathbb{E}(\hat \psi^s_n(\tau,\param))]\pi(\tau)^{1/(2+\eta)}\Big|^2 \right)  \leq 4(1+K^2)C_n/n.\]

For the dependent case, Lemma \ref{lemma:maxineq_dep} implies that if $\hat \psi^s_t(\tau,\param)$ is $\alpha$-mixing at an exponential rate, the moments are bounded and the sieve spaces are compact:
\[ \mathbb{E} \left(  \sup_{\param \in \paramspace_{k(n)}} \Big|\sqrt{n}[\hat \psi^s_n(\tau,\param)-\mathbb{E}(\hat \psi^s_n(\tau,\param))]\pi(\tau)^{1/(2+\eta)}\Big|^2 \right) \leq \left(1 +K\sqrt{C_n} \right)^2 \leq KC_n\]
with, for any $\vartheta \in (0,1)$ such that the integral exists:
\[C_n = \int_0^1 \left( x^{\vartheta/2-1}\sqrt{\log N_{[\,]}(x^{1/\gamma},\paramspace_{k(n)},\|\cdot\|_\paramspace)} + \log^2 N_{[\,]}(x^{1/\gamma},\paramspace_{k(n)},\|\cdot\|_\paramspace)\right) dx.\]
Lemma \ref{lem:CoveringNumbers} then derives bounds for $C_n$ in terms of $k(n)$.

\textit{Step 3.:} follows from the triangle inequality and the assumption that $B$ is a bounded linear operator.\\

\textit{Step 4.:} The following two inequalities can be derived from the inequality $|a-b|^2 \geq 1/2|a|^2 + |b|^2$, which is symmetric in $a$ and $b$:
\begin{align*}
&\Big|[B\hat \psi^S_n(\tau,\param)-B\hat \psi_n(\tau)]-\mathbb{E}[B\hat \psi^S_n(\tau,\param)-B\hat \psi_n(\tau)]\Big|^2 \\ &\geq 1/2 \Big|B\hat \psi^S_n(\tau,\param)-B\hat \psi_n(\tau)\Big|^2 + \Big|\mathbb{E}[B\hat \psi^S_n(\tau,\param)-B\hat \psi_n(\tau)]\Big|^2
\end{align*}
and
\begin{align*}
&\Big|[B\hat \psi^S_n(\tau,\param)-B\hat \psi_n(\tau)]-\mathbb{E}[B\hat \psi^S_n(\tau,\param)-B\hat \psi_n(\tau)]\Big|^2 \\ &\geq \Big|B\hat \psi^S_n(\tau,\param)-B\hat \psi_n(\tau)\Big|^2 + 1/2\Big|\mathbb{E}[B\hat \psi^S_n(\tau,\param)-B\hat \psi_n(\tau)]\Big|^2.
\end{align*}
Taking integrals on both sides and given that \[\int \Big|[B\hat \psi^S_n(\tau,\param)-B\hat \psi_n(\tau)]-\mathbb{E}[B\hat \psi^S_n(\tau,\param)-B\hat \psi_n(\tau)]\Big|^2 \pi(\tau)d\tau = O_p(C_n/n)\] uniformly in $h\in \paramspace_{k(n)}$, the desired result follows:
$1/2\hat Q_n^S(\param) \leq Q_n(\param)+O_p(C_n/n)$ and $1/2 Q_n(\param) \leq \hat Q_n^S(\param)+O_p(C_n/n).$

Lemma \ref{lem:CoveringNumbers} implies that $C_n = O(k(n)^4\log[k(n)]^4)$ in the dependent case, and $C_n = O(k(n)\log[k(n)])$ in the iid case. Combining this, condition (\ref{eq:separation}) in the Theorem, the rate for $Q_n(\Pi_{k(n)}\param_0)$ which is derived in Lemma \ref{lem:ObjApproxRate} together implies the conditions for Lemma \ref{lem:consistency} hold so that the estimator is consistent.
\end{proof}

\subsection{Rate of Convergence}

\begin{proof}[Proof of Theorem \ref{th:convrateMixture}:]

Let $C_n$ be as in the proof of Theorem \ref{th:consistencyMixtures},
let $\varepsilon>0$ and \[r_n = \max \left(\sqrt{\frac{C_n}{n}},\sqrt{\eta_n},\frac{\log[k(n)]^{2r/(b+2)}}{k(n)^{\gamma^2r}},\frac{1}{\sqrt{n}} \right).\] Proving the result amounts to showing that there exists $M>0$ and $N >0$ such that $ \forall n \geq N$:
  \begin{align}
    \mathbb{P} \left( \|\hat \param_n - \param_0 \|_{weak} \geq M r_n \right) < \varepsilon. \label{eq:cv_rate_ineq}
  \end{align}
  First, under the stated assumptions, the following inequalities hold:
  \begin{enumerate}
    \item $\hat Q_n^S(\param) \leq 2Q_n(\param) + O_p(C_n/n)$,
    \item $Q_n(\Pi_{k(n)}\param_0) \leq O(\max(\frac{\log[k(n)]^{4r/(b+2)}}{k(n)^{2\gamma^2r}},1/n^2)$,
    \item $\|\param-\param_0\|^2_{weak} \leq \underline{C}_w^{-1}[Q_n(\param)+O(1/n^2)]$.
  \end{enumerate}
  The first was derived in the proof of Theorem \ref{th:consistencyMixtures}, the second is due to Lemma \ref{lem:ObjApproxRate} and the third comes from Assumption \ref{ass:weaknorm} with Lemma \ref{lem:geom_ergo}.
  Applying them in order to (\ref{eq:cv_rate_ineq}):
  \begin{align*}
    &\mathbb{P} \left( \|\hat \param_n - \param_0 \|_{weak} \geq M r_n \right) \\& \leq \mathbb{P} \left( \inf_{\param \in \paramspace_{osn}, \, \|\param - \param_0 \|_{weak} \geq M r_n} \hat Q^S_n(\param) \leq \inf_{\param \in \paramspace_{osn}} \hat Q^S_n(\param) + O_p\left(\eta_n \right)  \right) \\
    &\leq   \mathbb{P} \left( \inf_{\param \in \paramspace_{osn}, \, \|\param - \param_0 \|_{weak} \geq M r_n} Q_n(\param) \leq \inf_{\param \in \paramspace_{osn}} Q_n(\param) + O_p\left(\max(\frac{C_n}{n},\eta_n ) \right)  \right) \\
    &\leq  \mathbb{P} \left( \inf_{\param \in \paramspace_{osn}, \, \|\param - \param_0 \|_{weak} \geq M r_n} Q_n(\param) \leq Q_n(\Pi_{k(n)}\param_0) + O_p\left(\max(\frac{C_n}{n},\eta_n ) \right)  \right)\\
    &\leq  \mathbb{P} \left(  \inf_{\param \in \paramspace_{osn}, \, \|\param - \param_0 \|_{weak} \geq M r_n} Q_n(\param) \leq O_p \left(\max(\frac{\log[k(n)]^{4r/(b+2)}}{k(n)^{2\gamma^2r}},\frac{1}{n},\frac{C_n}{n},\eta_n ) \right)  \right)\\
    &\leq  \mathbb{P} \left(  M^2 r_n^2 \leq O_p \left(\max(\frac{\log[k(n)]^{4r/(b+2)}}{k(n)^{2\gamma^2r}},\frac{1}{n},\frac{C_n}{n},\eta_n ) \right) \right)
  \end{align*}
  For $r_n$ defined above, this probability becomes: $\mathbb{P} \left(  M^2 \leq O_p(1) \right) \to 0 \text{ as } M\to \infty.$
  This concludes the first part of the proof. By definition of the local measure of ill-posedness:
  \begin{align*} \hspace*{-1.2cm}
  &\|\hat \param_n - \param_0\|_\paramspace \\ &\leq \|\Pi_{k(n)} \param_0 - \param_0\|_\paramspace + \|\hat \param_n - \Pi_{k(n)} \param_0\|_\paramspace \frac{\|\hat \param_n - \Pi_{k(n)} \param_0\|_{weak}}{\|\hat \param_n - \Pi_{k(n)} \param_0\|_{weak}}\\
  &\leq \|\Pi_{k(n)} \param_0 - \param_0\|_\paramspace + \tau_n\|\hat \param_n - \Pi_{k(n)} \param_0\|_{weak} \\
  &\leq \|\Pi_{k(n)} \param_0 - \param_0\|_\paramspace + \tau_n \left(\|\hat \param_n - \param_0\|_{weak} + \|\param_0 - \Pi_{k(n)} \param_0\|_{weak} \right) \\
  &\leq \|\Pi_{k(n)} \param_0 - \param_0\|_\paramspace + \tau_n \left(\|\hat \param_n - \param_0\|_{weak} + \underline{C}_w^{-1}Q_n(\Pi_{k(n)} \param_0) + O(1/n^2) \right).
  \end{align*}
  Applying Lemma \ref{lem:ObjApproxRate} again to $Q_n(\Pi_{k(n)} \param_0)$ concludes the proof.
\end{proof}
\begin{proof}[Proof of Corollary \ref{rmk:full_simu}:]
  The proof is immediate by taking the size of the simulated sample to be $nS$ instead of $n$, which implies $\sqrt{C_n/n} = k(n)^2\log[n]^2/\sqrt{n S}$ in the proof of Theorem \ref{th:convrateMixture}, and noting that $\hat\psi_n$ converges at a $\sqrt{n}$-rate so that convergence is no faster than $\min\left(\frac{k(n)^2\log[n]^2}{\sqrt{n\times S}},\frac{1}{\sqrt{n}}\right)$.
\end{proof}
\subsection{Asymptotic Normality}
\begin{proof}[Proof of Theorem \ref{th:asymnormal_mixture}:]

    Assumption \ref{ass:sufficient_cv_rate_mixture} ii-iii. allows the following linearization:
  \begin{align*}
    &\frac{\sqrt{n}}{\sigma^*_n} \left( \phi(\hat \param_n) - \phi(\param_0) \right)\\ &= \frac{\sqrt{n}}{\sigma^*_n}\frac{d\phi(\param_0)}{d\param}[\hat \param_n-\param_0]+o_p(1)\\
    &=\frac{\sqrt{n}}{\sigma^*_n}\frac{d\phi(\param_0)}{d\param}[\hat \param_n-\param_{0,n}]+o_p(1)\\
    &=\sqrt{n}\langle u_n^*,\hat \param_n-\param_{0,n}\rangle+o_p(1)\\
    &=\sqrt{n}\langle u_n^*,\hat \param_n-\param_{0}\rangle+o_p(1)\\
    &= \frac{\sqrt{n}}{2} \left( \int \left[ B\psi_\param(\tau,u_n^*) \overline{B\psi_\param(\tau,\hat \param_n-\param_0)} + \overline{B\psi_\param(\tau,u_n^*)} B\psi_\param(\tau,\hat \param_n-\param_0) \right]\right)\pi(\tau)d\tau +o_p(1).
  \end{align*}
  Using Lemma \ref{lem:equivs} \textit{a)} and \textit{b)}, replace the term $B\psi_\param(\tau,\hat \param_n-\param_0)$ under the integral with $B\hat \psi_n^S(\tau,\hat \param_n)-B\hat \psi_n^S(\tau,\param_0)$ so that:
  \begin{align*}
    \frac{\sqrt{n}}{\sigma^*_n} \left( \phi(\hat \param_n) - \phi(\param_0) \right) &=
    \frac{1}{2} \Big( \int \Big[ B\psi_\param(\tau,u_n^*) \overline{[B\hat \psi_n^S(\tau,\hat \param_n)-B\hat \psi_n^S(\tau,\param_0)]} \\
    &+ \overline{B\psi_\param(\tau,u_n^*)} [B\hat \psi_n^S(\tau,\hat \param_n)-B\hat \psi_n^S(\tau,\param_0)] \Big]\Big)\pi(\tau)d\tau+o_p(1).
  \end{align*}
  Now Lemma \ref{lem:equivs} \textit{c)} implies that $B\hat \psi_n^S(\tau,\hat \param_n)$ can be replaced with $B\hat \psi_n(\tau)$ up to a $o_p(1/\sqrt{n})$ so that the above becomes:
  \begin{align*}
    \frac{\sqrt{n}}{\sigma^*_n} \left( \phi(\hat \param_n) - \phi(\param_0) \right) &=
    \frac{\sqrt{n}}{2} \Big( \int \Big[ B\psi_\param(\tau,u_n^*) \overline{BZ_n^S(\tau)} + \overline{B\psi_\param(\tau,u_n^*)} BZ_n^S(\tau) \Big]\Big)\pi(\tau)d\tau+o_p(1).
  \end{align*}
  To conclude, apply a Central Limit Theorem to the real-valued random variable variable:
  \[\frac{1}{2}\int [B\psi_\param(\tau,u_n^*) \overline{B Z_t^S(\tau)} + \overline{B\psi_\param(\tau,u_n^*)} BZ_t^S(\tau)]\pi(\tau)d\tau.\]
  Because of $u_n^*$ and the geometric ergodicity of the simulated data, a CLT for non-stationary mixing triangular arrays is required. The results in \citet{Wooldridge1988} can be applied, the following verifies that the sufficient conditions hold. For any $\delta>0$:
  \begin{align*}
  &\mathbb{E} \left( \Big| \int [\psi_\param(\tau,u_n^*) \overline{Z_t^S(\tau)} + \overline{\psi_\param(\tau,u_n^*)} Z_t^S(\tau)]\pi(\tau)d\tau \Big|^{2+\delta} \right) \\
  &\leq 2^{2+\delta}\left[\mathbb{E} \left(  \int \Big|\overline{\psi_\param(\tau,u_n^*)} Z_t^S(\tau)\Big|\pi(\tau)d\tau \right)\right]^{2+\delta}\\
  &\leq 2^{2+\delta} \left( \int \Big|B\psi_\param(\tau,u_n^*)\Big|^2\pi(\tau)d\tau \right)^{\frac{2+\delta}{2}}  \left[\mathbb{E} \left(  \int \Big|BZ_t^S(\tau)\Big|^2\pi(\tau)d\tau \right)\right]^{\frac{2+\delta}{2}}.
  \end{align*}
  By definition of $u_n^*$ and $\|\cdot\|_{weak}$:
  \[ \left( \int \Big|B\psi_\param(\tau,u_n^*)\Big|^2\pi(\tau)d\tau \right)^{1/2} = \|v_n^*\|_{weak}/\sigma^*_n \in [1/\overline{a},1/\underline{a}].\]
  Because $B$ is bounded linear and $|Z_t^S(\tau)|\leq 2$: $\left[\mathbb{E} \left(  \int \Big|BZ_t^S(\tau)\Big|^2\pi(\tau)d\tau \right)\right]^{\frac{2+\delta}{2}} \leq [2M_B]^{2+\delta}.$
  Eventually, it implies:
  \begin{align*}
  &\mathbb{E} \left( \Big| \int [\psi_\param(\tau,u_n^*) \overline{Z_t^S(\tau)} + \overline{\psi_\param(\tau,u_n^*)} Z_t^S(\tau)]\pi(\tau)d\tau \Big|^{2+\delta} \right) \leq \frac{[4M_B]^{2+\delta}}{\underline{a}} < \infty.
  \end{align*}
  Given the mixing condition and the definition of $\sigma^*_n$:
  \[ \frac{\sqrt{n}}{2}\int [B\psi_\param(\tau,u_n^*)[ \overline{B Z_t^S(\tau)-B\mathbb{E}( Z_t^S(\tau))}] + \overline{B\psi_\param(\tau,u_n^*)} [BZ_t^S(\tau)-B\mathbb{E}( Z_t^S(\tau))]]\pi(\tau)d\tau \overset{d}{\to} \mathcal{N}(0,1). \]
  By geometric ergodicity and because the characteristic function is bounded $\sqrt{n}|\mathbb{E}( Z_t^S(\tau))| \leq C_\rho /\sqrt{n} = o(1)$, hence:
  \[ \frac{\sqrt{n}}{2}\int [B\psi_\param(\tau,u_n^*)\overline{B Z_t^S(\tau)} + \overline{B\psi_\param(\tau,u_n^*)} BZ_t^S(\tau)]\pi(\tau)d\tau \overset{d}{\to} \mathcal{N}(0,1). \]
  This concludes the proof.
  \end{proof}


\end{appendices}
\newpage

\begin{titlingpage} 
  \emptythanks
  \title{ {Supplement to\\\lQ A Sieve-SMM Estimator for Dynamic Models"}}
  \author{\Large Jean-Jacques Forneron\thanks{Department of Economics, Boston University, 270 Bay State Road, Boston, MA 02215.\newline Email: \href{mailto:jjmf@bu.edu}{jjmf@bu.edu}.}}
  \setcounter{footnote}{0}

  \clearpage 
  \maketitle 
  \thispagestyle{empty} 
  \begin{center}
  This Supplemental Material consists of Appendices 
  \ref{sec:proof_prelim}, \ref{apx:additional_results}, \ref{apx:proof_additional}, \ref{apx:add_appli} and \ref{apx:add_res} to the main text.
  \end{center}
\end{titlingpage}

\newpage %
\begin{appendices}

    \renewcommand\thetable{\thesection\arabic{table}}
    \renewcommand\thefigure{\thesection\arabic{figure}}
    \renewcommand{\theequation}{\thesection.\arabic{equation}}
    \renewcommand\thelemma{\thesection\arabic{lemma}}
    \renewcommand\thetheorem{\thesection\arabic{theorem}}
    \renewcommand\thedefinition{\thesection\arabic{definition}}
      \renewcommand\theassumption{\thesection\arabic{assumption}}
    \renewcommand\theproposition{\thesection\arabic{proposition}}
      \renewcommand\theremark{\thesection\arabic{remark}}
      \renewcommand\thecorollary{\thesection\arabic{corollary}}

\setcounter{page}{1}
\section{Proofs for the Preliminary Results} \label{sec:proof_prelim}

\begin{proof}[Proof of Lemma \ref{lem:ApproxSimuGAUT}]
  The proof proceeds by recursion. Denote $\Pi_{k(n)} f_j \in \mathcal{F}_{k(n)}$ the mixture approximation of $f_j$ from Lemma \ref{lem:Kruijer}. For $d_e=1$, Lemma \ref{lem:Kruijer} implies $\|f_1 -\Pi_{k(n)}f_1\|_{TV} = O(\frac{\log[k(n)]^{r/b}}{k(n)^r})$ and $\|f_1 -\Pi_{k(n)}f_1\|_{\infty} = O(\frac{\log[k(n)]^{r/b}}{k(n)^r}).$
  Suppose the result holds for $f_1\times \dots \times f_{d_e}$. Let $f=f_1\times \dots \times f_{d_e} \times f_{d_e+1}$; let:
  \begin{align*}
  &d_{t+1}=f_1\times \dots \times f_{d_e} \times f_{d_e+1}-\Pi_{k(n)}f_1\times \dots \times \Pi_{k(n)}f_{d_e} \times \Pi_{k(n)}f_{d_e+1}\\
  &d_{t}= f_1\times \dots \times f_{d_e}-\Pi_{k(n)}f_1\times \dots \times \Pi_{k(n)}f_{d_e}.
  \end{align*}
  The difference can be re-written recursively:
  \begin{align*}
  &d_{t+1}=d_t f_{d_e+1} + \Pi_{k(n)}f_1\times \dots \times \Pi_{k(n)}f_{d_e} \left(f_{d_e+1}-\Pi_{k(n)}f_{d_e+1} \right).
  \end{align*}
  Since $\int f_{d_e+1} = \int \Pi_{k(n)}f_1\times \dots \times \Pi_{k(n)}f_{d_e} =1$, the total variation distance is:
  $\|d_{t+1}\|_{TV}\leq \|d_t\|_{TV} + \|f_{d_e+1}-\Pi_{k(n)}f_{d_e+1} \|_{TV} = O( \frac{\log[k(n)]^{r/b}}{k(n)^r}).$ 
  And the supremum distance is:
  \begin{align*}
  &\|d_{t+1}\|_{\infty}\leq \|d_t\|_{\infty}\|f_{d_e+1} \|_\infty + \|\Pi_{k(n)}f_1\times \dots \times \Pi_{k(n)}f_{d_e}\|_{\infty}\|f_{d_e+1}-\Pi_{k(n)}f_{d_e+1} \|_{\infty}\\
  &\leq \|d_t\|_{\infty} \left(\|f_{d_e+1} \|_\infty + \|f_1\times \dots \times f_{d_e}\|_{\infty} \|f_{d_e+1}-\Pi_{k(n)}f_{d_e+1} \|_{\infty} \right)=O\left( \frac{\log[k(n)]^{r/b}}{k(n)^r}\right).
  \end{align*}
  \end{proof}

\begin{proof}[Proof of Lemma \ref{lem:tailsdist}]:\\
  To reduce notation, the $t$ and $s$ subscripts will be dropped in the following. The proof is similar for both $e_{1}$ and $e_{2}$ so the proof is only given for $e_{1}$.

  First, the densities of $e_{1}$ and $e_{2}$ are derived, the first two results follow. Noting that the draws are defined using quantile functions, inverting the formula yields: $\nu_{1} = \frac{1}{1-e_{1}^{2+\tail_1}}$. This is a proper CDF on $(-\infty,0]$ since $e_1 \rightarrow \frac{1}{1-e_{1}^{2+\tail_1}}$ is increasing and has limits $0$ at $-\infty$ and $1$ at $0$. Its derivative is the density function: $(2+\tail_1) \frac{e_1^{1+\tail_1}}{(1-e_1^{2+\tail_1})^2}$. It is continuous on $(-\infty,0]$ and has an asymptote at $-\infty$: $(2+\tail_1) \frac{e_1^{1+\tail_1}}{(1-e_1^{2+\tail_1})^2} \times e_1^{3+\tail_1} \to (2+\tail_1)$ as $e_1 \to -\infty$. Since $\tail_1 \in [\underline{\tail},\bar{\tail}]$ with $0<\underline{\tail}$ then $\mathbb{E}|e_1|^2 \leq C < \infty$ for some finite $C>0$. Similar results hold for $e_2$ which has density $(2+\tail_2) \frac{e_2^{1+\tail_2}}{(1+e_2^{2+\tail_2})^2}$ on $[0,+\infty)$.\\

  Second,  $\tail_1 \rightarrow e_1(\tail_1)$ is shown to be $L^2$-smooth. Let $|\tail_1-\tilde \tail_1| \leq \delta$, using the mean value theorem, for each $\nu_1$ there exists an intermediate value $\check{\tail_1} \in [\tail_1,\tilde \tail_1]$ such that:
  \[ \left( \frac{1}{\nu_1}-1\right)^{\frac{1}{2+\tail_1}}-\left( \frac{1}{\nu_1}-1\right)^{\frac{1}{2+\tilde \tail_1}} = \frac{1}{2+\check{\tail_1}} \log(\frac{1}{\nu_1}-1) \left( \frac{1}{\nu_1}-1\right)^{\frac{1}{2+\check{\tail_1}}} (\tail_1-\tilde \tail_1).\]
  The first term is bounded by $1/(2+\underline{\tail})$, the second is bounded by $\log(\frac{1}{\nu_1}+1) \left( \frac{1}{\nu_1}+1\right)^{\frac{1}{2+\underline{\tail}}}$, and the last term is bounded above, in absolute value, by $\delta$.

Finally, in order to conclude the proof, the integral $\int_0^{1} \log(\frac{1}{\nu_1}+1)\left( \frac{1}{\nu_1}+1 \right)^{\frac{2}{2+\underline{\tail}}}d\nu_1$ needs to be finite.
  By a change of variables, it can be re-written as: $\int_2^{\infty} \log(\nu) \nu ^{\frac{2}{2+\underline{\tail}}-2}d\nu.$
  Since $\frac{2}{2+\underline{\tail}}-2 < -1$, the integral is always finite and thus:
  \[ \left[ \mathbb{E}\left( \sup_{|\tail_1-\tilde \tail_1| \leq \delta} |e_{t,1}^s(\tail_1)-e_{t,1}^s(\tilde \tail_1)|^2 \right) \right]^{1/2} \leq \frac{\delta}{2+\underline{\tail}} \sqrt{\int_2^{\infty} \log(\nu) \nu ^{\frac{2}{2+\underline{\tail}}-2}d\nu}. \]
\end{proof}

\begin{proof}[Proof of Lemma \ref{lem:CoveringNumbers}:]
  Since $\paramspace_{k(n)}$ is contained in a ball of radius $\max(\overline{\mu}_{k(n)},\overline{\sigma},\|\theta\|_{\infty})$ in $\mathbb{R}^{3[k(n)+2]+d_\theta}$ under $\|\cdot\|_m$, the covering number for $\paramspace_{k(n)}$ can be computed under the $\|\cdot\|_m$ norm using a result from \citet{kolmogorov1959varepsilon}. As a result, the covering number $N(x,\paramspace_{k(n)},\|\cdot\|_m)$ satisfies:
  $N(x,\paramspace_{k(n)},\|\cdot\|_m) \leq 2\left( 3[k(n)+2]+d_\theta \right) \left( \frac{2\max(\bar\mu_{k(n)},\bar{\sigma})}{x}+1 \right)^{3[k(n)+2]+d_\theta}.$
  The rest follows from Lemmas \ref{lem:L2smooth} and \ref{lemma:maxineq_dep}.
\end{proof}

\begin{proof}[Proof of Lemma \ref{lem:ObjApproxRate}:]
  First, using the assumption that $B$ is a bounded linear operator:
  \begin{align*}
    &Q_n(\Pi_{k(n)}\param_0) \\ &\leq M^2_B \int \Big| \mathbb{E} \left( \hat \psi_n(\tau) - \hat \psi_n^S(\tau,\Pi_{k(n)}\param_0) \right) \Big|^2 \pi(\tau)d\tau \\
    &\leq 3M_B^2 \left( \int \Big| \mathbb{E} \left( \hat \psi_n(\tau) - \hat \psi_n^S(\tau,\param_0) \right) \Big|^2 \pi(\tau)d\tau + \int \Big| \mathbb{E} \left( \hat \hat \psi_n^S(\tau,\param_0) - \hat \psi_n^S(\tau,\Pi_{k(n)}\param_0) \right) \Big|^2 \pi(\tau)d\tau  \right)
  \end{align*}
  Each term can be bounded above individually. Re-write the first term in terms of distribution: $\Big| \mathbb{E} \left( \hat \psi_n(\tau) - \hat \psi_n^S(\tau,\param_0) \right) \Big| = \Big|\frac{1}{n} \sum_{t=1}^n\int e^{i\tau^\prime (\mathbf{y}_t,\mathbf{x_t})} [f_t^*(\mathbf{y}_t,\mathbf{x_t})-f_t(\mathbf{y}_t,\mathbf{x_t})]d\mathbf{y}_td\mathbf{x}_t \Big|$, 
  where $f_t$ is the distribution of $(\mathbf{y}_t(\param_0),\mathbf{x}_t)$ and $f_t$ the stationary distribution of $(\mathbf{y}_t(\param_0),\mathbf{x}_t)$.
  Using the geometric ergodicity assumption, for all $\tau$:
  \begin{align*}
    &\Big|\frac{1}{n} \sum_{t=1}^n\int e^{i\tau^\prime (\mathbf{y}_t,\mathbf{x_t})} [f_t^*(\mathbf{y}_t,\mathbf{x_t})-f_t(\mathbf{y}_t,\mathbf{x_t})]d\mathbf{y}_td\mathbf{x}_t \Big| \leq \frac{1}{n} \sum_{t=1}^n \int \Big|f_t^*(\mathbf{y}_t,\mathbf{x_t})-f_t(\mathbf{y}_t,\mathbf{x_t}) \Big|d\mathbf{y}_td\mathbf{x}_t
    \\&= \frac{2}{n} \sum_{t=1}^n \|f_t^*-f_t\|_{TV}
    \leq  \frac{2C_\rho}{n} \sum_{t=1}^n \rho^t \leq \frac{2C_\rho}{(1-\rho)n}
  \end{align*}
  for some $\rho \in (0,1)$ and $C_\rho >0$. This yields a first bound:
  \[ \int \Big| \mathbb{E} \left( \hat \psi_n(\tau) - \hat \psi_n^S(\tau,\param_0) \right) \Big|^2 \pi(\tau)d\tau \leq  \frac{4C^2_\rho}{(1-\rho)^2} \frac{1}{n^2} = O\left(\frac{1}{n^2}\right).\]
  
  The mixture norm $\|\cdot\|_m$ is not needed here to bound the second term since it involves population CFs. Some changes to the proof of Lemma \ref{lem:L2smooth} allows to find bounds in terms of $\|\cdot\|_\paramspace$ and $\|\cdot\|_{TV}$ for which Lemma \ref{lem:ApproxSimuGAUT} gives the approximation rates.
  
  To bound the second term, re-write the simulated data as:
  \begin{align*}
    & y_t^s = g_{obs,t}(\mathbf{x}_{t:1},\param,\mathbf{e}_{t:1}^s), \quad  u_t^s = g_{latent,t}(\param,\mathbf{e}_{t:1}^s)
  \end{align*}
  with $\param=(\theta,f)$, $e_t^s \sim f$ and $\mathbf{x}_{t:1}=(x_t,\dots,x_1),\mathbf{e}_{t:1}^s=(e_t^s,\dots,e_1^s)$. Under Assumption \ref{ass:DGPMixt} or \ref{ass:DGPmixtbis}, using the same sequence of shocks $(e_t^s)$:
  $\mathbb{E} \left( \Big\|g_{obs,t}(\mathbf{x}_{t:1},\param_0,\mathbf{e}_{t:1}^s)-g_{obs,t}(\mathbf{x}_{t:1},\Pi_{k(n)}\param_0,\mathbf{e}_{t:1}^s)\Big\|\right)
    \leq \overline{C} \| \Pi_{k(n)}f_0-f_0 \|^\gamma_\paramspace.$
  This is similar to the proof of Lemma \ref{lem:L2smooth}, first re-write the difference as:
  \begin{align*}
    &\mathbb{E} \Big( \Big\|g_{obs}(g_{obs,t-1}(\mathbf{x}_{t-1:1},\param_0,\mathbf{e}_{t-1:1}^s),x_t,\param_0, g_{latent}(g_{latent,t-1}(\param_0,\mathbf{e}_{t-1:1}^s),\param_0,e_t^s) )\\
    &-g_{obs}(g_{obs,t-1}(\mathbf{x}_{t-1:1},\Pi_{k(n)}\param_0,\mathbf{e}_{t-1:1}^s),x_t,\Pi_{k(n)}\param_0,g_{latent}(g_{latent,t-1}(\Pi_{k(n)}\param_0,\mathbf{e}_{t-1:1}^s),\Pi_{k(n)}\param_0,e_t^s \Big\|\Big).
  \end{align*}
  Using Assumptions \ref{ass:DGPMixt}-\ref{ass:DGPmixtbis}, the following recursive relationship holds:
  \begin{align*}
    &\mathbb{E} \Big( \Big\|g_{obs}(g_{obs,t-1}(\mathbf{x}_{t-1:1},\param_0,\mathbf{e}_{t-1:1}^s),x_t,\param_0, g_{latent}(g_{latent,t-1}(\param_0,\mathbf{e}_{t-1:1}^s),\param_0,e_t^s) )\\
    &\quad -g_{obs}(g_{obs,t-1}(\mathbf{x}_{t-1:1},\Pi_{k(n)}\param_0,\mathbf{e}_{t-1:1}^s),x_t,\Pi_{k(n)}\param_0,\\ &\quad\quad\quad g_{latent}(g_{latent,t-1}(\Pi_{k(n)}\param_0,\mathbf{e}_{t-1:1}^s),\Pi_{k(n)}\param_0,e_t^s)) \Big\|\Big)
    \\&\leq \Big[ \mathbb{E} \Big( \Big\|g_{obs}(g_{obs,t-1}(\mathbf{x}_{t-1:1},\param_0,\mathbf{e}_{t-1:1}^s),x_t,\param_0, g_{latent}(g_{latent,t-1}(\param_0,\mathbf{e}_{t-1:1}^s),\param_0,e_t^s) )\\&\quad
    -g_{obs}(g_{obs,t-1}(\mathbf{x}_{t-1:1},\Pi_{k(n)}\param_0,\mathbf{e}_{t-1:1}^s),x_t,\Pi_{k(n)}\param_0,\\ &\quad\quad\quad g_{latent}(g_{latent,t-1}(\Pi_{k(n)}\param_0,\mathbf{e}_{t-1:1}^s),\Pi_{k(n)}\param_0,e_t^s)) \Big\|^2\Big) \Big]^{1/2}
    \\
    & \leq \overline{C}_1 \Big[ \mathbb{E} \Big( \Big\|g_{obs,t-1}(\mathbf{x}_{t-1:1},\param_0,\mathbf{e}_{t-1:1}^s)-g_{obs,t-1}(x_{t-1},\dots,x_1,\Pi_{k(n)}\param_0,\mathbf{e}_{t-1:1}^s) \Big\|^2\Big) \Big]^{1/2}
    \\
    &+ \overline{C_2} \|\param_0 - \Pi_{k(n)}\param_0\|^\gamma_\paramspace
    + \overline{C}_3 \Big[\mathbb{E} \Big( \Big\|g_{latent,t}(\param_0,\mathbf{e}_{t:1}^s)-g_{latent,t}(\Pi_{k(n)}\param_0,\mathbf{e}_{t:1}^s) \Big\|^2\Big) \Big]^{\gamma/2}.
  \end{align*}
  The last term also has a recursive structure:
  \begin{align*}
  &\left[\mathbb{E} \Big( \Big\|g_{latent,t}(\param_0,\mathbf{e}_{t:1}^s)-g_{latent,t}(\Pi_{k(n)}\param_0,\mathbf{e}_{t:1}^s) \Big\|^2\Big)\right]^{1/2}
  \\&\leq \overline{C}_4 \left[\mathbb{E} \Big( \Big\|g_{latent,t-1}(\param_0,\mathbf{e}_{t-1:1}^s)-g_{latent,t-1}(\Pi_{k(n)}\param_0,\mathbf{e}_{t-1:1}^s) \Big\|^2\Big)\right]^{1/2} + \overline{C}_5 \|\param_0 - \Pi_{k(n)}\param_0\|^\gamma_\paramspace.
  \end{align*}
  Together these inequalities imply:
  \begin{align*}
    &\mathbb{E} \Big( \Big\|g_{obs}(g_{obs,t-1}(x_{t-1},\dots,x_1,\param_0,\mathbf{e}_{t-1:1}^s),x_t,\param_0,g_{latent}(g_{latent,t-1}(\param_0,\mathbf{e}_{t-1:1}^s),\param_0,e_t^s) )\\&-g_{obs}(g_{obs,t-1}(x_{t-1},\dots,x_1,\Pi_{k(n)}\param_0,\mathbf{e}_{t-1:1}^s),x_t,\Pi_{k(n)}\param_0,\\ &g_{latent}(g_{latent,t-1}(\Pi_{k(n)}\param_0,\mathbf{e}_{t-1:1}^s),\Pi_{k(n)}\param_0,e_t^s \Big\|\Big)
    \\&\leq \frac{1}{1-\overline{C}_1} \left( \overline{C}_2 \|\param_0 - \Pi_{k(n)}\param_0\|^\gamma_\paramspace + \overline{C}_3 \frac{\overline{C}_5^\gamma}{(1-\overline{C_4})^\gamma} \|\param_0 - \Pi_{k(n)}\param_0\|^{\gamma^2}_\paramspace\right).
  \end{align*}
  
  Recall that $\|\tau\|_\infty \sqrt{\pi(\tau)}$ is bounded above and $\pi(\tau)^{1/4}$ is integrable so that:
  \begin{align*}
  &\int \Big| \mathbb{E} \left( e^{i\tau^\prime (\mathbf{y}_t(\param_0,\mathbf{x}_{t:1}))} - e^{i\tau^\prime (\mathbf{y}_t(\Pi_{k(n)}\param_0,\mathbf{x}_{t:1}))} \right) \Big|^2 \pi(\tau) d\tau \\ & \leq \left( \overline{C}_2 \|\param_0 - \Pi_{k(n)}\param_0\|^\gamma_\paramspace + \overline{C}_3 \frac{\overline{C}_5^\gamma}{(1-\overline{C_4})^\gamma} \|\param_0 - \Pi_{k(n)}\param_0\|^{\gamma^2}_\paramspace\right)\frac{\sup_\tau [\|\tau\|_{\infty} \sqrt{\pi(\tau)}] \int \pi(\tau)^{1/4}d\tau}{1-\overline{C}_1}.
  \end{align*}

  To conclude the proof, the difference due to $e_t^s$ needs to be bounded. In order to do so, it suffice to bound the following integral:
  \begin{align*}
  &\hspace*{-0.5cm}\int e^{i\tau^\prime (\mathbf{y}_t(y_0,u_0,\mathbf{x}_{t:1},\param_0,\mathbf{e}_{t:1}^s),\mathbf{x}_t)} \left(  \prod_{j=1}^t f_0(e_j^s)- \prod_{j=1}^t \Pi_{k(n)}f_0(e_j^s)  \right)f_{\mathbf{x}}(\mathbf{x}_{t:1})d\mathbf{e}_{t:1}^s d\mathbf{x}_{t:1}.
  \end{align*}
  A direct bound on this integral yields a term of order of $t \|f_0-\Pi_{k(n)}f_0\|_{TV}$ which increases with $t$, which is too fast to generate useful rates. Rather than using a direct bound, consider Assumptions \ref{ass:DGPMixt}-\ref{ass:DGPmixtbis}. The time-series $y_t^s$ can be approximated by another time-series term which only depends on a fixed and finite $(e_t^s,\dots,e_{t-m}^s)$ for a given integer $m \geq 1$. Making $m$ grow with $n$ at an appropriate rate allows to balance the bias $m \|f_0-\Pi_{k(n)}f_0\|_{TV}$ (computed from a direct bound) and the approximation due to $m<t$.
  
  The $m$-approximation rate of $y_t$ is now derived. Let $\param=(\theta,f) \in \paramspace$, $e_t^s,\dots,e_1^s \sim f$ and $\tilde y_t^s$ such that $\tilde y_{t-m}^s = 0, \tilde u_{t-m}^s = 0$ and then $\tilde y_{j}^s = g_{obs}(\tilde y_{j-1}^s,x_j,\param,\tilde u_j^s), \tilde u_j^s = g_{latent}(\tilde u_{j-1}^s,\param,e_j^s)$ for $t-m+1 \leq j \leq t$. Each observation $t$ is approximated by its own time-series. For observation $t-m$, by construction: $\mathbb{E} \left( \Big\|y_{t-m}^s - \tilde y_{t-m}^s\Big\| \right) =\mathbb{E} \left( \Big\|y_{t-m}^s\Big\| \right) \leq \left[ \mathbb{E} \left( \Big\|y_{t-m}^s\Big\|^2 \right) \right]^{1/2}$ and $\mathbb{E} \left( \Big\|u_{t-m}^s - \tilde u_{t-m}^s\Big\| \right) = \mathbb{E} \left( \Big\|u_{t-m}^s\Big\| \right) \leq \left[ \mathbb{E} \left( \Big\|u_{t-m}^s\Big\|^2 \right) \right]^{1/2}.$
  Then, for any $t\geq \tilde t \geq t-m$:
  \begin{align*}
  &\mathbb{E} \left( \Big\|u_{\tilde t}^s - \tilde u_{\tilde t}^s\Big\| \right) \leq \overline{C}_4 \left[ \mathbb{E} \left( \Big\|u_{\tilde t-1}^s -\tilde u_{\tilde t-1}^s \Big\|^2 \right)\right]^{1/2}\\
  &\mathbb{E} \left( \Big\|y_{\tilde t}^s - \tilde y_{\tilde t}^s\Big\| \right) \leq \overline{C}_3\overline{C}_4^{\gamma } \left[\mathbb{E} \left( \Big\|u_{\tilde t-1}^s-\tilde u_{\tilde t-1}^s \Big\|^2 \right)\right]^{\gamma/2} + \overline{C}_1 \left[ \mathbb{E}\left( \Big\|y_{\tilde t-1}^s-\tilde y_{\tilde t-1}^s \Big\|^2 \right) \right]^{1/2}.
  \end{align*}
  The previous two results and a recursion arguments leads to the following inequality:
  \begin{align}
  &\mathbb{E} \left( \Big\|u_t^s - \tilde u_t^s\Big\| \right) \leq \overline{C}_4^{m} \left[\mathbb{E} \left( \Big\|u_{t-m}^s \Big\|^2 \right)\right]^{1/2} \label{eq:subst_obs}\\
  &\mathbb{E} \left( \Big\|y_t^s - \tilde y_t^s\Big\| \right) \leq \overline{C}_3\overline{C}_4^{\gamma m} \left[ \mathbb{E} \left( \Big\|u_{t-m}^s \Big\|^2 \right)\right]^{\gamma/2} + \overline{C}_1^m  \left[ \mathbb{E} \left( \Big\|y_{t-m}^s \Big\|^2 \right) \right]^{1/2}.\label{eq:subst_latent}
  \end{align}
  For $\param=\param_0,\Pi_{k(n)}\param_0$ since the expectations are finite and bounded by assumption,\\ $\mathbb{E} \left( \Big\|y_t^s - \tilde y_t^s\Big\| \right) \leq \overline{C} \max(\overline{C}_1,\overline{C}_4)^{\gamma m}$ with $0 \leq \max(\overline{C}_1,\overline{C}_4) <1$ and some $\overline{C}>0$. For the first observations $t \leq m$ the data is unchanged, $y_t^s = \tilde y_t^s$, so that the bound still holds. The integral can be split and bounded:
  \begin{align*}
  &\Big|\int e^{i\tau^\prime (\mathbf{y}_t(y_0,u_0,\mathbf{x}_{t:1},\param_0,\mathbf{e}_{t:1}^s),\mathbf{x}_t)} \left(  \prod_{j=1}^t f_0(e_j^s) - \prod_{j=1}^t \Pi_{k(n)}f_0(e_j^s)  \right)f_{\mathbf{x}}(\mathbf{x}_{t:1})d\mathbf{e}_{t:1}^sd\mathbf{x}_{t:1}\Big|
  \\&\leq \Big| \mathbb{E} \left( [\hat \psi_n^S(\tau,\param_0)-\hat \psi_n^S(\tau,\Pi_{k(n)}\param_0)]-[\tilde \psi_n^S(\tau,\param_0)-\tilde \psi_n^S(\tau,\Pi_{k(n)}\param_0)] \right) \Big| \\&+
   \int\Big| \left(  \prod_{j=t-m+1}^t f_0(e_j^s) - \prod_{j=t-m+1}^t \Pi_{k(n)}f_0(e_j^s)  \right)d\mathbf{e}_{t:t-m+1}^s\Big| \\
   &\leq 4 \overline{C} \max(\overline{C}_1,\overline{C}_4)^{\gamma m} + 2 m \| \Pi_{k(n)}f_0 - f_0 \|_{TV}.
  \end{align*}
  The last inequality is due to the cosine and sine functions being uniformly Lipschitz continuous and equations (\ref{eq:subst_obs})-(\ref{eq:subst_latent}). Recall that $\| \Pi_{k(n)}f_0 - f_0 \|_{TV} = O(\frac{\log[k(n)]^{2r/b}}{k(n)^r})$.
  To balance the two terms, pick: $m = -\frac{r}{\gamma \log [\max(\overline{C}_1,\overline{C}_4)]} \log [k(n)] > 0$.
   Then $\max(\overline{C}_1,\overline{C}_4)^{\gamma m} = k(n)^{-r}$ and
   \[ \overline{C} \max(\overline{C}_1,\overline{C}_4)^{\gamma m} + 2 m \| \Pi_{k(n)}f_0 - f_0 \|_{TV} = O\left( \frac{\log[k(n)]^{2r/b+1}}{k(n)^r}\right).\]
  Combining all the bounds above yields:
  \begin{align*}
  Q_n(\Pi_{k(n)}\param_0) = O\left(\max \left[  \frac{\log[k(n)]^{4r/b+2}}{k(n)^{2r}},\frac{\log[k(n)]^{4 \gamma^2 r/b}}{k(n)^{2\gamma^2 r}}, \frac{1}{n^2} \right]\right)
  \end{align*}
  where $\|\cdot\|_\paramspace = \|\cdot\|_\infty$ or $\|\cdot\|_{TV}$ so that $\| \param_0 - \Pi_{k(n)}\param_0  \|^{\gamma^2}_\paramspace = O(\frac{\log[k(n)]^{4 \gamma^2 r/b}}{k(n)^{2\gamma^2 r}})$. The term due to the non-stationarity is of order $1/n^2 = o\left(\max \left[ \frac{\log[k(n)]^{4r/b+2}}{k(n)^{2r}},\frac{\log[k(n)]^{4 \gamma^2 r/b}}{k(n)^{2\gamma^2 r}} \right]\right)$ so it can be ignored. This concludes the proof.
\end{proof}

\begin{proof}[Proof of Lemma \ref{lem:cv_rate_mixture_norm}:]
  Using the inequality $1/2|a|^2 \leq |a-b|^2 +|b|^2$ for  any $a,b \in \mathbb{R}$:
  \begin{align*}
    &0 \leq 1/2 \int \Big| B \frac{d\mathbb{E}(\hat \psi_n^S(\tau,\Pi_{k(n)}\param_0))}{d\param}[\hat \param_n-\Pi_{k(n)}\param_0]\Big|^2\pi(\tau)d\tau \\
    &\leq \int \Big| B \frac{d\mathbb{E}(\hat \psi_n^S(\tau,\param_0))}{d\param}[\hat \param_n-\param_0]\Big|^2\pi(\tau)d\tau \\
    &+ \int \Big| B \frac{d\mathbb{E}(\hat \psi_n^S(\tau,\param_0))}{d\param}[\hat \param_n-\param_0]-B \frac{d\mathbb{E}(\hat \psi_n^S(\tau,\Pi_{k(n)}\param_0))}{d\param}[\hat \param_n-\Pi_{k(n)}\param_0]\Big|^2\pi(\tau)d\tau\\
    &\leq \int \Big| B \frac{d\mathbb{E}(\hat \psi_n^S(\tau,\param_0))}{d\param}[\hat \param_n-\param_0]\Big|^2\pi(\tau)d\tau
    + \int \Big| B \frac{d\mathbb{E}(\hat \psi_n^S(\tau,\Pi_{k(n)}\param_0))}{d\param}[\Pi_{k(n)} \param_0-\param_0]\Big|^2\pi(\tau)d\tau\\
    &+ \int \Big| B \frac{d\mathbb{E}(\hat \psi_n^S(\tau,\param_0))}{d\param}[\hat \param_n-\param_0]-B \frac{d\mathbb{E}(\hat \psi_n^S(\tau,\param_0))}{d\param}[\hat \param_n-\Pi_{k(n)}\param_0]\Big|^2\pi(\tau)d\tau.
  \end{align*}
  By assumption the term on the left is $O_p(\delta_n^2)$, by condition ii. the middle term is $O_p(\delta_n^2)$ and condition i. implies that the term on the right is also $O_p(\delta_n^2)$. It follows that:
  \begin{align}
    \int \Big| B \frac{d\mathbb{E}(\hat \psi_n^S(\tau,\Pi_{k(n)}\param_0))}{d\param}[\hat \param_n-\Pi_{k(n)}\param_0]\Big|^2\pi(\tau)d\tau = O_p(\delta_n^2). \label{eq:obj_appr_mixt}
  \end{align}
  Now note that both $\hat \param_n$ and $\Pi_{k(n)}\param_0$ belong to the finite dimensional space $\paramspace_{k(n)}$ parameterized by $(\theta,\omega,\mu,\sigma)$. To save space, $\hat \param_n$ will be represented by $\hat \varphi_n=(\hat \theta_n,\hat \omega_n,\hat \mu_n,\hat \sigma_n)$ and $\Pi_{k(n)}\param_0$ by $\varphi_{k(n)}=(\theta_{k(n)},\omega_{k(n)},\mu_{k(n)},\sigma_{k(n)})$.
  Using this notation, equation (\ref{eq:obj_appr_mixt}) becomes:
  \begin{align*}
    &\int \Big| B \frac{d\mathbb{E}(\hat \psi_n^S(\tau,\Pi_{k(n)}\param_0))}{d\param}[\hat \param_n-\Pi_{k(n)}\param_0]\Big|^2\pi(\tau)d\tau
    \\ &= \int \Big| B \frac{d\mathbb{E}(\hat \psi_n^S(\tau,\Pi_{k(n)}\param_0))}{d(\theta,\omega,\mu,\sigma)}[\hat \varphi_n- \varphi_{k(n)}]\Big|^2\pi(\tau)d\tau \\
    &= trace \left([\hat \varphi_n- \varphi_{k(n)}]^\prime \int  B \frac{d\mathbb{E}(\hat \psi_n^S(\tau,\Pi_{k(n)}\param_0))}{d(\theta,\omega,\mu,\sigma)}^\prime \overline{B \frac{d\mathbb{E}(\hat \psi_n^S(\tau,\Pi_{k(n)}\param_0))}{d(\theta,\omega,\mu,\sigma)}} \pi(\tau)d\tau [\hat \varphi_n- \varphi_{k(n)}] \right) \\
    &\geq \underline{\lambda}_n \|\hat \varphi_n- \varphi_{k(n)}\|^2 = \underline{\lambda}_n \|\hat \param_n-\Pi_{k(n)}\param_0\|_{m}^2.
  \end{align*}
  It follows that $0 \leq \underline{\lambda}_n  \|\hat \param_n-\Pi_{k(n)}\param_0\|_{m}^2 \leq O_p(\delta_n^2)$ so that the rate of convergence in mixture norm is:
  $\|\hat \param_n-\Pi_{k(n)}\param_0\|_{m} = O_p\left(\delta_n\underline{\lambda}_n^{-1/2} \right).$
\end{proof}

\begin{proof}[Proof of Lemma \ref{lem:stoch_eq_mixture_short}]
Using the rate assumptions and Lemma \ref{lem:stoch_eq_mixture} implies the desired result.
\end{proof}

\section{Intermediate Results} \label{apx:additional_results} 

\begin{lemma}[Kruijer, Rousseau and van der Vaart, 2010]
  \label{lem:Kruijer}
  Suppose that $f$ is a continuous univariate density satisfying:
  i) Smoothness: $f$ is $r$-times continuously differentiable with bounded $r$-th derivative. ii) Tails: $f$ has exponential tails, i.e. there exists $\bar{e}, M_{f_1}, a,b >0$ such that:
    $f_1(e) \leq M_{f_1}e^{-a|e|^b},\, \forall |e| \geq \bar{e}.$
    iii) Monotonicity in the Tails: $f$ is strictly positive and there exists $\underline{e}<\overline{e}$ such that $f_S$ is weakly decreasing on $(-\infty,\underline{e}]$ and weakly increasing on $[\overline{e},\infty)$. 
  Let $\mathcal{F}_{k}$ be the sieve space consisting of Gaussian mixtures with the following restrictions. iv) Bandwidth: $\sigma_j \geq \underline{\sigma}_{k} = O(\frac{\log[k(n)]^{2/b}}{k})$. v) Location Parameter Bounds: $\mu_j \in [-\bar{\mu}_{k},\bar{\mu}_{k}]$. vi) Growth Rate of Bounds: $\bar{\mu}_{k} = O\left(\log[k]^{1/b} \right)$. Then there exists a mixture sieve approximation of $f$, $\Pi_k f \in \mathcal{F}_k$, such that as $k\to\infty$:
  $\|f-\Pi_k f\|_{\mathcal{F}} = O \left(  \frac{\log[k(n)]^{2r/b}}{k(n)^r} \right)$,
  where $\|\cdot\|_\mathcal{F} = \|\cdot\|_{TV}$ or $\|\cdot\|_{\infty}$.
\end{lemma}

\begin{lemma}[Chen and Pouzo, 2012] \label{lem:consistency} Let $\hat \param_n$ be such that $\hat Q_n(\hat \param_n) \leq \inf_{\param \in \paramspace_{k(n)}}+O_{p^*}(\eta_n)$, where $(\eta_n)_{n\geq 1}$ is a positive real-valued sequence such that $\eta_n=o(1)$. Let $\bar Q_n : \paramspace \rightarrow [0,+\infty)$ be a sequence of non-random measurable functions and let the following conditions hold:
  a. i) $0\leq \bar Q_n(\param_0)=o(1)$; ii) there is a positive function $g_0(n,k,\varepsilon)$ such that: $\inf_{h\in\paramspace_k: \, \|\param-\param_0\|_\paramspace>\varepsilon} \bar Q_n(\param) \geq g_0(n,k,\varepsilon) >0 \text{ for each } n,k \geq 1,$ 
    and $\lim \inf_{n\to\infty} g_0(n,k(n),\varepsilon)\geq 0$ for all $\varepsilon >0$. b. i) $\paramspace$ is an infinite dimensional, possibly non-compact subset of a Banach space $(B,\|\|_\paramspace)$; ii) $\paramspace_k \subseteq \paramspace_{k+1} \subseteq \paramspace$ for all $k \geq 1$, and there is a sequence $\{ \Pi_{k(n)} \param_0 \in \paramspace_{k(n)} \}$ such that $\bar Q_n(\Pi_{k(n)} \param_0)=o(1)$. c. $\hat Q_n(\param)$ is jointly measurable in the data $(y_t,x_t)_{t \geq 1}$ and the parameter $h \in \paramspace_{k(n)}$.
    d. i) $\hat Q_n(\Pi_{k(n)} \param_0) \leq K_0 \bar Q_n(\Pi_{k(n)} \param_0)+O_{p^*}(c_{0,n})$ for some $c_{0,n}=o(1)$ and a finite constant $K_0>0$; ii) $\hat Q_n(\param) \geq K \bar Q_n(\param)-O_{p^*}(c_n)$ uniformly over $h\in \paramspace_{k(n)}$ for some $c_n=o(1)$ and a finite constant $K>0$; iii) $\max(c_{0,n},c_n,\bar Q_n(\Pi_{k(n)} \param_0),\eta_n)=o(g_0(n,k(n),\varepsilon))$ for all $\varepsilon>0$.
  Then for all $\varepsilon >0$: $\mathbb{P}^* \left( \|\hat \param_n-\param_0\|_\paramspace > \varepsilon \right) \to 0 \text{ as } n\to \infty.$
\end{lemma}

\begin{lemma} \label{lemma:covariance_ineq}
  Let $(Y_t)_{t \geq 1}$ mean zero, $\alpha$-mixing with rate $\alpha(m)$ such that $\sum_{m\geq 1} \alpha(m)^{1/p} < \infty$ for some $p > 1$, and $|Y_t| \leq 1$ for all $t\geq 1$. Then we have $\mathbb{E}\left( n| \bar Y_n |^2 \right) \leq 1 + 24 \sum_{m \geq 1} \alpha(m)^{1/p}$.
  \end{lemma}

  \begin{lemma} \label{lemma:ineq_Rio} Let $(X_t)_{t>0}$ be a sequence of real-valued, centered random variables and $(\alpha_m)_{m\geq0}$ be the sequence of strong mixing coefficients. Suppose that $X_t$ is uniformly bounded and there exists $A,C>0$ such that $\alpha(m)\leq A\exp(-Cm)$ then there exists $K>0$ that depends only on the mixing coefficients such that for any $p \geq 2$:
  \[ \mathbb{E} \left( | \sqrt{n}\bar X_n |^p \right)^{1/p} \leq K \left[ \sqrt{p} \left( \int_0^1 \min(\alpha^{-1}(u),n)\sum_{t=1}^n \frac{Q_t^2(u)}{n}\right)^{1/2} + n^{1/p-1/2} p^2 \|\sup_{t>0} X_t\|_{\infty} \right] \]
  where $Q_t$ is the quantile function of $X_t$, $\min(\alpha^{-1}(u),n) = \sum_{i=k}^{n} \mathbbm{1}_{u \leq \alpha_k}$.
  \end{lemma}

\begin{lemma} \label{lemma:maxineq_dep}
  Suppose that $(X_t(\param))_{t > 0}$ is a real valued, mean zero random process for any $\param\in\paramspace$. Suppose that it is $\alpha$-mixing with exponential decay: $\alpha(m) \leq A\exp(-Cm)$ for $A,C >0$ and bounded $|X_t(\param)| \leq 1$. Let $\mathcal{X} = \big\{ X: \paramspace \to \mathbb{C}, \param \to X_t(\param) \big\}$ and suppose that $\int_0^1 \log^2 N_{[\,]}(x,\mathcal{X},\|\cdot\|)dx < \infty$ then:
  $\int_0^1 x^{\vartheta/2-1}\sqrt{\log N_{[\,]}(x,\mathcal{X},\|\cdot\|)} + \log^2 N_{[\,]}(x,\mathcal{X},\|\cdot\|) < \infty$ for all $\vartheta\in(0,1)$ and:
  \begin{align*}
    &\mathbb{E}\big( \sup_{\param \in \paramspace} |\sqrt{n}[\hat \psi^S_t(\param)-\mathbb{E}(\hat \psi^S_t(\param))] |^2\big) \leq K \big( \int_0^1 x^{\vartheta/2-1}\sqrt{\log N_{[\,]}(x,\mathcal{X},\|\cdot\|)} + \log^2 N_{[\,]}(x,\mathcal{X},\|\cdot\|)  dx\big).
  \end{align*}
  \end{lemma}

\begin{assumptionbis}{ass:DGPMixt}[Data Generating Process - $L^2$-Smoothness] \label{ass:DGPmixtbis}
  $y_t^s$ is simulated according to the dynamic model (\ref{eq:observed})-(\ref{eq:unobserved}) where $g_{obs}$ and $g_{latent}$ satisfy the following $L^2$-smoothness conditions for some $\gamma \in (0,1]$ and any $\delta \in (0,1)$:
 \begin{enumerate}[nosep,wide=0pt]
  \item[$y(i)^\prime$.] For some $0 \leq \bar{C}_1 < 1$:\\
   $\big[ \mathbb{E} \big(\sup_{\|\param_1-\param_2\|_\paramspace \leq \delta}\|g_{obs}(y_t^s(\param_1),x_t,\param_1,u_t^s(\param_1))-g_{obs}(y_t^s(\param_2),x_t,\param_1,u_t^s(\param_1))\|^2 \big| y_t^s(\param_1),y_t^s(\param_2) \big) \big]^{1/2} \leq \bar{C}_1 \|y_t^s(\param_1)-y_t^s(\param_2)\|$
 
  \item[$y(ii)^\prime$.] For some $0 \leq \bar{C}_2 < \infty$:\\
   $\big[ \mathbb{E} \big(\sup_{\|\param_1-\param_2\|_\paramspace \leq \delta}\|g_{obs}(y_t^s(\param_1),x_t,\param_1,u_t^s(\param_1))-g_{obs}(y_t^s(\param_1),x_t,\param_2,u_t^s(\param_1))\|^2\big) \big]^{1/2} \leq \bar{C}_2\delta^\gamma$
 
  \item[$y(iii)^\prime$.] For some $0 \leq \bar{C}_3 < \infty$:\\
  $\big[\mathbb{E} \big(\sup_{\|\param_1-\param_2\|_\paramspace \leq \delta}\|g_{obs}(y_t^s(\param_1),x_t,\param_1,u_t^s(\param_1))-g_{obs}(y_t^s(\param_1),x_t,\param_1,u_t^s(\param_2))\|^2 \big| u_t^s(\param_1),u_t^s(\param_2) \big)\big]^{1/2}
     \leq \bar{C}_3 \|u_t^s(\param_1)-u_t^s(\param_2)\|^\gamma$
 
  \item[$u(i)^\prime$.] For some $0 \leq \bar{C}_4 < 1$:\\
   $\big[\mathbb{E} \big(\sup_{\|\param_1-\param_2\|_\paramspace \leq \delta}\|g_{latent}(u_{t-1}^s(\param_1),\param,e_t^s(\param_1))-g_{latent}(u_{t-1}^s(\param_2),\param,e_t^s(\param_1))\|^2\big)\big]^{1/2}
     \leq \bar{C}_4 \|u_{t-1}^s(\param_1)-u_{t-1}^s(\param_2)\|$

  \item[$u(ii)^\prime$.] For some $0 \leq \bar{C}_5 < \infty$:\\
     $\big[\mathbb{E}\big( \sup_{\|\param_1-\param_2\|_\paramspace \leq \delta} \|g_{latent}(u_{t-1}^s(\param_1),\param_1,e_t^s(\param_1))-g_{latent}(u_{t-1}^s(\param_1),\param_2,e_t^s(\param_1))\|^2 \big)\big]^{1/2} \leq \bar{C}_5 \delta^\gamma$
 
   \item[$u(iii)^\prime$.] For some $0 \leq \bar{C}_5 < \infty$:\\
   $\big[\mathbb{E}\big( \sup_{\|\param_1-\param_2\|_\paramspace \leq \delta} \|g_{latent}(u_{t-1}^s(\param_1),\param_1,e_t^s(\param_1))-g_{latent}(u_{t-1}^s(\param_1),\param_1,e_t^s(\param_2))\|^2 \big| e_t^s(\param_1),e_t^s(\param_2) \big)\big]^{1/2} \leq \bar{C}_6 \|e_t^s(\param_1)-e_t^s(\param_2)\|$
 \end{enumerate}
 for $\|\param_1-\param_2\|_\paramspace  =\|\theta_1-\theta_2\| + \|f_1-f_2\|_\infty$ or $\|\theta_1-\theta_2\| + \|f_1-f_2\|_{TV}$.
 \end{assumptionbis}

 \begin{lemma}  \label{lem:geom_ergo} Suppose that $(\mathbf{y}_t^s,\mathbf{x}_t)_{t \geq 1}$ is geometrically ergodic for $\param=\param_0$ and the moments are bounded $|\hat \psi_t^s(\tau,\param_0)| \leq M$ for all $\tau$ then $Q_n(\param_0) = O(1/n^2).$
 \end{lemma}

 \begin{lemma}[Stochastic Equicontinuity] \label{lem:stoch_eq_mixture}
  Let $M_n = \log\log(n+1)$ and $\delta_{mn} = \delta_n / \sqrt{\underline{\lambda}_n}$. Let $\Delta_n^S(\tau,\param) = \hat \psi_n^S(\tau,\param) - \mathbb{E}(\hat \psi_n^S(\tau,\param))$. Suppose that the assumptions of Lemma \ref{lem:cv_rate_mixture_norm} and the conditions for Theorem \ref{th:asymnormal_mixture} hold then for any $\eta > 0$, uniformly over $\param \in \paramspace_{k(n)}$ :
  \begin{align*}
    &\left[ \mathbb{E} \left( \sup_{\|\param-\Pi_{k(n)}\param_0\|_m \leq M_n \delta_{mn}} \Big| \Delta_n^S(\tau,\param)-\Delta_n^S(\tau,\Pi_{k(n)}\param_0)\Big|^2 \pi(\tau)^{\frac{2}{2+\eta}}  \right) \right]^{1/2} \leq C \frac{(M_n\delta_{mn})^{\frac{\gamma^2}{2}}}{\sqrt{n}} I_{m,n}
  \end{align*}
  Where $I_{m,n}$ is defined as:
  \[ I_{m,n} = \int_0^1\left( x^{-\vartheta/2}\sqrt{\log N([xM_n\delta_{mn}]^{\frac{2}{\gamma^2}},\paramspace_{k(n)},\|\cdot\|_m)}+ \log^2 N([xM_n\delta_{mn}]^{\frac{2}{\gamma^2}},\paramspace_{k(n)},\|\cdot\|_m)\right)dx \]
  For the mixture sieve the integral is a $O(k(n)\log[k(n)]+k(n)|\log(M_n\delta_{mn})|)$ so that:
  \begin{align*}
    &\left[ \mathbb{E} \left( \int  \sup_{\|\param-\Pi_{k(n)}\param_0\|_m \leq M_n \delta_{mn}} \Big| \Delta_n^S(\tau,\param)-\Delta_n^S(\tau,\Pi_{k(n)}\param_0) \Big|^2  \pi(\tau)d\tau \right) \right]^{1/2} \\
    &= O\left( (M_n\delta_{mn})^{\frac{\gamma^2}{2}}\max(\log[k(n)]^2,|\log[M_n\delta_{mn}]|^2)\frac{k(n)^2}{\sqrt{n}} \right)
  \end{align*}
  Now suppose that $(M_n\delta_{mn})^{\frac{\gamma^2}{2}}\max(\log[k(n)]^2,|\log[M_n\delta_{mn}]|^2)k(n)^2=o(1)$. The first stochastic equicontinuity result is:
  \[ \left[\mathbb{E} \left( \int  \sup_{\|\param-\Pi_{k(n)}\param_0\|_m \leq M_n \delta_{mn}} \Big| \Delta_n^S(\tau,\param)-\Delta_n^S(\tau,\Pi_{k(n)}\param_0)\Big|^2  \pi(\tau)d\tau \right) \right]^{1/2} = o(1/\sqrt{n}). \]
  Also, suppose that $\param \rightarrow \int \mathbb{E} \Big| \hat \psi_t^s(\tau,\param_0)-\hat \psi_t^s(\tau,\param) \Big|^2\pi(\tau)d\tau$ is continuous at $\param=\param_0$ under the norm $\|\cdot\|_\paramspace$, uniformly in $t \geq 1$. Then, the second stochastic equicontinuity result is:
    \[ \left[ \mathbb{E} \left( \int  \sup_{\|\param-\Pi_{k(n)}\param_0\|_m \leq M_n \delta_{mn}} \Big| \Delta_n^S(\tau,\param)-\Delta_n^S(\tau,\param_0) \Big|^2  \pi(\tau)d\tau \right) \right]^{1/2} = o(1/\sqrt{n}). \]
\end{lemma}

\begin{lemma} \label{lem:equivs}
  Suppose that $\|\hat \param_n-\param_0\|_{weak} = O_p(\delta_n)$. Under the Assumptions of Theorem \ref{th:asymnormal_mixture}:
\begin{enumerate}
  \item[a)] $\int  \psi_\beta(\tau,u_n^*) \left( \overline{ B\mathbb{E}(\hat \psi_n^S(\tau,\hat \param_n)-\hat \psi_n^S(\tau,\param_0))-B\frac{d\mathbb{E}(\hat \psi_n^S(\tau,\param_0))}{d\param}[\hat \param_n-\param_0]} \right)\pi(\tau)d\tau =o(1/\sqrt{n}).$
  \item[b)] $\int  \psi_\beta(\tau,u_n^*) \left( \overline{ B\mathbb{E}(\hat \psi_n^S(\tau,\hat \param_n)-\hat \psi_n^S(\tau,\param_0))-B[\hat \psi_n^S(\tau,\hat \param_n)-\hat \psi_n^S(\tau,\param_0)]} \right)\pi(\tau)d\tau =o(1/\sqrt{n}).$
  \item[c)] $\int  \left[ \psi_\beta(\tau,u_n^*) \left( \overline{ B[\hat \psi_n(\tau)-\hat \psi_n^S(\tau,\hat \param_n)]} \right)+ \overline{\psi_\beta(\tau,u_n^*)} \left( B[\hat \psi_n(\tau)-\hat \psi_n^S(\tau,\hat \param_n)] \right) \right]\pi(\tau)d\tau  =o(1/\sqrt{n}).$
\end{enumerate}
\end{lemma}

\section{Proofs for the Intermediate Results} \label{apx:proof_additional}

\begin{proof}[Proof of Lemma \ref{lemma:covariance_ineq}:]
  The proof follows from \citet{Davydov1968}'s inequality: let $p,q,r \geq 0, 1/p+1/q+1/r =1$, for any random variables $X,Y$:
  $|cov(X,Y)| \leq 12 \alpha(\sigma(X),\sigma(Y))^{1/p}\mathbb{E}(|X|^q)^{1/q}\mathbb{E}(|Y|^r)^{1/r}$, 
  where $\alpha(\sigma(X),\sigma(Y))$ is the mixing coefficient between $X$ and $Y$. As a result:
  \begin{align*}
    \mathbb{E}\left( n| \bar Y_n |^2 \right) & = \frac{1}{n}\sum_{t=1}^n \mathbb{E}(|X_n|^2) + \frac{1}{n} \sum_{t\neq t^\prime} cov(Y_t,Y_{t^\prime}) \leq 1 + 2 \times \frac{1}{n} \sum_{t > t^\prime} cov(Y_t,Y_{t^\prime})\\
      & \leq 1 + 24 \times \frac{1}{n} \sum_{t > t^\prime} \alpha(\sigma(Y_t),\sigma(Y_{t^\prime}))^{1/p} (\mathbb{E}|Y_t|^q)^{1/q}(\mathbb{E}|Y_{t^\prime}|^r)^{1/r}\\
      & = 1 +24  \sum_{m=1}^n \frac{n-m}{n}\alpha(m)^{1/p}
       \leq 1 +24  \sum_{m = 1}^\infty \alpha(m)^{1/p}.
  \end{align*}
  \end{proof}

\begin{proof}[Proof of Lemma \ref{lemma:ineq_Rio}:]
  Theorem 6.3 \citet{rio2000} implies the following inequality:
  \[ \mathbb{E} \left( |\sum_{t=1}^n X_t |^p\right) \leq a_ps_n^p + nb_p \int_0^1 \min(\alpha^{-1}(u),n)^{p-1}Q^p(u)du\]
  where $a_p=p4^{p+1}(p+1)^{p/2}$ and $b_p=\frac{p}{p-1} 4 ^{p+1}(p+1)^{p-1}$, $Q = \sup_{t >0}Q_t$ and\\ $s_n^2 = \sum_{t=1}^n \sum_{t^\prime =1}^n |\text{cov}(X_t,X_{t^\prime})|$. Since $X_t$ is uniformly bounded, using the results from Appendix C in \citet{rio2000}:
  $\int_0^1 \min(\alpha^{-1}(u),n)^{p-1}Q^p(u)du \leq 2 \left[ \sum_{k=0}^{n-1} (k+1)^{p-1} \alpha_k \right]\|\sup_{t >0} X_t\|_\infty.$
  Because the strong-mixing coefficients are exponentially decreasing, it implies:
  \begin{align*}
    \sum_{k=0}^{n-1} (k+1)^{p-1} \alpha_k &\leq A\exp(C)\sum_{k \geq 1} k^{p-1} \exp(-Ck) \leq A\exp(C) (p-1)^{p-1} \frac{1}{(1-\exp(-C))^{p-1}}
  \end{align*}
  And Corollary 1.1 of \citet{rio2000} yields:
  $s_n^2 \leq 4 \int_0^1 \min(\alpha^{-1}(u),n)\sum_{t=1}^n Q_k^2(u)du.$ Altogether:
  \begin{align*}
    \mathbb{E} \left( |\sqrt{n}\bar X_n|^p \right)^{1/p}  &\leq K_1(p+1)^{1/2}\left( \int_0^1 \min(\alpha^{-1}(u),n) \sum_{t=1}^n\frac{Q^2_t(u)}{n} \right)^{1/2}\\
    & + K_2n^{1/p-1/2}(p-1)^{(p-1)/p}(p+1)^{(p-1)/p}\| \sup_{t>0} X_t \|_\infty \\
    & \leq K \left( \sqrt{p}\left( \int_0^1 \min(\alpha^{-1}(u),n) \sum_{t=1}^n\frac{Q^2_t(u)}{n} \right)^{1/2} + n^{1/p-1/2}p^2 \| \sup_{t>0} X_t \|_\infty \right).
  \end{align*}
  with $K_1 \geq 2^{1/p}p^{1/p}4^{(p+1)/p}$, $K_2 \geq (p/[p-1])^{1/p}4^{(p+1)/p}2^{1/p}A\exp(C)\frac{1}{(1-\exp(-C))^{(p-1)/p}}$. Note that since $p \geq 2$, $2^{1/p} \leq \sqrt{2}, p^{1/p} \leq 1, 4^{(p+1)/p} \leq 16$, etc. The constants $K_1,K_2$ do not depend on $p$. $K$ only depends on the constants $A$ and $C$.
  \end{proof}

  \begin{proof}[Proof of Lemma \ref{lemma:maxineq_dep}:] 
    Let $Z_n(\param)=\frac{1}{\sqrt{n}}\sum_{t=1}^n X_t(\param)$, by Lemma \ref{lemma:ineq_Rio}:
  \[ \|Z_n(\param) \|_p = \mathbb{E} \left( |Z_n(\param)|^p \right)^{1/p}  \leq K \left( \sqrt{p}\frac{1}{n}\sum_{t=1}^n \| X_t(\param) \|^{\vartheta/2} + p^2 n^{-1/2+1/p}\| \sup_{t>0} X_t(\param) \|_\infty \right).\]
  The term $\frac{1}{n}\sum_{t=1}^n \| X_t(\param) \|^\vartheta$ comes from H\"older's inequality, for any $\vartheta \in (0,1)$:
  \begin{align*}
      \Big| \int_0^1 \min(\alpha^{-1}(u),n) \sum_{t=1}^n\frac{Q^2_t(u)}{n} &\Big|^{1/2} \leq \left( \int_0^1 \min(\alpha^{-1}(u),n)^{1/(1-\vartheta)} \right)^{\frac{1-\vartheta}{2}}\left( \int_0^1 |\frac{1}{n} \sum_{t=1}^n Q_t(u)^2|^{1/\vartheta} \right)^{\frac{\vartheta}{2}} \\
      &\leq \left( \frac{1}{1-\vartheta} \sum_{j=1}^n (1+j)^{1/(1-\vartheta)} \alpha(j) \right)^{\frac{1-\vartheta}{2}} \frac{1}{n} \sum_{t=1}^n\left( \int_0^1 |Q_t(u)|^{2/\vartheta} du\right)^{\frac{\vartheta}{2}}\\
      &\leq \left( \frac{1}{1-\vartheta} \sum_{j=1}^n (1+j)^{1/(1-\vartheta)} \alpha(j) \right)^{\frac{1-\vartheta}{2}} \frac{1}{n} \sum_{t=1}^n\|Q_t\|_{1}^{\vartheta/2}.
  \end{align*}
  The last inequality follows from assuming $|Q_t| \leq 1$. To simplify notation, use $\frac{1}{n} \sum_{t=1}^n\|Q_t\|_{1}^{\vartheta}$ rather than $\frac{1}{n} \sum_{t=1}^n\|Q_t\|_{1}^{\vartheta/2}$. Also since $\alpha(j)$ has exponential decay, $\sum_{j=1}^\infty (1+j)^{1/(1-\vartheta)} \alpha(j) < \infty$ so the first term is a constant which only depends on $(\alpha(j))_j$ and $\vartheta$.
  To derive the inequality, construct bracketing pairs $(\param_j^k,\Delta_j^k)_{1 \leq j \leq N(k)}$ with $N(k) = N_{[\,]}(2^{-k},\mathcal{X},\|\cdot\|_2)$ the minimal number of brackets needed to cover $\mathcal{X}$. By definition of $N(k)$ there exists brackets $(\Delta_{t,j}^k)_{j=1,\dots,N(k)}$ such that: 1) $\mathbb{E}\left( |\Delta_{t,j}^k |^2 \right)^{1/2} \leq 2^{-k}$ for all $t,j,k$. 2) For all $\param \in \paramspace$ and $k\geq 1$, there exists an index $j$ such that $|X_t(\param)-X_t(\param_j^k) \leq \Delta_{t,j}^k$.
  Note that brackets constructed the usual way need not be $\alpha$-mixing, a construction which preserve the dependence properties is given at the end of the proof.

  Assume that, without loss of generality, $|\Delta_j^k|\leq 1$ for all $j,k$. Let $(\pi_k(\param),\Delta_k(\param))$ be a bracketing pair for $\param \in \paramspace$. Let $q_0,k,q$  be positive integers such that $q_0 \leq k \leq q$ and let $T_k(\param) = \pi_k \circ \pi_{k+1} \circ \dots \circ \pi_q(\param)$.
  Using the following identity:
  \begin{align*}
    &\left[\mathbb{E} \left( \sup_{\param \in \paramspace} |Z_n(\param)|^2  \right) \right]^{1/2} \\ &= \left[\mathbb{E} \left( \sup_{\param \in \paramspace} |Z_n(\param)-Z_n(T_q(\param))+ \sum_{k=q_0+1}^q [Z_n(T_k(\param))-Z_n(T_{k-1}(\param))]+Z_n(T_{q_0}(\param))|^2  \right)\right]^{1/2}
  \end{align*}
  and the triangle inequality, decompose the identity into three groups:
  \begin{align*}
  \left[ \mathbb{E} \left( \sup_{\param \in \paramspace} |Z_n(\param)|^2  \right) \right]^{1/2} &\leq \left[\mathbb{E} \left( \sup_{\param \in \paramspace} |Z_n(\param)-Z_n(T_q(\param))|^2\right)\right]^{1/2}\\&+ \sum_{k=q_0+1}^q \left[\mathbb{E} \left( \sup_{h \in \paramspace} |Z_n(T_k(\param))-Z_n(T_{k-1}(\param))|^2\right)\right]^{1/2}\\&+ \left[\mathbb{E} \left( \sup_{\param \in \paramspace} |Z_n(T_{q_0}(\param))|^2  \right) \right]^{1/2}  \leq E_{q+1} + \sum_{k=q_0+1}^q E_{k} + E_{q_0}.
  \end{align*}
  The following inequality is due to \citet{Pisier1983}, for any $X_1,\dots,X_N$ random variables:\\ $\left[\mathbb{E}\left(  \max_{1 \leq t \leq N} |X_t|^p \right)\right]^{1/p} \leq N^{1/p}\max_{1 \leq t \leq N} \left[\mathbb{E}\left(  |X_t|^p \right)\right]^{1/p}.$
  Now that $\{ T_k(\param), \param \in \paramspace\}$ has at most $N(k)$ elements by construction. Some terms can be simplified:\\ $E_k = \mathbb{E} \left( \max_{g \in T_k(\paramspace)} |Z_n(g)-Z_n(T_{k-1}(g))|^2\right)^{1/2}$ for $q_0+1\leq k \leq q.$ For $p \geq 2$ using both H\"older and Pisier's inequalities:
  \begin{align*}
    &E_k \leq \left[\mathbb{E} \left( \sup_{\param \in T_k(\paramspace)} |Z_n(\param)-Z_n(T_{k-1}(\param))|^p\right)\right]^{1/p} \leq N(k)^{1/p} \max_{g \in T_k(\paramspace)} \left[\mathbb{E} \left( |Z_n(g)-Z_n(T_{k-1}(g))|^p\right)\right]^{1/p}.
  \end{align*}
  By the definition of $\Delta_j^k$: $E_k \leq N(k)^{1/p} \max_{1 \leq j \leq N(k)} \left[\mathbb{E} \left( |\Delta_j^k(g)|^p\right)\right]^{1/p}.$
  This is also valid for $E_{q+1}$. Using Rio's inequality for $\alpha$-mixing dependent processes:
  \begin{align*} E_k &\leq K N(k)^{1/p} \left( \sqrt{p} \max_{g \in T_k(\paramspace)}\| \Delta^k(g) \|^{\vartheta/2}_1 + p^2 n^{-1/2+1/p} \max_{g \in T_k(\paramspace)}\| \Delta^k(g) \|_\infty \right)
  \\&\leq K N(k)^{1/p} \left( \sqrt{p}2^{-\vartheta/2 k} + p^2 n^{-1/2+1/p} \right)
  \\&\leq K N(k)^{1/p} 2^{-k}\left( \sqrt{p}2^{k-\vartheta/2 k} + p^2 [n^{-1/2}2^k]^{1-2/p}2^{2k/p} \right). \end{align*}
  For $p>2$ and $2^q/\sqrt{n} \geq 1$, the inequality becomes:
  \begin{align*} E_k &\leq K N(k)^{1/p} 2^{-k}\left( \sqrt{p}2^{k-\vartheta/2 k} + p^2 [n^{-1/2}2^q]2^{2k/p} \right). \end{align*}
    Choosing $p=k+\log N(k)$ implies:
    \begin{align*}
      &N(k)^{1/p} \leq \exp(1),\,
      \sqrt{p} \leq \sqrt{k} + \sqrt{\log N(k)},\,
      p^2 \leq 4[k^2 + \log^2 N(k)],\,
      2^{2k/p} \leq 4.
    \end{align*}
    Applying these bounds to the previous inequality:
    \begin{align*} E_k &\leq 16 K \exp(1) 2^{-k}\left( [\sqrt{k}+\sqrt{\log N(k)}]2^{k-\vartheta/2 k} + [k^2+\log(N(k))^2]\frac{2^q}{\sqrt{n}} \right)
      \\ &\leq \frac{2^q}{\sqrt{n}} 16 K \exp(1) 2^{-k}\left( [\sqrt{k}+\sqrt{\log N(k)}]2^{k-\vartheta/2 k} + k^2+\log(N(k))^2 \right). \end{align*}
  Note that $\sum_{k \geq 1} (\sqrt{k}+k^2)2^{-k} \leq 2\sum_{k \geq 1} k^2 2^{-k} =12.$
  Hence:
  \begin{align*}
    &\sum_{k=q_0+1}^{q+1} E_k \leq \frac{2^{q+1}}{\sqrt{n}} 16 K \exp(1) \left( 12 + \int_0^1 [x^{\vartheta/2-1}\sqrt{\log N_{[\,]}(x,\mathcal{X},\|\cdot\|)}+\log^2 N_{[\,]}(x,\mathcal{X},\|\cdot\|)] dx\right).
  \end{align*}
  Pick the smallest integer $q$ such that $q \geq \log(n)/(2\log 2)-1$ so that $4\sqrt{n} \geq 2^q \geq \sqrt{n}/2$ and $2^q/\sqrt{n} \in [1/2,4]$.
  Only $E_{q_0}$ remains to be bounded, using Rio's inequality again:
  \begin{align*}
  \left[\mathbb{E} \left( \sup_{\param \in \paramspace} |Z_n(T_{q_0}(\param))|^2  \right)\right]^{1/2} &\leq K N(q_0)^{1/p}\left( \sqrt{p} \max_{h \in T_{q_0}(\paramspace)} \|X_1(\param) \|^\vartheta + p^2 n^{-1/2+1/p} \|X_1(\param) \|_\infty \right).
  \end{align*}
  For any $\varepsilon >0$ pick $p = \max\left( 2+\varepsilon, q_0 + \log N(q_0)\right)$ then:
$  N(q_0)^{1/p} \leq \exp(1),\,
  n^{-1/2+1/p} \leq n^{-1/2+1/(2+\varepsilon)} \leq 1.$
  Then conclude that:
  \begin{align*}
  \left[\mathbb{E} \left( \sup_{\param \in \paramspace} |Z_n(T_{q_0}(\param))|^2  \right)\right]^{1/2} &\leq 4 \exp(1) K \left( \sqrt{q_0}+\sqrt{\log N(q_0)} + q_0^2 + \log N(q_0)^2\right) \\ &\leq K^\prime \log N(q_0)^2  \leq K^\prime \int_0^1 \log^2 N_{[\,]}(x,\mathcal{X},\|\cdot\|)dx
  \end{align*}
  Hence, there exists a constant $K >0$ which only depends on $(\alpha(m))_{m >0}$ such that:
  \[
  \left[\mathbb{E} \left( \sup_{\param \in \paramspace} |Z_n(\param)|^2  \right)\right]^{1/2} \leq K \int_0^1 [x^{\vartheta/2-1}\sqrt{\log N_{[\,]}(x,\mathcal{X},\|\cdot\|)}+ \log^2 N_{[\,]}(x,\mathcal{X},\|\cdot\|)]dx.
  \]
  Let $\sqrt{C_n} =  K \int_0^1 [x^{\vartheta/2-1}\sqrt{\log N_{[\,]}(x,\mathcal{X},\|\cdot\|)}+ \log^2 N_{[\,]}(x,\mathcal{X},\|\cdot\|)]dx$, then $\mathbb{E} \left( \sup_{\param \in \paramspace} |Z_n(\param)|^2  \right) \leq C_n$ for all $n \geq 1$.

    \paragraph{Bracketing:} Because of the dynamics, the dependence of $X_t$ can vary with $\param$, which is not the case in \citet{BenHariz2005} or \citet{Andrews1994}. The following details the construction of the brackets $(\Delta_{t,j}^k)$ in the current setting. Suppose that $\param \rightarrow X_t(\param)$ is $L^p$-smooth. Let $\param_1^k,\dots,\param_{N(k)}^k$ be such that $\paramspace_{k_n} \subseteq \cup_{j=1}^{N(k)} B_{[\delta/C]^{\gamma}}(\param_j^k)$ then for $j \leq N(k)$ and some $Q\geq 2$:
    $\left[\mathbb{E} \left( \sup_{ \|\param-\param_j^k\|_\paramspace \leq [\delta/C]^\gamma} |X_t(\param)-X_t(\param_j^k)|^Q\right)\right]^{1/Q} \leq \delta.$
    Let $\Delta_{t,j}^k = \sup_{ \|\param-\param_j^k\|_\paramspace \leq [\delta/C]^\gamma} |X_t(\param)-X_t(\param_j^k)|$ then $\left[ \mathbb{E}\left( \Delta_{t,j}^{2k} \right) \right]^{1/2} \leq \left[ \mathbb{E}\left( \Delta_{t,j}^{Qk} \right) \right]^{1/Q}$ by H\"older's inequality which is smaller than $\delta$ by construction. $\left[ \mathbb{E}(|\Delta_{t,j}^k|^2) \right]^{1/2} \leq \delta=2^{-k}$ by construction. However, there is no guarantee that $(\Delta_{t,j}^k)_{t \geq 1}$ as constructed above is $\alpha$-mixing. Another construction for the bracket which preserves the mixing property is now suggested. Let $B \subseteq \paramspace$ a non-empty compact set in $\paramspace$. Note that since the $(\param_j^k)$ cover $\paramspace$, they also cover $B$. Let $\tilde \Delta_{t,j}^k$ be such that $|\frac{1}{n} \sum_{t=1}^n \tilde \Delta_{t,j}^k|=\sup_{\param \in B,\, \|\param-\param_j^k\| \leq [\delta/C]^\gamma} |\frac{1}{n} \sum_{t=1}^n X_t(\param)-X_t(\param_j^k)|$. Because $B$ is compact, the supremum is attained at some $\tilde \param_j^k \in B$. For all $t=1,\dots,n$, take $\tilde \Delta_{t,j}^k = X_t(\tilde \param_j^k)-X_t(\param_j^k)$. For each $(j,k)$  the sequence $(\tilde \Delta_{t,j}^k)_{t \geq 0}$ is $\alpha$-mixing by construction. Furthermore, by construction: $|\tilde \Delta_{t,j}^k|\leq |\Delta_{t,j}^k|$ and thus $\left[ \mathbb{E}(|\tilde \Delta_{t,j}^k|^Q) \right]^{1/Q} \leq 2^{-k}.$ These brackets, built in $B$ rather than $\paramspace$, preserve the mixing properties. The rest of the proof applied to $B$ implies:
    \begin{align*}
      &\mathbb{E}\left( \sup_{\param \in B} |\sqrt{n}[\hat \psi^S_t(\param)-\mathbb{E}(\hat \psi^S_t(\param))] |^2\right)\\ &\leq K \left( \int_0^1 x^{\vartheta/2-1}\sqrt{\log N_{[\,]}(x^{1/\gamma},B,\|\cdot\|)} + \log^2 N_{[\,]}(x^{1/\gamma},B,\|\cdot\|)  dx\right) \\ & \leq K \left( \int_0^1 x^{\vartheta/2-1}\sqrt{\log N_{[\,]}(x^{1/\gamma},\paramspace,\|\cdot\|)} + \log^2 N_{[\,]}(x^{1/\gamma},\paramspace,\|\cdot\|)  dx\right).
    \end{align*}
  For an increasing sequence of compact sets $B_{k} \subseteq B_{k+1} \subseteq \paramspace$ dense in $\paramspace$, there is an  increasing and bounded sequence:
  \begin{align*}
    &\mathbb{E}\left( \sup_{\param \in B_k} |\sqrt{n}[\hat \psi^S_t(\param)-\mathbb{E}(\hat \psi^S_t(\param))] |^2\right) \leq \mathbb{E}\left( \sup_{\param \in B_{k+1}} |\sqrt{n}[\hat \psi^S_t(\param)-\mathbb{E}(\hat \psi^S_t(\param))] |^2\right) \\ & \leq K \left( \int_0^1 x^{\vartheta/2-1}\sqrt{\log N_{[\,]}(x^{1/\gamma},\paramspace,\|\cdot\|)} + \log^2 N_{[\,]}(x^{1/\gamma},\paramspace,\|\cdot\|)  dx\right).
  \end{align*}
  This sequence is thus convergent with limit less or equal than the upper-bound. Hence, it must be that the supremum over $\paramspace$ is also bounded. It can thus be assumed that $(\Delta_{t,j}^k)_{t \geq 1}$ are $\alpha$-mixing.
  \end{proof}

  \begin{proof}[Proof of Lemma \ref{lem:geom_ergo}:]
    Since $(\mathbf{y}_t^s,\mathbf{x}_t)$ is geometrically ergodic, the joint density converges to the stationary distribution at a geometric rate: $\|f_{t}(y,x)-f^*_{t}(y,x)\|_{TV} \leq C \rho^t$, $\rho <1$. Because $B$ is bounded linear and the moments $\hat \psi_n,\hat \psi_n^s$ are bounded above by $M$, uniformly in $\tau$:
    \begin{align*}
     &Q_n(\param_0) \leq M_B^2 \int \left| \mathbb{E} \left( \hat \psi_n^S(\tau,\param_0) \right) - \lim_{n\to\infty} \mathbb{E}\left(\hat \psi_n(\tau) \right) \right|^2 \pi(\tau)d\tau\\
     &\leq M^2 M_B^2 \int \left| \frac{1}{n} \sum_{t=1}^n \int [f_{t}(y,x)-f^*_{t}(y,x)]dydx \right|^2 \pi(\tau)d\tau \\
     &\leq M^2 M_B^2  \left( \frac{1}{n} \sum_{t=1}^n \int \left| f_{t}(y,x)-f^*_{t}(y,x)\right|dydx \right)^2\\
     & \leq C M^2 M_B^2  \left( \frac{1}{n} \sum_{t=1}^n \rho^t \right)^2
      \leq \frac{C M^2 M_B^2}{(1-\rho)^2} \times \frac{1}{n^2}
     = O(1/n^2).
    \end{align*}
    \end{proof}

\begin{proof}[Proof of Lemma \ref{lem:stoch_eq_mixture}]
  Lemma \ref{lemma:maxineq_dep} implies that for some $C>0$:
  \begin{align*}
    &\left[\mathbb{E} \left( \sup_{\|\param_1-\param_2\|_m \leq \delta, \|\param_j-\Pi_{k(n)}\param_0\|_m \leq M_n \delta_{m,n},j=1,2} \Big| \hat \psi_t^s(\tau,\param_1)-\hat \psi_t^s(\tau,\param_2) \Big|^2 \right) \right]^{1/2} \frac{\sqrt{\pi(\tau)}}{(M_n\delta_{m,n})^{\gamma^2/2}}
   \\ &\leq C k(n)^{2\gamma^2} \left( \frac{\delta}{M_n\delta_{m,n}}\right)^{\gamma^2/2}. \end{align*}
  Next, apply the inequality of Lemma \ref{lemma:maxineq_dep} to generate the bound:
   \begin{align*}
    &\left[\mathbb{E} \left( \sup_{\|\param-\Pi_{k(n)}\param_0\|_m \leq M_n\delta_{m,n}} \Big| \Delta_n^S(\tau,\param)-\Delta_n^S(\tau,\Pi_{k(n)}\param_0) \Big|^2 \right)\right]^{1/2}\sqrt{\pi(\tau)} \leq \overline{C} \frac{(M_n\delta_{m,n})^{\gamma^2/2}}{\sqrt{n}} J_{m,n}
  \end{align*}
  for some $\overline{C}>0,\vartheta \in (0,1)$ and
  \begin{align*}
    &J_{m,n}  =\\&\int_0^1\left( x^{-\vartheta/2}\sqrt{\log N(\left[\frac{xM_n\delta_{mn}}{k(n)^{2\gamma^2}}\right]^{\frac{2}{\gamma^2}},\paramspace_{k(n)},\|\cdot\|_m)}+ \log^2 N(\left[\frac{xM_n\delta_{mn}}{k(n)^{2\gamma^2}}\right]^{\frac{2}{\gamma^2}},\paramspace_{k(n)},\|\cdot\|_m)\right)dx.
  \end{align*}
  Since $\int \sqrt{\pi(\tau)}d\tau < \infty$, the term on the left-hand side of the inequality can be squared and multiplied by  $\sqrt{\pi(\tau)}$. Then, taking the integral:
  \begin{align*}
    &\left[\mathbb{E} \left( \int \sup_{\|\param-\Pi_{k(n)}\param_0\|_m \leq M_n\delta_{m,n}} \Big| \Delta_n^S(\tau,\param)-\Delta_n^S(\tau,\Pi_{k(n)}\param_0) \Big|^2\pi(\tau)d\tau \right)\right]^{1/2} \leq \overline{C}_\pi \frac{(M_n\delta_{m,n})^{\gamma^2/2}}{\sqrt{n}}J_{m,n} \end{align*}
  where $\overline{C}_\pi = \overline{C}\int \sqrt{\pi(\tau)}d\tau$.
  Note that $J_{m,n}=O(k(n)^2 \max(\log[k(n)]^2,\log[M_n\delta_{m,n}]^2))$.
  
  To prove the final statement, notation will be shortened using $\Delta\hat \psi_t^s(\tau,\param) = \hat \psi_t^s(\tau,\param_0) - \hat \psi_t^s(\tau,\param)$. Note that, by applying \citet{Davydov1968}'s inequality:
  \begin{align*}
    &n\mathbb{E} \Big| \Delta\hat \psi_n^S(\tau,\Pi_{k(n)}\param_0)-\mathbb{E}[\Delta\hat \psi_n^S(\tau,\Pi_{k(n)}\param_0)] \Big|^2  
    \leq \frac{1}{n}\sum_{t=1}^n \mathbb{E} \Big| \Delta\hat \psi_t^s(\tau,\Pi_{k(n)}\param_0)-\mathbb{E}[\Delta\hat \psi_t^s(\tau,\Pi_{k(n)}\param_0)] \Big|^2 \\
    &+ \frac{24}{n}\sum_{m = 1}^n (n-m) \alpha(m)^{1/3} \max_{1\leq t \leq n}\left( \mathbb{E} \Big| \Delta\hat \psi_t^s(\tau,\Pi_{k(n)}\param_0)-\mathbb{E}[\Delta\hat \psi_t^s(\tau,\Pi_{k(n)}\param_0)] \Big|^6 \right)^{2/3}\\
    &\leq \left( 1 + 24 \sum_{m \geq 1} \alpha(m)^{1/3}\right) \max_{1\leq t \leq n}\left( \mathbb{E} \Big| \Delta\hat \psi_t^s(\tau,\Pi_{k(n)}\param_0)-\mathbb{E}[\Delta\hat \psi_t^s(\tau,\Pi_{k(n)}\param_0)] \Big|^6 \right)^{2/3}\\
    &\leq 4^{8/3} \left( 1 + 24 \sum_{m \geq 1} \alpha(m)^{1/3}\right) \max_{1\leq t \leq n}\left( \mathbb{E} \Big| \Delta\hat \psi_t^s(\tau,\Pi_{k(n)}\param_0)-\mathbb{E}[\Delta\hat \psi_t^s(\tau,\Pi_{k(n)}\param_0)] \Big|^2 \right)^{2/3}.
  \end{align*}
  The last inequality is due to $|\Delta \hat \psi_t^s(\tau,\param)|\leq 2$. By the continuity assumption the last term is a $o(1)$ when  $\|\param_0-\Pi_{k(n)}\|_\paramspace \to 0$. As a result: $\int\mathbb{E} \Big| \Delta\hat \psi_n^S(\tau,\Pi_{k(n)}\param_0)-\mathbb{E}[\Delta\hat \psi_n^S(\tau,\Pi_{k(n)}\param_0)] \Big|^2\pi(\tau)d\tau = o(1/n). $
  To conclude the proof, apply a triangle inequality and the results above:
  \begin{align*}
      &\left[\mathbb{E} \left( \int  \sup_{\|\param-\Pi_{k(n)}\param_0\|_m \leq M_n \delta_{mn}} \Big| \Delta_n^S(\tau,\param)-\Delta_n^S(\tau,\param_0) \Big|^2  \pi(\tau)d\tau \right)\right]^{1/2} \\
      &\leq \left[\mathbb{E} \left( \int  \sup_{\|\param-\Pi_{k(n)}\param_0\|_m \leq M_n \delta_{mn}} \Big| \Delta_n^S(\tau,\param)-\Delta_n^S(\tau,\Pi_{k(n)}\param_0) \Big|^2  \pi(\tau)d\tau \right)\right]^{1/2}\\
      &+ \left[ \int \mathbb{E} \left( \Big| \Delta\hat \psi_n^S(\tau,\Pi_{k(n)}\param_0)-\mathbb{E}[\Delta\hat \psi_n^S(\tau,\Pi_{k(n)}\param_0)] \Big|^2\pi(\tau)d\tau \right) \right]^{1/2} =o(1/\sqrt{n}).
  \end{align*}
  \end{proof}

\begin{proof}[Proof of Lemma \ref{lem:equivs}:] Let $R_n(\param,\param_0)=\mathbb{E}(\hat \psi_n^S(\tau,\param)-\hat \psi_n^S(\tau,\param_0))-\frac{d\mathbb{E}(\hat \psi_n^S(\tau,\param_0))}{d\param}[\param-\param_0]$.
  \begin{enumerate}[nosep,wide=0pt]
    \item[a)] Since $B$ bounded linear, the Cauchy-Schwarz inequality implies:
    \begin{align*}
      &\Big|\int  \psi_\beta(\tau,u_n^*) \big( \overline{ B\mathbb{E}(\hat \psi_n^S(\tau,\hat \param_n)-\hat \psi_n^S(\tau,\param_0))-B\frac{d\mathbb{E}(\hat \psi_n^S(\tau,\param_0))}{d\param}[\hat \param_n-\param_0]} \big)\pi(\tau)d\tau \Big|
      \\&= \Big|\int  \psi_\beta(\tau,u_n^*) \big( \overline{ B R_n(\hat \param_n,\param_0)} \big)\pi(\tau)d\tau \Big|
      \leq M_B
      \big( \int  |\psi_\beta(\tau,u_n^*)|^2\pi(\tau)d\tau \big)^{1/2} \big( \int \Big| R_n(\hat \param_n,\param_0)\Big|^2\pi(\tau)d\tau \big)^{1/2}
    \end{align*}
  
    By definition of $M_{n}$ and the inequality above:
    \begin{align*}
      & \mathbb{P} \left(  \Big|\int  \psi_\beta(\tau,u_n^*) \left( \overline{ BR_n(\hat \param_n,\param_0)} \right)\pi(\tau)d\tau \Big| > \frac{\varepsilon}{\sqrt{n}} \right) \\
      &\leq  \mathbb{P} \Bigg[ M_B^2
      \left( \int  |\psi_\beta(\tau,u_n^*)|^2\pi(\tau)d\tau \right)
      \sup_{\|\param-\param_0\|_{weak} \leq M_{n}\delta_n}\left( \int \Big| R_n(\param,\param_0) \Big|^2\pi(\tau)d\tau \right)> \frac{\varepsilon^2}{n}\Bigg]\\
      &+ \mathbb{P}\left( \|\hat \param_n-\param_0\|_\paramspace > M_{n}\delta_n \right)
    \end{align*}
    $\mathbb{P}\left( \|\hat \param_n-\param_0\|_\paramspace > M_{n}\delta_n \right)\to 0$ regardless of $\varepsilon$. Furthermore, Assumption \ref{ass:sufficient_cv_rate_mixture} \textit{ii.} implies:
    \begin{align*}
      &\sup_{\|\param-\param_0\|_{weak} \leq M_{n}\delta_n}\left( \int \Big| \mathbb{E}(\hat \psi_n^S(\tau,\param)-\hat \psi_n^S(\tau,\param_0))-\frac{d\mathbb{E}(\hat \psi_n^S(\tau,\param_0))}{d\param}[\param-\param_0]\Big|^2\pi(\tau)d\tau \right)^{1/2} \\&= \sup_{\|\param-\param_0\|_{weak} \leq M_{n}\delta_n}\left( \int \Big| R_n(\param,\param_0)\Big|^2\pi(\tau)d\tau \right)^{1/2} = O\left((M_{n}\delta_n)^2\right).
    \end{align*}
    Assumption \ref{ass:sufficient_cv_rate_mixture} \textit{i.} implies that $(M_{n}\delta_n)^{2}  =o(\frac{1}{\sqrt{n}})$, and thus: $\mathbb{P} \big(  \Big|\int  \psi_\beta(\tau,u_n^*) \big( \overline{ BR_n(\hat \param_n,\param_0)} \big)\pi(\tau)d\tau \Big| > \frac{\varepsilon}{\sqrt{n}} \big) = o(1)$
    regardless of $\varepsilon>0$. Hence: $\int  \psi_\beta(\tau,u_n^*) \big( \overline{ BR_n(\hat \param_n,\param_0)} \big)\pi(\tau)d\tau = o_p(1/\sqrt{n}).$
    \item[b)] Let $\Delta_n^S(\tau,\param) = \hat \psi_n^S(\tau,\param) - \mathbb{E}[\hat \psi_n^S(\tau,\param)]$. By the second stochastic equicontinuity result of Lemma \ref{lem:stoch_eq_mixture} and the Cauchy-Schwarz inequality:
    \begin{align*}
      &\Big|\int  \psi_\beta(\tau,u_n^*) \left( \overline{ B[\Delta_n^S(\hat \param_n)-\Delta_n^S(\param_0)]} \right)\pi(\tau)d\tau\Big| \\
      &\leq \left( \int  |\psi_\beta(\tau,u_n^*)|^2 \pi(\tau)d\tau \right)^{1/2} \left( \int \Big|B[\Delta_n^S(\hat \param_n)-\Delta_n^S(\param_0)] \Big|^2\pi(\tau)d\tau \right)^{1/2}\\
      &\leq M_B\left( \int  |\psi_\beta(\tau,u_n^*)|^2 \pi(\tau)d\tau \right)^{1/2} \left( \int \Big|[\Delta_n^S(\hat \param_n)-\Delta_n^S(\param_0)] \Big|^2\pi(\tau)d\tau \right)^{1/2}\\
      &\leq M_B \left( \int  |\psi_\beta(\tau,u_n^*)|^2 \pi(\tau)d\tau \right)^{1/2} \left( \sup_{\|\param - \Pi_{k(n)}\param_0\|\leq M_n \delta_{mn}}\int \Big|[\Delta_n^S( \param)-\Delta_n^S(\param_0)] \Big|^2\pi(\tau)d\tau \right)^{1/2}\\
      &=o_p(1/\sqrt{n}),
    \end{align*}
    where the last inequality holds with probability going to $1$ by definition of $M_n \delta_{mn}$.
    \item[c)]
    Let $\varepsilon_n = \pm \frac{1}{\sqrt{n}M_n} =o(\frac{1}{\sqrt{n}})$. For $h \in (0,1)$ define $\hat\param(h)  =\hat \param_n + h \varepsilon_n u_n^*$. Since $\hat \param_n = \hat\param(0)$. Recall that $\hat\param_n$ is the approximate minimizer of $\hat Q_n^s$ so that: $0 \leq \hat Q_n^S(\hat \param_n) \leq \inf_{\param \in \paramspace_{k(n)}} \hat Q_n^S(\param) + O_p(\eta_n).$
    Hence the following holds:
    \begin{align}
       0 &\leq \frac{1}{2} \left( \hat Q_n^S(\hat\param(1))- \hat Q_n^S(\hat\param(0))\right) +O_p(\eta_n) \label{eq:term1}\\
      &=\frac{1}{2} \Big[ \int B\left( \hat \psi_n(\tau) - \hat \psi_n^S(\tau,\hat \param(0)) \right)\overline{B\left( \hat \psi_n^S(\tau,\hat \param(0)) - \hat \psi_n^S(\tau,\hat \param(1)) \right)}\pi(\tau)d\tau \label{eq:term2}\\
      &+ \int \overline{B\left( \hat \psi_n(\tau) - \hat \psi_n^S(\tau,\hat \param(0)) \right)}B\left( \hat \psi_n^S(\tau,\hat \param(0)) - \hat \psi_n^S(\tau,\hat \param(1))\right) \pi(\tau)d\tau \label{eq:term3}\\
      &+\int \Big|B\left( \hat \psi_n^S(\tau,\hat \param(0))-\hat \psi_n^S(\tau,\hat \param(1)) \right)\Big|^2\pi(\tau)d\tau \Big] + O_p(\eta_n). \label{eq:term4}
    \end{align}
    To prove Lemma \ref{lem:equivs} c),  (\ref{eq:term2})-(\ref{eq:term3}) are expanded individually and shown to be  $o_p(1/\sqrt{n})$ and  (\ref{eq:term4}) is bounded, shown to be negligible under the assumptions.
  
    The first step deals with (\ref{eq:term4}):
    \begin{align*}
      &\left(\int \Big|B\left( \hat \psi_n^S(\tau,\hat \param(0))-\hat \psi_n^S(\tau,\hat \param(1)) \right)\Big|^2\pi(\tau)d\tau \right)^{1/2} \leq M_B\left(\int \Big|\hat \psi_n^S(\tau,\hat \param(0))-\hat \psi_n^S(\tau,\hat \param(1)) \Big|^2\pi(\tau)d\tau \right)^{1/2}\\
      &\leq \left(\int \Big| [\hat \psi_n^S(\tau,\hat \param(0))-\hat \psi_n^S(\tau,\hat\param(1))]-\mathbb{E}[\hat \psi_n^S(\tau,\hat \param(0))-\hat \psi_n^S(\tau,\hat\param(1))] \Big|^2\pi(\tau)d\tau \right)^{1/2}\\
      &+\left(\int \Big|\mathbb{E}[\hat \psi_n^S(\tau,\hat \param(0))-\hat \psi_n^S(\tau,\hat\param(1))] \Big|^2\pi(\tau)d\tau \right)^{1/2}
    \end{align*}
    By the triangle inequality and the stochastic equicontinuity results  from Lemma \ref{lem:stoch_eq_mixture}:
    \begin{align*}
      &\left(\int \Big| [\hat \psi_n^S(\tau,\hat \param(0))-\hat \psi_n^S(\tau,\hat\param(1))]-\mathbb{E}[\hat \psi_n^S(\tau,\hat \param(0))-\hat \psi_n^S(\tau,\hat\param(1))] \Big|^2\pi(\tau)d\tau \right)^{1/2} \\ &= O_p\left(\frac{I_{m,n}(M_{n}\delta_{mn})^{\gamma^2/2}}{\sqrt{n}}\right).
    \end{align*}
    Also, note that $\hat \param(1)=\hat \param(0)+\varepsilon_n u_n^*$, so that the Mean Value Theorem applies to last term:
    \begin{align*}
      &\left(\int \Big|\mathbb{E}[\hat \psi_n^S(\tau,\hat \param(0))-\hat \psi_n^S(\tau,\hat\param(1))] \Big|^2\pi(\tau)d\tau \right) = \left(\int \Big|\frac{d\mathbb{E}[\hat \psi_n^S(\tau,\hat \param(\tilde h))}{d\param}[\varepsilon_n u_n^*] \Big|^2\pi(\tau)d\tau \right)
    \end{align*}
    for some intermediate value $\tilde h\in(0,1)$. Also, by assumption:
    $\big(\int \Big|\frac{d\mathbb{E}[\hat \psi_n^S(\tau,\hat \param(\tilde h))}{d\param}[u_n^*] \Big|^2\pi(\tau)d\tau \big)^{1/2} = O_p(1).$
    Together these two imply:
    $\big(\int \Big|\mathbb{E}[\hat \psi_n^S(\tau,\hat \param(0))-\hat \psi_n^S(\tau,\hat\param(1))] \Big|^2\pi(\tau)d\tau \big)^{1/2} = O(\varepsilon_n).$
    This yields the bound for (\ref{eq:term4}):
    \begin{align*}
      \int \Big|B\left( \hat \psi_n^S(\tau,\hat \param(0))-\hat \psi_n^S(\tau,\hat \param(1)) \right)\Big|^2\pi(\tau)d\tau \leq O_p(\varepsilon_n^2) + O_p\left(\frac{(M_{n}\delta_{mn})^{\gamma^2}I_{m,n}^2}{n}\right).
    \end{align*}

    The remaining terms, (\ref{eq:term2})-(\ref{eq:term3}), are conjugates of each other. A bound for (\ref{eq:term2}) is also valid for (\ref{eq:term3}). Expanding (\ref{eq:term2}) yields:
    \begin{align}
        & \int B\left( \hat \psi_n(\tau) - \hat \psi_n^S(\tau,\hat \param(0)) \right)\overline{B\left( \hat \psi_n^S(\tau,\hat \param(0)) - \hat \psi_n^S(\tau,\hat \param(1)) \right)}\pi(\tau)d\tau \tag{\ref{eq:term2}}\\
        & = \int B\left( \hat \psi_n(\tau) - \hat \psi_n^S(\tau,\hat \param(0)) \right)\left[ \overline{B\left( \Delta_n^S(\tau,\hat \param(0)) - \Delta_n^S(\tau,\hat \param(1)) \right)} \right]\pi(\tau)d\tau \label{eq:term5} \\
        & +\int B\left( \hat \psi_n(\tau) - \hat \psi_n^S(\tau,\hat \param(0)) \right)\overline{B\mathbb{E}\left( \hat \psi_n^S(\tau,\hat \param(0)) - \hat \psi_n^S(\tau,\hat \param(1)) \right)}\pi(\tau)d\tau. \label{eq:term6}
    \end{align}
    Applying the Cauchy-Schwarz inequality to (\ref{eq:term5}) implies:
    \begin{align}
      &\Big|\int B\left( \hat \psi_n(\tau) - \hat \psi_n^S(\tau,\hat \param(0)) \right)\left[ \overline{B\left( \Delta_n^S(\tau,\hat \param(0)) - \Delta_n^S(\tau,\hat \param(1)) \right)} \right]\pi(\tau)d\tau\Big| \tag{\ref{eq:term5}} \\
      &\leq M_B\left( \int \Big| B\hat \psi_n(\tau) - B\hat \psi_n^S(\tau,\hat \param(0)) \Big|^2\pi(\tau) d\tau \right)^{1/2} \label{eq:term6a}\\
      &\times \left( \int \Big| \Delta_n^S(\tau,\hat \param(0)) - \Delta_n^S(\tau,\hat \param(1)) \Big|^2 \pi(\tau)d\tau \right)^{1/2} \label{eq:term7}
    \end{align}
    The term (\ref{eq:term6a}) can be bounded above using the triangle inequality:
    \begin{align*}
      & \big( \int \Big| B\hat \psi_n(\tau) - B\hat \psi_n^S(\tau,\hat \param(0)) \Big|^2\pi(\tau) d\tau \big)^{1/2} \\
      &\leq M_B\big( \int \Big| \hat \psi_n(\tau) - \hat \psi_n^S(\tau,\param_0) \Big|^2\pi(\tau) d\tau \big)^{1/2} +\big( \int \Big| B\hat \psi^S_n(\tau,\param_0) - B\hat \psi_n^S(\tau,\hat \param(0)) \Big|^2\pi(\tau) d\tau \big)^{1/2}.
    \end{align*}
    An application of Lemma \ref{lemma:covariance_ineq} and the geometric ergodicity of $(\mathbf{y}_t^s,\mathbf{x}_t)$ yields:\\
    $\left( \int \Big| \hat \psi_n(\tau) - \hat \psi_n^S(\tau,\param_0) \Big|^2\pi(\tau) d\tau \right)^{1/2}=O_p(1/\sqrt{n}).$
   Then, expanding the term in $\hat \psi_n^s$:
    \begin{align*}
      &\left( \int \Big| B\hat \psi^S_n(\tau,\param_0) - B\hat \psi_n^S(\tau,\hat \param(0)) \Big|^2\pi(\tau) d\tau \right)^{1/2} \leq \left( \int \Big| B\mathbb{E}[\hat \psi^S_n(\tau,\param_0) - \hat \psi_n^S(\tau,\hat \param(0))] \Big|^2\pi(\tau) d\tau \right)^{1/2}\\
      &+ M_B\left( \int \Big| [\hat \psi^S_n(\tau,\param_0) - \hat \psi_n^S(\tau,\hat \param(0))]- \mathbb{E}[\hat \psi^S_n(\tau,\param_0) - \hat \psi_n^S(\tau,\hat \param(0))] \Big|^2\pi(\tau) d\tau \right)^{1/2}\\
      &\leq \left( \int \Big| B\mathbb{E}[\hat \psi^S_n(\tau,\param_0) - \hat \psi_n^S(\tau,\hat \param(0))] \Big|^2\pi(\tau) d\tau \right)^{1/2} + O_p\left(\frac{(M_{n}\delta_{mn})^{\gamma^2/2} I_{m,n}}{\sqrt{n}}\right)\\
      &\leq M_B\left( \int \Big| \mathbb{E}[\hat \psi^S_n(\tau,\param_0) - \hat \psi_n^S(\tau,\hat \param(0))] - \frac{d\mathbb{E}(\hat \psi_n^S(\tau,\param_0))}{d\param}[\param_0-\hat \param(0)] \Big|^2\pi(\tau) d\tau \right)^{1/2} \\
      &+ \left( \int \Big|B\frac{d\mathbb{E}(\hat \psi_n^S(\tau,\param_0))}{d\param}[\param_0-\hat \param(0)] \Big|^2\pi(\tau) d\tau \right)^{1/2} + O_p\left(\frac{(M_{n}\delta_{mn})^{\gamma^2/2} I_{m,n}}{\sqrt{n}}\right).
    \end{align*}
    Note that Assumption \ref{ass:sufficient_cv_rate_mixture} ii. implies that:
    \begin{align*}
      &\left( \int \Big| \mathbb{E}[\hat \psi^S_n(\tau,\param_0) - \hat \psi_n^S(\tau,\hat \param(0))] - \frac{d\mathbb{E}(\hat \psi_n^S(\tau,\param_0))}{d\param}[\param_0-\hat \param(0)] \Big|^2\pi(\tau) d\tau \right)^{1/2} = O_p(M_n\delta_n).
    \end{align*}
    By definition of the weak norm: $\left( \int \Big|B\frac{d\mathbb{E}(\hat \psi_n^S(\tau,\param_0))}{d\param}[\param_0-\hat \param(0)] \Big|^2\pi(\tau) d\tau \right)^{1/2} = \|\hat \param_n - \param_0\|_{weak}.$
    Furthermore, $\|\hat \param_n - \param_0\|_{weak} = O_p(\delta_n)$ by assumption.
    Overall, the following bound holds for (\ref{eq:term6}): $
      \left( \int \Big| B\hat \psi_n(\tau) - B\hat \psi_n^S(\tau,\hat \param(0)) \Big|^2\pi(\tau) d\tau \right)^{1/2} \leq O_p\left(\frac{1}{\sqrt{n}}\right) + O_p\left(\delta_{n}\right) + O_p\left(\frac{(M_{n}\delta_{mn})^{\gamma^2/2} I_{m,n}}{\sqrt{n}}\right).$
    Re-arranging (\ref{eq:term7}) to apply the stochastic equicontinuity result again yields:
    \begin{align*}
      &\left( \int \Big| \Delta_n^S(\tau,\hat \param(0)) - \Delta_n^S(\tau,\hat \param(1))  \Big|^2 \pi(\tau)d\tau \right)^{1/2} 
      \leq \left( \int \Big| \Delta_n^S(\tau,\param_0) - \Delta_n^S(\tau,\hat \param(1)) \Big|^2 \pi(\tau)d\tau \right)^{1/2}\\
      &+ \left( \int \Big| \Delta_n^S(\tau,\param_0) - \Delta_n^S(\tau,\hat \param(0)) \Big|^2 \pi(\tau)d\tau \right)^{1/2}
      = O_p\left(\frac{(M_{n}\delta_{mn})^{\gamma^2/2} I_{m,n}}{\sqrt{n}}\right).
    \end{align*}
    Using the bounds for (\ref{eq:term6}) and (\ref{eq:term7}) yields the bound for (\ref{eq:term5}):
    \begin{align*}
      &\Big|\int B\left( \hat \psi_n(\tau) - \hat \psi_n^S(\tau,\hat \param(0)) \right)\left[ \overline{B\left( \Delta_n^S(\tau,\hat \param(0)) - \Delta_n^S(\tau,\hat \param(1)) \right)} \right]\pi(\tau)d\tau\Big|\\
      &\leq O_p\left(\frac{(M_{n}\delta_{mn})^{\gamma^2/2} I_{m,n}}{\sqrt{n}}\right)O_p\left(\max\left(M_n\delta_{n},\frac{1}{\sqrt{n}},\frac{(M_{n}\delta_{mn})^{\gamma^2/2} I_{m,n} }{\sqrt{n}}\right)\right).
    \end{align*}
    To bound (\ref{eq:term6}), apply the Mean Value theorem up to the second order:
    \begin{align*}
      &\int B\left( \hat \psi_n(\tau) - \hat \psi_n^S(\tau,\hat \param(0)) \right)\overline{B\mathbb{E}\left( \hat \psi_n^S(\tau,\hat \param(0)) - \hat \psi_n^S(\tau,\hat \param(1)) \right)}\pi(\tau)d\tau\\
      &=-\int B\left( \hat \psi_n(\tau) - \hat \psi_n^S(\tau,\hat \param(0)) \right)\overline{ B\frac{d\mathbb{E}( \hat \psi_n^S(\tau,\hat \param(0) ))}{d\param}[\varepsilon_n u_n^*] }\pi(\tau)d\tau\\
      &+\frac{1}{2}\int B\left( \hat \psi_n(\tau) - \hat \psi_n^S(\tau,\hat \param(0)) \right)\overline{ B\frac{d^2\mathbb{E}( \hat \psi_n^S(\tau,\hat \param(\tilde h) ))}{d\param d\param}[\varepsilon_n u_n^*,\varepsilon_n u_n^*]  }\pi(\tau)d\tau\\
      &= -\int B\left( \hat \psi_n(\tau) - \hat \psi_n^S(\tau,\hat \param(0)) \right)\overline{B\frac{d\mathbb{E}( \hat \psi_n^S(\tau,\param_0 ))}{d\param}[\varepsilon_n u_n^*] }\pi(\tau)d\tau+O_p(\varepsilon_n^2)\\
      &+\int B\left( \hat \psi_n(\tau) - \hat \psi_n^S(\tau,\hat \param(0)) \right) \overline{B\left[ \frac{d\mathbb{E}( \hat \psi_n^S(\tau,\hat \param(0) ))}{d\param}[\varepsilon_n u_n^*]-\frac{d\mathbb{E}( \hat \psi_n^S(\tau,\param_0 ))}{d\param}[\varepsilon_n u_n^*]  \right]}\pi(\tau)d\tau.
    \end{align*}
    Where the $O_p(\varepsilon_n^2)$ term is due to the Cauchy-Schwarz inequality and Assumption \ref{ass:sufficient_cv_rate_mixture} ii.:
    \begin{align*}
      &\Big|\int B\left( \hat \psi_n(\tau) - \hat \psi_n^S(\tau,\hat \param(0)) \right)\overline{ \frac{1}{2}B\frac{d^2\mathbb{E}( \hat \psi_n^S(\tau,\hat \param(\tilde t) ))}{d\param d\param}[\varepsilon_n u_n^*,\varepsilon_n u_n^*]  }\pi(\tau)d\tau\Big|^2\\
      &\leq \frac{\varepsilon_n^2}{2}\left( \int \Big|B\left( \hat \psi_n(\tau) - \hat \psi_n^S(\tau,\hat \param(0)) \right)\Big|^2\pi(\tau)d\tau \right)\int \Big| B\frac{d^2\mathbb{E}( \hat \psi_n^S(\tau,\hat \param(\tilde t) ))}{d\param d\param}[ u_n^*,u_n^*]  \Big|^2 \pi(\tau)d\tau.
    \end{align*}
    It was shown above that:
    \[ \left( \int \Big|B\left( \hat \psi_n(\tau) - \hat \psi_n^S(\tau,\hat \param(0)) \right)\Big|^2\pi(\tau)d\tau \right) = O_p\left(\max\left(M_n\delta_{n},\frac{1}{\sqrt{n}},\frac{(M_{n}\delta_{mn})^{\gamma^2/2} I_{m,n} }{\sqrt{n}}\right)^2\right).\]
    Also, by Assumption \ref{ass:sufficient_cv_rate_mixture} \textit{ii.}:
     $\left( \int \Big| B\frac{d^2\mathbb{E}( \hat \psi_n^S(\tau,\hat \param(\tilde t) ))}{d\param d\param}[ u_n^*,u_n^*]  \Big|^2 \pi(\tau)d\tau\right) = O_p(1).$
  
    Finally, applying the Cauchy-Schwarz inequality to the last term of the expansion of (\ref{eq:term6}) yields:
    \begin{align*}
      &\int B\left( \hat \psi_n(\tau) - \hat \psi_n^S(\tau,\hat \param(0)) \right) \left[ \overline{B\frac{d\mathbb{E}( \hat \psi_n^S(\tau,\hat \param(0) ))}{d\param}[\varepsilon_n u_n^*]-B\frac{d\mathbb{E}( \hat \psi_n^S(\tau,\param_0 ))}{d\param}[\varepsilon_n u_n^*] } \right]\pi(\tau)d\tau\\
      &\leq   \left(\int \Big|B \hat \psi_n(\tau) - \hat \psi_n^S(\tau,\hat \param(0))\Big|^2\pi(\tau)d\tau \right)^{1/2} \\
      &\times  \varepsilon_n\left(\int \Big|B\frac{d\mathbb{E}( \hat \psi_n^S(\tau,\hat \param(0) ))}{d\param}[u_n^*]-B\frac{d\mathbb{E}( \hat \psi_n^S(\tau,\param_0 ))}{d\param}[u_n^*] \Big|^2\pi(\tau)d\tau \right)^{1/2}\\
      &= O_p\left(\varepsilon_n \max\left(M_n\delta_{n},\frac{1}{\sqrt{n}},\frac{(M_{n}\delta_{mn})^{\gamma^2/2} I_{m,n}}{\sqrt{n}} \right)\delta_{n}\right).
    \end{align*}
    Using inequality (\ref{eq:term1}) together with the bounds above and the expansions of (\ref{eq:term2}) and (\ref{eq:term3}) yields:
    \begin{align*}
       0 \leq  &-\varepsilon_n \int B\left( \hat \psi_n(\tau) - \hat \psi_n^S(\tau,\hat \param(0)) \right)\overline{B\frac{d\mathbb{E}( \hat \psi_n^S(\tau,\param_0 ))}{d\param}[ u_n^*]}\pi(\tau)d\tau\\
       &-\varepsilon_n \int \overline{ B\left( \hat \psi_n(\tau) - \hat \psi_n^S(\tau,\hat \param(0)) \right)} B\frac{d\mathbb{E}( \hat \psi_n^S(\tau,\param_0 ))}{d\param}[ u_n^*] \pi(\tau)d\tau\\
       &+O_p\left(\varepsilon_n^2\right)+O_p\left(\frac{(M_{n}\delta_{mn})^{\gamma^2/2} I_{m,n}}{\sqrt{n}}\max(M_n\delta_{n},\frac{1}{\sqrt{n}},\frac{(M_{n}\delta_{mn})^{\gamma^2/2} I_{m,n}}{\sqrt{n}}) )\right)\\
       &+O_p\left(\varepsilon_n M_n\delta_{n}\max(M_n\delta_{n},\frac{1}{\sqrt{n}},\frac{(M_{n}\delta_{mn})^{\gamma^2/2} I_{m,n}}{\sqrt{n}}) \right)+O_p\left(\frac{[(M_{n}\delta_{mn})^{\gamma^2/2} I_{m,n}]^2}{n}\right)
    \end{align*}
  Since $\varepsilon_n = \pm \frac{1}{\sqrt{n}M_n}$, dividing by $\varepsilon_n$ both keeps and flips the inequality so that:
  \begin{align*}
     &\int B\left( \hat \psi_n(\tau) - \hat \psi_n^S(\tau,\hat \param_n) \right)\overline{B\frac{d\mathbb{E}( \hat \psi_n^S(\tau,\param_0 ))}{d\param}[ u_n^*]}\pi(\tau)d\tau\\
     &+\int \overline{ B\left( \hat \psi_n(\tau) - \hat \psi_n^S(\tau,\hat \param_n) \right)} B\frac{d\mathbb{E}( \hat \psi_n^S(\tau,\param_0 ))}{d\param}[ u_n^*] \pi(\tau)d\tau\\
     &=O_p(\varepsilon_n)+O_p\left(\frac{(M_{n}\delta_{mn})^{\gamma^2/2} I_{m,n}}{\varepsilon_n\sqrt{n}}\max\left(M_n\delta_{n},\frac{1}{\sqrt{n}},\frac{(M_{n}\delta_{mn})^{\gamma^2/2} I_{m,n}}{\sqrt{n}}) \right)\right)\\
     &+O_p\left(\max\left(M_n\delta_{n},\frac{1}{\sqrt{n}},\frac{(M_{n}\delta_{mn})^{\gamma^2/2} I_{m,n}}{\sqrt{n}}\right) \delta_{n}\right)+O_p\left(\frac{[(M_{n}\delta_{mn})^{\gamma^2/2} I_{m,n}]^2}{\varepsilon_n n}\right).
  \end{align*}
  By construction, $\varepsilon_n = o_p(1/\sqrt{n})$ and Assumption \ref{ass:sufficient_cv_rate_mixture} i. implies that $(M_{n}\delta_{mn})^{\gamma^2/2} I_{m,n} = o(1)$ so that all terms above are $o(1/\sqrt{n})$.
  To conclude the proof, note that:
  \begin{align*}
       &\int B\left( \hat \psi_n(\tau) - \hat \psi_n^S(\tau,\hat \param_n) \right)\overline{B\frac{d\mathbb{E}( \hat \psi_n^S(\tau,\param_0 ))}{d\param}[ u_n^*]}\pi(\tau)d\tau\\
       &+\int \overline{ B\left( \hat \psi_n(\tau) - \hat \psi_n^S(\tau,\hat \param_n) \right)} B\frac{d\mathbb{E}( \hat \psi_n^S(\tau,\param_0 ))}{d\param}[ u_n^*] \pi(\tau)d\tau\\
       &= \int [\psi_\param(\tau,u_n^*)\left( \overline{B[\hat \psi_n(\tau)-\hat \psi_n^S(\tau,\hat \param_n)}]\right)+ \overline{\psi_\param(\tau,u_n^*)}\left( B[\hat \psi_n(\tau)-\hat \psi_n^S(\tau,\hat \param_n)]\right)]
       =o_p(1/\sqrt{n}).
  \end{align*}
  \end{enumerate}
  
  \end{proof}

  \section{Additional Results for the Applications}  \label{apx:add_appli}
  
  \subsection{Verifying the Primitive Conditions in the First Application} \label{apx:check_conditions}
  Recall the data generating process used in Sections \ref{sec:MonteCarlo} and \ref{sec:empirical}:
  \begin{align}
    y_t &= \mu_y + \rho_y (y_{t-1} - \mu_y) + \sigma_t (e_{1,t} + \vartheta_y e_{1,t-1}), \, \sigma_t^2 = \mu_\sigma + \rho_\sigma \sigma_{t-1}^2 + \kappa_\sigma e_{2,t}, \tag{\ref{eq:SV2}}
  \end{align}
The following verifies 1) the identification condition, that is for any $L \geq \underline{L}$, to be determined, Assumption \ref{ass:sid} ii holds if $f$ has sub-exponential tails, as required in Assumption \ref{ass:sid} i, and 2) that Assumption \ref{ass:DGPMixt} is satisfied. Geometric ergodicity can be verified by checking if Assumption 2.1 and the additional condition in Theorem 3.1 of \citet{cline1999} hold. Using their notation, $\alpha(\cdot)$ is linear and $\gamma(\cdot)$ is a product so the required conditions are verified. 
\paragraph{Identification:} Assume $e_{1,t} \sim f$ with $\mathbb{E}(e_{1,t})=0,\mathbb{E}(e_{1,t}^2)=1$ and $e_{2,t} \sim f_2$ a non-negative, known distribution with finite moment of order $p$ for any $p \geq 1$, and $\mathbb{E}(e_{2,t}) = \text{var}(e_{2,t})=1$. Assume $\rho_\sigma \in [0,1)$, $\mu_\sigma \geq 0$ and $\kappa_\sigma > 0$. For $L \geq 1$, let $\mathbf{y}_t = (y_t,\dots,y_{t-L})$ and $\psi(\tau,\theta,f) = \int \exp(i\tau^\prime \mathbf{y}_t)f(\mathbf{y}_t,\theta,f)d \mathbf{y}_t$, note that $\partial_\tau \psi(0,\theta,f) = i \mathbb{E}(\mathbf{y}_t)= i(\mu_y,\dots,\mu_y)$ so that $\mu_y$ is identified. Similar any joint moments of $\mathbf{y}_t$ can be recovered from the CF $\psi$. It suffices to show that moments spanned by $\mathbf{y}_t$ can be used to identify $(\theta,f)$. The coefficient $\rho_y$ is identified by the moment condition $\mathbb{E}( [y_t - \mu_y - \rho_y(y_{t-1}-\mu_y)]y_{t-2} )=0$. Take $\tilde{y}_t = y_t - \mu_y - \rho_y(y_{t-1}-\mu_y)$, we have: $\tilde{y}_t = \sigma_t[e_t + \vartheta_y e_{t-1}]$. 

Compute two more moments: $\mathbb{E}(\tilde{y}_t^2) = \mathbb{E}(\sigma_t^2)(1+\vartheta^2)$, and ${E}(\tilde{y}_t\tilde{y}_{t-1}) = \vartheta\mathbb{E}(\sigma_t\sigma_{t-1}).$
Unlike the MA(1) with time-invariant volatility, these two moments alone are not sufficient to identify $\vartheta$ because $|\mathbb{E}(\sigma_t\sigma_{t-1})| \leq \mathbb{E}(\sigma_t^2)$, strictly with time-varying volatility. 

Consider three additional moments: $\mathbb{E}(\tilde{y}_t^2\tilde{y}_{t-2}^2) = \mathbb{E}(\sigma_t^2\sigma_{t-2}^2)(1+\vartheta^2)^2$,  $\mathbb{E}(\tilde{y}_t^2\tilde{y}_{t-4}^2) = \mathbb{E}(\sigma_t^2\sigma_{t-4}^2)(1+\vartheta^2)^2$, and $\mathbb{E}(\tilde{y}_t^2\tilde{y}_{t-2}^2 \tilde{y}_{t-4}^2) = \mathbb{E}(\sigma_t^2\sigma_{t-2}^2\sigma_{t-4}^2)(1+\vartheta^2)^3,$
the main idea here is to lag twice each time to only measure dependence in $\sigma_t^2$, lagging once would pick-up autocorrelations due to the MA(1) component. Let $\overline{\sigma}^2 = \mathbb{E}(\sigma_t^2)$, we have:
$\mathbb{E}(\sigma_t^2) = \frac{\mu_\sigma + \kappa_\sigma}{1-\rho_\sigma}$, $\mathbb{E}( [\sigma_t^2-\bar{\sigma}^2][\sigma_{t-2}^2-\bar{\sigma}^2] ) = \rho_\sigma^2 \frac{\kappa_\sigma^2 \text{var}(u_t)}{1-\rho_\sigma^2}$, and $\mathbb{E}( [\sigma_t^2-\bar{\sigma}^2][\sigma_{t-4}^2-\bar{\sigma}^2] ) = \rho_\sigma^4 \frac{\kappa_\sigma^2 \text{var}(u_t)}{1-\rho_\sigma^2}.$
Taking a ratio, we can identify $\rho_\sigma \geq 0$ by assumption:
$\frac{ \mathbb{E}(\tilde{y}_t^2 \tilde{y}_{t-2}^2) - \mathbb{E}(\tilde{y}_t^2)^2 }{ \mathbb{E}(\tilde{y}_t^2 \tilde{y}_{t-4}^2) - \mathbb{E}(\tilde{y}_t^2)^2 } = \frac{ \mathbb{E}(\sigma_t^2 \sigma_{t-2}^2) - \mathbb{E}(\sigma_t^2)^2 }{ \mathbb{E}(\sigma_t^2 \sigma_{t-4}^2) - \mathbb{E}(\sigma_t^2)^2 } = \rho_\sigma^2.$
We will assume $\rho_\sigma > 0$ in the following. Similarly, using moments of $\tilde{y}_t$ can compute:
$\frac{\mathbb{E}(\sigma_t^2)^2}{\mathbb{E}(\sigma_t^2 \sigma_{t-2}^2) - \mathbb{E}(\sigma_t^2)^2} = \frac{(\mu_\sigma + \kappa_\sigma)^2}{\kappa_\sigma^2} \frac{1-\rho_\sigma^2}{(1-\rho_\sigma)^2} \rho_\sigma^2 \text{var}(e_{2,t}),$
since $f_2$ is known, this identifies the ratio $(\kappa_\sigma + \mu_\sigma)/\kappa_\sigma$ since the indivial terms are non-negative. Now:
$\mathbb{E}(\tilde{y}_t^2) = \frac{\mu_\sigma + \kappa_\sigma}{\kappa_\sigma(1-\rho_\sigma)}\kappa_\sigma(1+\vartheta^2),$
identifies the product $\kappa_\sigma(1+\vartheta^2)$. The moment $\mathbb{E}(\tilde{y}_t \tilde{y}_{t-1})$ does not have a closed-form expression but can be approximated by expanding $\sqrt{\sigma_t}$ around the mean $\overline{\sigma}^2 = (\mu_\sigma + \kappa_\sigma)/(1-\rho_\sigma)$:
$\mathbb{E}(\sigma_t\sigma_{t-1}) \simeq \mathbb{E} \left( [\bar{\sigma} + \frac{1}{2\bar{\sigma}}(\sigma_t^2-\bar{\sigma}^2)][\bar{\sigma} + \frac{1}{2\bar{\sigma}}(\sigma_{t-1}^2-\bar{\sigma}^2)] \right) = \frac{1}{4\bar{\sigma}^2} \mathbb{E}( [\sigma_t^2-\bar{\sigma}^2][\sigma_{t-1}^2-\bar{\sigma}^2] ).$
The coefficients $\kappa_\sigma,\vartheta$ are then separately identified using the system of equation:
$\mathbb{E}(y_ty_{t-1}) = \vartheta\frac{1}{4\bar{\sigma}^2} \mathbb{E}( [\sigma_t^2-\bar{\sigma}^2][\sigma_{t-1}^2-\bar{\sigma}^2] )$,  
  $\mathbb{E}(y_t^2) = \bar{\sigma}^2 (1+\vartheta^2)$, and 
  $\mathbb{E}(y_t^2y_{t-2}^2) - [\mathbb{E}(y_t^2)]^2 = \rho_\sigma (1+\vartheta^2)\mathbb{E}( [\sigma_t^2-\bar{\sigma}^2][\sigma_{t-1}^2-\bar{\sigma}^2] )$,
using the same approach as for identifying the parameters of an MA(1) model with time-invariant volatility. This implies that $\underline{L} = 5$ lags are sufficient to identify $\theta = (\mu_y,\rho_y,\vartheta_y,\mu_\sigma,\rho_\sigma,\kappa_\sigma)$.
If the unknown distribution $f$ has sub-exponential tails, then its moment generating function is analytic on some interval and the distribution is determined by its moments. The idea is to solve for the moments of $e_{1,t}$ recursively from moments of $y_t$. We already assume that $\mathbb{E}(e_{1,t})=0,\mathbb{E}(e_{1,t}^2)=1$. The third moment $\mathbb{E}(\tilde{y}_t^3) = \mathbb{E}(e_{1,t}^3) \mathbb{E}(\sigma^3)(1+\vartheta^3)$, where the last two terms can be computed from knowledge of $\theta$. Using the Binomial Theorem: $\mathbb{E}(\tilde{y}_t^k) = \mathbb{E}(\sigma_t^k) \sum_{j=0}^k C_{k-j}^j \mathbb{E}(e_{1,t}^{k-j})\mathbb{E}(e_{1,t}^{j}) \vartheta^j$. With $k=3$, this pins down the third moments, for $k=4$ the only unknown is the fourth moment, etc. Hence, once $\theta$ is known $(\mathbb{E}(\tilde{y}_t^3),\dots,\mathbb{E}(\tilde{y}_t^k))$ identifies $(\mathbb{E}(e_{1,t}^3),\dots,\mathbb{E}(e_{1,t}^k))$ for any $k \geq 3$. Since $f$ is determined by its moments, it uniquely determines the distribution itself so that $(\theta,f)$ is jointly identified. With ergodicity, this implies $\lim_{n\to\infty} \mathbb{E}(\hat\psi_n(\tau)-\hat\psi_n^s(\tau,\beta))=0$, $\forall\tau$ if, and only if, $\beta = \beta_0$.
\paragraph{Data Generating Process:} Condition y(i): $\|g_{\text{obs}}(y_1,\beta_1,\sigma)-g_{\text{obs}}(y_2,\beta_1,\sigma)\| = |\rho_y| \|y_1-y_2\| \leq \bar{\rho}_y \|y_1-y_2\|$, which implies the strict contraction property if $|\rho_y| \leq \bar{\rho}_y < 1$. For condition y(ii), $\|g_{\text{obs}}(y_1,\mu_1,\rho_1,\vartheta_1,\sigma)-g_{\text{obs}}(y_1,\mu_2,\rho_2,\vartheta_2,\sigma)\| \leq |\mu_1-\mu_2| + |\rho_1-\rho_2|\times  |y_{1}| + \sigma |\vartheta_1-\vartheta_2| \times |e_1|$ which satisfies the desired bound if $|y_{t-1}|$, $\sigma_t$, and $|e_{t-1}|$  have bounded second moments. This is implied by restrictions on the parameters $\theta$ and the distribution $f$. For condition y(iii), note that the $\sqrt{\cdot}$ function is H\"older continuous with exponent $1/2$ so that $\|g_{\text{obs}}(y_1,\beta,\sigma_1)-g_{\text{obs}}(y_1,\beta_1,\sigma)\| \leq |e_t + \vartheta e_{t-1}| \times \sqrt{|\sigma_1-\sigma_2|}$, and $\mathbb{E}(|e_t + \vartheta e_{t-1}|^2) \leq 3 (1+\overline{\vartheta}^2)$ if $|\vartheta| \leq \overline{\vartheta}$ and $\mathbb{E}(e_t^2)=1$. Hence, the assumptions on the DGP are satisfied.
\subsection{Additional Results for the Second Application} \label{apx:second_application}
Table \ref{tab:CMT} below reports estimates for $1/\tau,1/\gamma$ instead of $\tau,\gamma$ in Table \ref{tab:estimates}. CIs are reported for $\tau$, $\gamma$ by transforming $[1/\hat\tau_n \pm 1.96 \text{se}(1/\hat\tau_n)]$. 
\begin{table}[ht] \setlength\tabcolsep{4.5pt} \renewcommand{\arraystretch}{0.935}
  \begin{center} \caption{Estimates, Standard Errors, Confidence Intervals without the Delta-Method} \label{tab:CMT}
  \begin{tabular}{l|cc|c||cc|c} \hline \hline
          & $1/\hat\tau_n$ & $\text{se}(1/\hat\tau_n)$ & 95\% CI for $\tau$ & $1/\hat\gamma_n$ & $\text{se}(1/\hat\gamma_n)$ & 95\% CI for  $\gamma$\\ \hline
    $k=1$ & 0.001 & 0.004 & $[128.35,\,+\infty)$ & 0.029 & 0.013 & $[18.52,\,266.65]$\\
    $k=2$ & 0.020 & 0.008 & $[28.81,\,204.99]$ & 0.050 & 0.012 & $[13.61,\,38.78]$\\
    $k=3$ & 0.018 & 0.006 & $[32.95,\,158.60]$ & 0.079 & 0.021 & $[8.43,\,26.19]$\\
    $k=4$ & 0.019 & 0.005 & $[34.17,\,107.94]$ & 0.096 & 0.025 & $[6.97,\,21.11]$\\
    $k=5$ & 0.015 & 0.005 & $[38.77,\,245.53]$ & 0.084 & 0.022 & $[7.90,\,24.69]$\\ \hline \hline
  \end{tabular} \vspace*{-0.9cm} \end{center} 
\end{table} 
\section{Additional Results}  \label{apx:add_res}
\subsection{Convergence rate in the MA(1) model}
The following derives the rate of convergence for the MA(1) process: $y_t = e_t + \vartheta e_{t-1}$, $e_t \overset{iid}{\sim} f$, first when $S=+\infty$. Here $\beta = (\vartheta,f) \in [-1,1] \times \mathcal{F}$. Take $L \geq 1$, then the joint distribution $\mathbf{y}_t = (y_t,y_{t-1})$ uniquely identifies $\beta$. Let $h(\tau,e,\vartheta) = e^{ i \tau_1 e_1 + i \vartheta \tau_2 e_2 + i \tau_2 e_2 + i \vartheta \tau_2 e_3 } $. The CF of $\mathbf{y}_t$ is: $\psi(\tau;\beta) = \int h(\tau,e,\vartheta)f(e_1)f(e_2)f(e_3) d e_1 de_2 de_3,$ for  $L=1$ where $\tau = (\tau_1,\tau_2)$. Let $\beta_k = (\vartheta_0,f_k)$ and $\hat\beta_n$ be an exact minimizer of $Q_n$, then by triangular inequalities in $\mathbb{L}^2(\pi)$:
{\small \begin{align*}
  &\Big( \int |\psi(\tau;\hat\beta_n)-\psi(\tau;\beta_0)|^2 \pi(\tau)d\tau \Big)^{1/2} - \Big( \int |\hat \psi_n(\tau)-\psi(\tau;\beta_0)|^2 \pi(\tau)d\tau \Big)^{1/2}\\
  &\leq \sqrt{Q_n(\hat\beta_n)} \leq \sqrt{Q_n(\beta_k)}
  \leq \Big(\int |\psi(\tau;\beta_k)-\psi(\tau;\beta_0)|^2 \pi(\tau)d\tau \Big)^{1/2} + \Big( \int |\hat \psi_n(\tau)-\psi(\tau;\beta_0)|^2 \pi(\tau)d\tau \Big)^{1/2}.
\end{align*}}
The last term is $O_p(n^{-1/2})$ plus $(\int |\psi(\tau;\beta_k)-\psi(\tau;\beta_0)|^2 \pi(\tau)d\tau )^{1/2} \leq (L+1)\|f_k-f_0\|_{TV}$ because the exponential has modulus $1$ and the density $f$ appears $L+1$ times in the CF. This is related to the bias accumulation discussed in the main text. From this we deduce the convergence rate under the distance implied by the CF:
\begin{align*}
  \Big( \int |\psi(\tau;\hat\beta_n)-\psi(\tau;\beta_0)|^2 \pi(\tau)d\tau \Big)^{1/2} &\leq 2 \Big( \int |\hat \psi_n(\tau)-\psi(\tau;\beta_0)|^2 \pi(\tau)d\tau \Big)^{1/2} + (L+1)\|f_k-f_0\|_{TV},
\end{align*}
which is a $O_p(\max[ n^{-1/2}, \log[k]^{2r/b}k^{-r}])$, since $\|f_k-f_0\|_{TV} = O(\log[k]^{2r/b}k^{-r})$ under the smoothness and tails assumptions. Because here $S=+\infty$, we can use $k \log[k]^{-2/b} \asymp n^{-1/2r}$ which gives:
$\int |\psi(\tau;\hat\beta_n)-\psi(\tau;\beta_0)|^2 \pi(\tau)d\tau \Big)^{1/2} = O_p(n^{-1/2}),$ in line with Corollary \ref{rmk:full_simu}. For $r=2$, this implies $k \asymp n^{-1/4}$, up to log-terms.
Asymptotically, $( \int |\psi(\tau;\hat\beta_n)-\psi(\tau;\beta_0)|^2 \pi(\tau)d\tau )^{1/2} \asymp \|\hat\beta_n - \beta_0\|_{\text{weak}}$ which implies the convergence rate in weak norm. It involves the derivative $\psi_\beta (\tau,f)[v]$, i.e. $\psi_f (\tau,\beta)[v] = \int h(\tau,e,\vartheta)\{ v(e_1)f(e_2)f(e_3) + f(e_1)v(e_2)f(e_3) + f(e_1)f(e_2)v(e_3) \}de_1de_2de_3$ and $\psi_\vartheta (\tau,\beta) = \int [\tau_1 e_2 + \tau_2 e_3] h(\tau,e,\vartheta) f(e_1)f(e_2)f(e_3) de_1de_2de_3$, for $L=1$. The local measure of ill-posedness $\tau_n$ is not closed-form, making the rate in stronger norm intractable. For $S < +\infty$, the term $\sup_{\beta \in \paramspace_{k(n)}}( \int |\psi(\tau;\beta)-\hat\psi_n^S(\tau;\hat\beta_n)|^2 \pi(\tau)d\tau )^{1/2} = O_p([k(n)\log[k(n)]]^2/\sqrt{nS})$ also affects  the rate of convergence. Here geometric ergodicity automatically holds; an MA(1) being m-dependent regardless of the MA coefficient.

\subsection{Sieve Long-Run Variance} \label{sec:sieveLRR}
The following derives the formula for the sieve long-run variance $\sigma_n^{\star 2}$. 
For brevity of notation, let $Z_t(\tau) = \hat\psi_t^S(\tau,\beta_0) - \hat\psi_t(\tau)$ and $Z_n(\tau) = \frac{1}{n} \sum_{t} Z_t(\tau)$. Let:
$S_t^\star =  \frac{1}{2} \int \{ \psi_\beta(\tau,v_n^\star)\overline{ Z_t(\tau) } + \overline{ \psi_\beta(\tau,v_n^\star) }Z_t(\tau)\}\pi(\tau)d\tau$,
the sieve score is $S_n^\star = \frac{1}{n} \sum_t S_t^\star$, and the sieve long-run variance is: $\sigma_n^{\star 2} = n \mathbb{E}(S_n^{\star 2}) = \mathbb{E}(S_t^{\star 2}) + 2 \sum_{j = 1}^{n-1} \frac{n-j}{n} \mathbb{E}(S_t^{\star}S_{t-j}^{\star})$.
For any $j \geq 0$, we have:
{ \begin{align*}
  &\mathbb{E}(S_t^\star S_{t-j}^\star ) = \frac{1}{4} \int  \Big\{ \psi_\beta(\tau,v_n^\star) \mathbb{E}[ \overline{Z_t(\tau_1)} \overline{Z_{t-j}(\tau_2)} ] \psi_\beta(\tau_2,v_n^\star) + \psi_\beta(\tau,v_n^\star) \mathbb{E}[ \overline{Z_t(\tau_1)} Z_{t-j}(\tau_2) ] \overline{ \psi_\beta(\tau_2,v_n^\star) } \\
  &+ \overline{\psi_\beta(\tau,v_n^\star)} \mathbb{E}[ Z_t(\tau_1) \overline{Z_{t-j}(\tau_2)} ]  \psi_\beta(\tau_2,v_n^\star) + \overline{\psi_\beta(\tau,v_n^\star)} \mathbb{E}[ Z_t(\tau_1) Z_{t-j}(\tau_2) ]  \overline{\psi_\beta(\tau_2,v_n^\star)}  \Big\} \pi(\tau_1)\pi(\tau_2) d\tau_1 d\tau_2.
\end{align*}}
Let $K_j: \mathbb{L}^2(\pi) \to \mathbb{L}^2(\pi)$ be a linear operator such that:
$K_j f(\tau_1) = \frac{1}{2} \int \Big\{ \mathbb{E}[ \overline{Z_t(\tau_1)} \overline{Z_{t-j}(\tau_2)} ] f(\tau_2) + \mathbb{E}[ \overline{Z_t(\tau_1)} Z_{t-j}(\tau_2) ] \overline{f(\tau_2)} \Big\}\pi(\tau_2)d\tau_2,$
with the associated inner-product in $\mathbb{L}^2(\pi)$: $\langle f_1,f_2 \rangle_{\pi} = \frac{1}{2}\int \{ f_1(\tau)\overline{f_2(\tau)} + \overline{f_1(\tau)}f_2(\tau) \} \pi(\tau)d\tau$.\footnote{ Notice that $\langle v_1,v_2\rangle = 1/2 \int \{ \psi_\beta(\tau,v_1) \overline{\psi_\beta(\tau,v_2)} \} + \overline{\psi_\beta(\tau,v_1)} \psi_\beta(\tau,v_2)  \}\pi(\tau)d\tau$ is also $\langle \psi_\beta(\cdot,v_1), \psi_\beta(\cdot,v_2)\rangle_{\pi}$. } Compactly re-write the autocovariance: $\mathbb{E}(S_t^\star S_{t-j}^\star ) = \langle \psi_\beta(\cdot,v_n^\star), K_j \psi_\beta(\cdot,v_n^\star) \rangle_{\pi}.$ Then, by linearity: $\sigma_n^{\star 2} = \langle \psi_\beta(\cdot,v_n^\star), K_n \psi_\beta(\cdot,v_n^\star) \rangle_{\pi}$,
where $K_n = K_0 + 2 \sum_{j=1}^{n-1} \frac{n-j}{n} K_j$ is the long-run variance operator. Their sample counterparts are: $\hat{\psi}_\beta(\tau,v) = d_\beta \hat\psi_n^S(\tau,\hat\beta_n)[v]$, $ \langle v_1,v_2 \rangle_n = \frac{1}{2} \int \{ \hat{\psi}_\beta(\tau,v_1) \overline{ \hat{\psi}_\beta(\tau,v_2) } + \overline{ \hat{\psi}_\beta(\tau,v_1) } \hat{\psi}_\beta(\tau,v_2)   \} \pi(\tau)d\tau$, $\hat{v}_n^\star$ such that $\langle \hat{v}_n^\star, v \rangle_n = d_\beta \phi(\hat\beta_n)[v]$ for any $v$. Let $\hat{Z}_t(\tau) = \hat\psi_t^S(\tau,\hat\beta_n) - \hat\psi_t(\tau)$, $\hat{S}_t^\star =  \frac{1}{2} \int \{ \hat{\psi}_\beta(\tau,\hat{v}_n^\star)\overline{ \hat{Z}_t(\tau) } + \overline{ \hat{\psi}_\beta(\tau,\hat{v}_n^\star) }\hat{Z}_t(\tau)\}\pi(\tau)d\tau$, and $\hat{S}_n^\star = \frac{1}{n} \sum_t \hat{S}_t^\star$. Using an estimate $\hat{K}_n$ of $K_n$, we have: $\|\hat{v}_{n,sd}^\star\|^2 = \hat{\sigma}_n^{\star 2} = \langle \hat{\psi}_\beta(\cdot,\hat{v}_n^\star), \hat{K}_n \hat{\psi}_\beta(\cdot,\hat{v}_n^\star) \rangle_{\pi} = \langle \hat{v}_n^\star, \hat{v}_n^\star \rangle_{n,\hat{K}_n}$. Now, to estimate the long-run variance operator $K_n$, take $j \geq 0$ and let $\hat{K}_j$ be such that:
$\hat{K}_j f(\tau_1) = \frac{1}{2} \int \Big\{ \frac{1}{n} \big[\sum_{t=j+1}^n \overline{\hat{Z}_t(\tau_1)} \overline{\hat{Z}_{t-j}(\tau_2)}\big] f(\tau_2) + \frac{1}{n} \big[\sum_{t=j+1}^n \overline{\hat{Z}_t(\tau_1)} \hat{Z}_{t-j}(\tau_2)\big] \overline{f(\tau_2)} \Big\}\pi(\tau_2)d\tau_2$; then $\hat{K}_n = \hat{K}_0 + 2 \sum_{j=1}^{n-1} \omega(j/T_n) \hat{K}_j$, where $\omega$ and $T_n$ are the HAC kernel and bandwidth.

\begin{assumption} \label{ass:LRRV} Suppose i. $\sup_{\beta \in \mathcal{N}_{osn}} \sup_{ v \in \overline{V}^1_{k(n)} } | d_\beta \phi(\beta)[v] - d_\beta \phi(\beta_0)[v] | = o(1)$, ii. for each $k(n)$, any $\beta \in \mathcal{N}_{osn}$, and any $v \in \overline{V}^1_{k(n)}$, $\hat{\psi}_{\beta}(\cdot,v) \in \mathbb{L}^2(\pi)$, $\sup_{v_1,v_2 \in \overline{V}^1_{k(n)}}|\langle v_1,v_2 \rangle_n -\langle v_1,v_2 \rangle | = o_p(1)$, iii. $\sup_{v \in \overline{V}_{k(n)}^1} | \langle v,v \rangle_{n,K_n} - \langle v,v \rangle_{K_n} | = o_p(1)$, iv. $\|\hat{K}_n - K_n\|_{op} = o_p(1)$.
\end{assumption}
Where $\|\cdot\|_{op}$ is the operator norm  in $(\mathbb{L}^2(\pi),\langle \cdot,\cdot \rangle_{\pi})$. Assumption \ref{ass:LRRV} i-iii is based on Assumption 4.1 in \citet{Chen2015a}. Given Assumption \ref{ass:sid} iii, Proposition 3.3 in \citet{Carrasco2007} imply Assumption \ref{ass:LRRV} iv holds under Assumption \ref{ass:HAC} below.

\begin{assumption} \label{ass:HAC} Suppose i. $\omega : \mathbb{R} \to [0,1]$, $\omega(0)=1$, $\omega(-x)=\omega(x)$, $\forall x \in \mathbb{R}$, $\omega \in \mathbb{L}^2(\mathbb{R})$, $\omega$ is continuous at $0$ and all, but finitely many, values of $x$; ii. $T_n^{2 \nu + 1}/n \to \gamma \in (0,\infty)$ for some $\nu$ for which $\|\omega^\nu\|<\infty$ and $\|f_{Y}^{\nu}\| < \infty$, $\omega^\nu$ and $f_{Z}^{\nu}$ are the $\nu$-th derivative of $\omega$ and $f_Y$, the spectral density of $(\mathbf{y}_t,\mathbf{y}_t^s)$ at $0$.
\end{assumption}

\begin{proposition} \label{prop:sieveLRR}
  Suppose Assumption \ref{ass:LRRV} holds, then $\big| \hat{\sigma}_n^\star/\sigma_n^\star-1 \big| = o_p(1)$.
\end{proposition}
Proposition \ref{prop:sieveLRR} follows from Theorem 4.2 in \citet{Chen2015a}, where now Step 2A in their proof \citep[p9]{Chen2015aS} requires $\|\hat{K}_n-K_n\|_{op}=o_p(1)$ as in Assumption \ref{ass:LRRV} iv. 
The formula used in the main text is easier to implement, but equivalent. For each $j \geq 0$: $\int \text{real}\{ \psi_\beta(\tau_1,v_n^\star) \mathbb{E}[ \overline{Z_t(\tau_1)}\text{real}[Z_{t-j}(\tau_2) \overline{ \psi_\beta(\tau_2,v_n^\star) }] ] \}\pi(\tau_1)\pi(\tau_2)d\tau_1 d\tau_2 = \langle \psi_\beta(\cdot,v_n^\star), K_j \psi_\beta(\cdot,v_n^\star) \rangle_{\pi}$. Because $\mathbb{E}$, $\int$ and $\text{real}$ are linear operators, they arrange into: \begin{align*} &\langle \psi_\beta(\cdot,v_n^\star), K_j \psi_\beta(\cdot,v_n^\star) \rangle_{\pi}\\ &= \mathbb{E}\Big\{ \left(\int \text{real}[\psi_\beta(\tau_1,v_n^\star) \overline{Z_t(\tau_1)}] \pi(\tau_1)d\tau_1 \right) \left(\int \text{real}[\psi_\beta(\tau_2,v_n^\star) \overline{Z_{t-j}(\tau_2)}] \pi(\tau_2)d\tau_2 \right) \Big\}. \end{align*}
Then replace $\text{real}[\psi_\beta(\tau_1,v_n^\star) \overline{Z_t(\tau_1)}] = \text{real}[\psi_\beta(\tau_1,v_n^\star)] \text{real}[Z_t(\tau_1)] + \text{im}[\psi_\beta(\tau_1,v_n^\star)] \text{im}[Z_t(\tau_1)]$. Next, let $\varphi = (\theta,\omega,\mu,\sigma)$ denote the parameter $\beta$ in the sieve basis. For any $v$, $v^\prime d_\varphi \phi(\beta_0) = \langle v, v_n^\star \rangle = v^\prime \text{real}[ \int \psi_{\varphi^\prime}(\tau,\beta_0) \overline{\psi_{\varphi^\prime}(\tau,\beta_0)}\pi(\tau)d\tau] v_n^\star$ so $v_n^\star = \text{real}[ \int \psi_{\varphi^\prime}(\tau,\beta_0) \overline{\psi_{\varphi^\prime}(\tau,\beta_0)}\pi(\tau)d\tau]^{-1} d_\varphi \phi(\beta_0)$. Now, substitude $v_n^\star$ into $\langle \psi_\beta(\cdot,v_n^\star), K_n \psi_\beta(\cdot,v_n^\star) \rangle_{\pi}$ to get the sandwich formula. The same derivations applied to the sample quantities yield the formula in the main text. 

\newpage

\section{Additional Monte-Carlo Results} \label{sec:MC_Extra}
\subsection{Main Examples: $n=100$} \label{sec:MC_Extra_main}

\begin{table}[h] \setlength\tabcolsep{4.5pt} \renewcommand{\arraystretch}{0.935}
  \begin{center} \caption{Bias, Standard Deviation and Size} \label{tab:resMC100} {
    \small
    \begin{tabular}{cl|aaaa|bb|b} \hline \hline
      & & \multicolumn{4}{c}{Sieve-SMM} & \multicolumn{2}{c}{Bayesian} & GMM\\ \hline
      & $k$ & \multicolumn{2}{c}{2} & \multicolumn{2}{c|}{3} & 2 & 3 & -\\
      & $S$ & \mc{1}{1} & \mc{1}{5} & \mc{1}{1} & \multicolumn{1}{c|}{5} & - & - & -\\ \hline
      \multirow{3}{*}{AR(1)}       & bias  & -0.023 & -0.025 & -0.030 & -0.032 & -0.024 & -0.023 & -0.031 \\ 
      & std & \,\,0.119 & \,\,0.092 & \,\,0.112 & \,\,0.090 & \,\,0.073 & \,\,0.072 & \,\,0.083 \\
      & size  & \,\,0.047 & \,\,0.037 & \,\,0.033 & \,\,0.027 & \,\,0.055 & \,\,0.053 & \,\,0.016 \\ \hline \hline
    \end{tabular} }\\
  {\footnotesize \textit{Note: $n=100$, $1000$ Monte-Carlo replications.}}
  \end{center}
\end{table}

\begin{table}[ht] \setlength\tabcolsep{4.5pt} \renewcommand{\arraystretch}{0.935}
   \caption{Bias, Standard Deviations and Size for the SV Model (\ref{eq:SV2})}
  \centering \small
  \begin{tabular}{lr|ccccc|ccccc}
    \hline\hline
    & & \multicolumn{5}{c|}{$k=2$} & \multicolumn{5}{c}{$k=4$} \\
  S & & $\mu_y$ & $\rho_y$ & $\vartheta_y$ & $\rho_\sigma$ & $\kappa_\sigma$ & $\mu_y$ & $\rho_y$ & $\vartheta_y$ & $\rho_\sigma$ & $\kappa_\sigma$ \\ 
    \hline 
  \multirow{3}{*}{1} & bias & 0.000 & -0.003 & 0.005 & -0.167 & -0.092 & -0.000 & -0.001 & 0.010 & -0.078 & -0.105 \\ 
   & std & 0.014 & 0.014 & 0.082 & 0.276 & 0.280 & 0.012 & 0.014 & 0.066 & 0.182 & 0.201 \\ 
  & size & 0.315 & 0.200 & 0.100 & 0.100 & 0.155 & 0.215 & 0.140 & 0.060 & 0.020 & 0.085 \\ \hline
  \multirow{3}{*}{5} & bias & 0.001 & -0.006 & 0.020 & -0.077 & -0.133 & 0.000 & -0.006 & 0.011 & -0.051 & -0.083 \\ 
  & std & 0.009 & 0.012 & 0.062 & 0.216 & 0.210 & 0.008 & 0.012 & 0.053 & 0.126 & 0.138 \\ 
  & size & 0.535 & 0.170 & 0.090 & 0.005 & 0.055 & 0.335 & 0.125 & 0.050 & 0.000 & 0.005 \\ \hline
  \multirow{3}{*}{20} & bias & 0.000 & -0.005 & 0.012 & -0.041 & -0.116 & -0.000 & -0.005 & 0.008 & -0.020 & -0.066 \\ 
  & std & 0.008 & 0.011 & 0.060 & 0.177 & 0.193 & 0.007 & 0.011 & 0.056 & 0.071 & 0.113 \\ 
  & size & 0.505 & 0.155 & 0.045 & 0.000 & 0.000 & 0.425 & 0.105 & 0.050 & 0.000 & 0.000 \\   
     \hline \hline
  \end{tabular}\\{\footnotesize \textit{Note: $n=750$, $200$ Monte-Carlo replications. $(\mu_y,\rho_y,\vartheta_y,\rho_\sigma,\kappa_\sigma) = (0.025,0.98,-0.73,0.7,0.6)$, $\kappa_\sigma$ is scaled by $10^4$ for readability. For reference - the bias, std and size for GMM estimates of $\mu_y$ are $0.001$, $0.008$, $0.455$ and for $\rho_y$ $-0.006$, $0.027$, $0.06$. }}
  \end{table}

\subsection{Sensitivity to the Estimation Inputs} \label{sec:MC_Extra_sens}

\begin{table}[H] \setlength\tabcolsep{4.5pt} \renewcommand{\arraystretch}{0.935}
  \begin{center} \caption{Sensitivity to Estimation Inputs in the AR(1) Model (\ref{eq:ar_1})} \label{tab:tuning} {
    \small
    \begin{tabular}{b|a|bb|bb|bb} \hline \hline
     & Baseline & $L=2$ & $L=6$ & $\underline{\sigma}_{k(n)}=1.4$ & $\underline{\sigma}_{k(n)}=2.2$ & $B=100$ & $B=250$ \\ \hline
     $\hat\rho_n$ & 0.630 & 0.638 & 0.623 & 0.607 & 0.623 & 0.629 & 0.627\\ \hline \hline
    \end{tabular} }\\
    {\footnotesize \textit{Note: plots based on one simulated sample, $n=200$, $k=3$. Baseline: $L=4$, $\underline{\sigma}_{k(n)}=1.8$, $B=500$}}
  \end{center}
\end{table}

\begin{figure}[h]
  \begin{center} \caption{Sensitivity to Estimation Inputs in the AR(1) Model (\ref{eq:ar_1})}\label{fig:tuning}
  \includegraphics[scale=0.55]{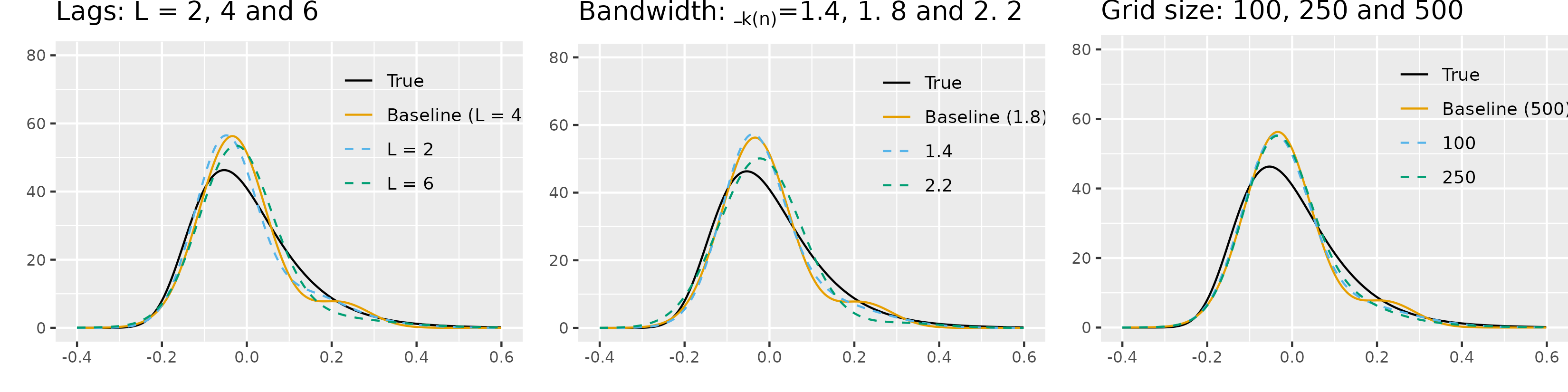}\\
  {\footnotesize \textit{Note: plots based on one simulated sample, $n=200$, $k=3$.}} \end{center}
\end{figure}

\subsection{Asset Pricing in a Stationnary Production Economy} \label{sec:MC_Extra_AP}

\paragraph{Model:} The model is a simplified version of \citet{Ruge-Murcia2016} with CRRA preferences. Log-productivity and inflation follow AR(1) processes instead of a VAR(1). The main equations are summarized below. Utility is CRRA, using the notation of the empirical application:
$U_t = \frac{C_t^{1-\gamma}}{1-\gamma} + \beta \mathbb{E}_t(U_{t+1}).$ 
Productivity $Z_t$ evolves according to:
$\log Z_{t+1} = \rho_z \log Z_{t} + e_{1,t+1},$
where $e_{1,t+1} \overset{iid}{\sim} f_1$ has mean zero. Inflation $\pi_{t+1} = P_{t+1}/P_t$ evolves according to:
$\log \pi_{t+1} = \log\bar{\pi} + \rho_\pi (\log \pi_{t}- \log\bar{\pi}) + e_{2,t+1},$
where $e_{2,t+1} \overset{iid}{\sim} f_2$ has mean zero. Production is Cobb-Douglas $Y_t = Z_t K_t^\alpha$, capital evolves according to $K_{t+1} = (1-\delta)K_t + I_t$ with quadratic adjustment costs $\Psi_t = \frac{\psi}{2}(I_t/K_t-\delta)^2K_t$. The bond pricing equation is:  $Q_t^\ell = \beta\mathbb{E}_t \left[  \left( \frac{C_{t+1}}{C_t} \right)^{-\gamma} \frac{Q_{t+1}^{\ell-1}}{\pi_{t+1}} \right],$
where $Q_t^0=1$. Only the $3$m yield is computed. Two parameters $(\alpha,\delta)=(0.35,0.025)$ are calibrated as in \citet{Ruge-Murcia2016}. Consumption and investment growth are de-meaned to remove the effect of calibration on the levels. Estimation is performed as in the empirical application but with $S=1$. For reference, estimates with parametric skewed-logistic (\textsc{skl}) shocks are also reported. The settings for the optimizer are the same for all three estimations.  The average time is 20mn for the parametric estimates, 31mn for $k=3$ and 49mn for $k=5$ in a 12-core cluster environment. Table \ref{tab:AP} reports average, median estimates and their standard deviations. \textsc{skl} estimates for $\gamma$ and $\phi$ are very skewed upwards, mixtures estimates appear to be more robust and closer to their large sample approximation. Figure \ref{fig:AP} reports the average and interquantile range of the density estimates $\hat f_{1,n},\hat f_{2,n}$. 
   
\begin{table}[h] \setlength\tabcolsep{4.5pt} \renewcommand{\arraystretch}{0.935}
  \begin{center} \caption{Asset Pricing Model: Estimates of $\theta$} \label{tab:AP} {
    \small
    \begin{tabular}{ll|bbbbbb}
      \hline \hline
     & & $\beta$ & $\gamma$ & $\phi$ & $\rho_z$ & $\rho_\pi$ & $\log\overline{\pi}$\\ \hline
     & true & 0.99 & 4 & 20 & 0.95 & 0.9 & 0.009 \\ \hline
     \multirow{4}{*}{\textsc{skl}} & mean & 0.990 & 4.893 & 20.441 & 0.946 & 0.894 & 0.009 \\ 
      & median & 0.990 & 4.174 & 20.154 & 0.951 & 0.897 & 0.009\\
      & std & 0.003 & 3.470 & 6.034 & 0.026 & 0.025 & 0.002 \\ 
      & size & 0.036 & 0.216 & 0.076 & 0.152 & 0.536 & 0.200 \\ \hline
      \multirow{4}{*}{$k=3$} & mean & 0.990 & 4.233 & 20.456 & 0.950 & 0.901 & 0.009 \\ 
      & median & 0.990 & 4.100 & 20.198 & 0.950 & 0.899 & 0.009\\
      & std & 0.002 & 1.017 & 2.410 & 0.006 & 0.018 & 0.001 \\ 
      & size & 0.112 & 0.100 & 0.172 & 0.084 & 0.008 & 0.004 \\ \hline
      \multirow{4}{*}{$k=5$} & mean & 0.990 & 4.183 & 20.351 & 0.950 & 0.902 & 0.009 \\ 
      & median & 0.990 & 4.077 & 20.203 & 0.950 & 0.901 & 0.009\\
      & std & 0.002 & 0.696 & 1.540 & 0.004 & 0.013 & 0.001 \\ 
      & size & 0.064 & 0.104 & 0.176 & 0.072 & 0.000 & 0.000 \\  \hline
      \multicolumn{2}{c|}{ub} & 0.9999  & 20 & 60 &  0.99 & 0.99 & 0.02\\
      \multicolumn{2}{c|}{lb} & 0.95 &  0.01 & 0.01 &  0.7 &  0.6 & 0.003\\
      \hline \hline
    \end{tabular}
     }\\
    {\footnotesize \textit{Note: $n=235$, $S=1$, $250$ Monte-Carlo replications.}}
  \end{center}
\end{table} 

\begin{figure}[H]
  \begin{center} \caption{Asset Pricing Model: Estimates of $f_1,f_2$}\label{fig:AP}
  \includegraphics[scale=0.4]{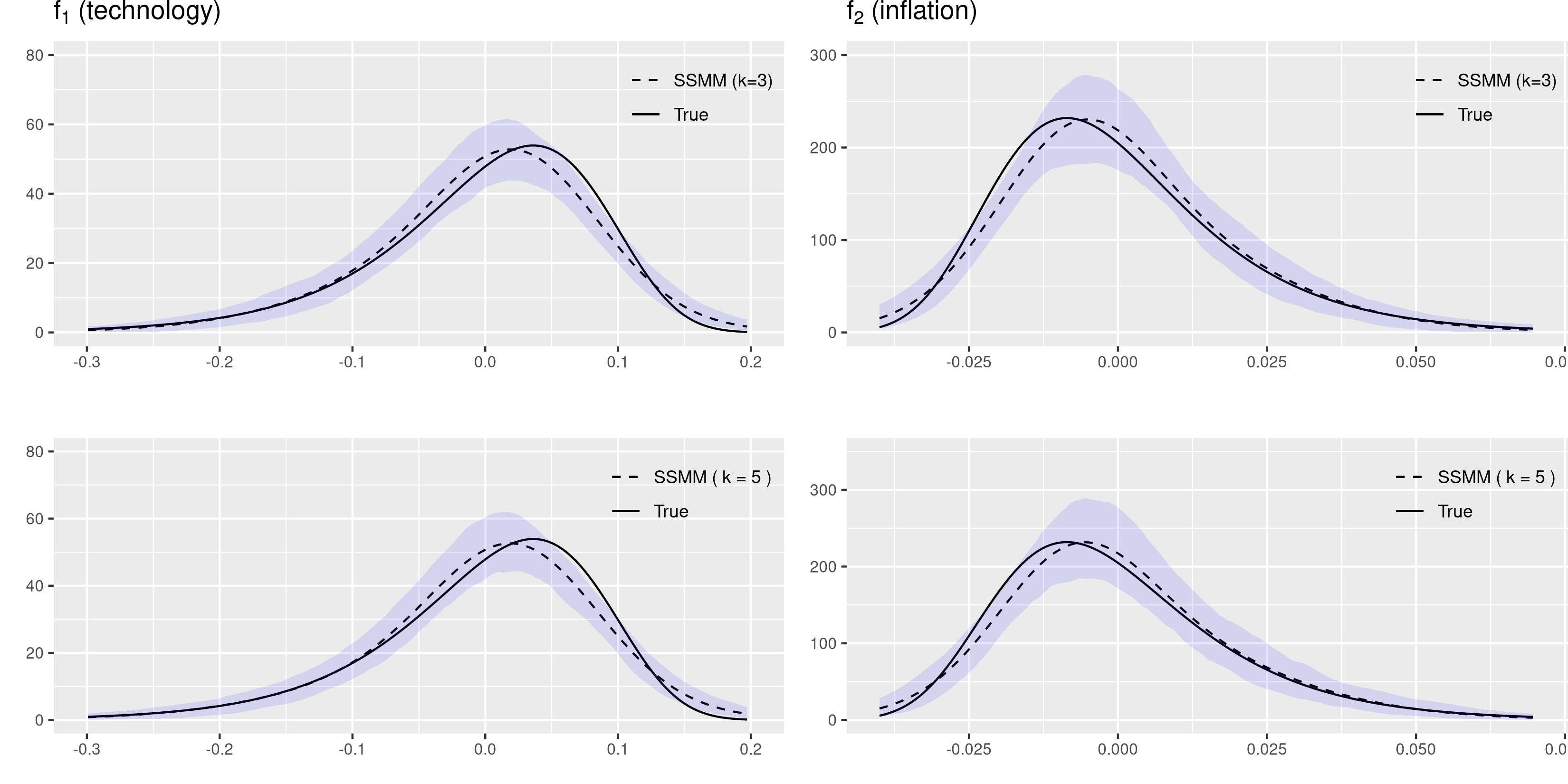}\\
  {\footnotesize \textit{Note: bands = 95\% pointwise interquantile range, $n=235$, $k=3$, $100$ Monte-Carlo replications.}} \end{center}
\end{figure}

\end{appendices}
\end{document}